\documentclass[pra,11pt,onecolumn,superscriptaddress,nofootinbib]{revtex4}
\usepackage{color}
\usepackage{amsmath}
\usepackage{amsfonts}
\usepackage{amsmath}
\usepackage{amsthm}
\usepackage{amssymb}
\usepackage{dsfont}
\usepackage[all]{xy}
\usepackage{graphicx}
\usepackage{relsize}
\usepackage{pgf}
\usepackage{pgf,pgfarrows,pgfnodes}
\usepackage{color}
\usepackage[caption=false]{subfig}
\usepackage{float}
\usepackage{tikz}
\usepackage{pdfpages}
\usepackage{hyperref}
\usepackage{color}
\usepackage{dcolumn}
\usepackage{bm}
\usepackage{verbatim}
\usepackage{amscd}
\usepackage{setspace}
\usepackage{enumerate}
\newcommand{\calD}{\mathcal{D}}

\newcommand{\zz}{\mathbb Z}

\newcommand{\calM}{\mathcal M}
\newcommand{\calG}{\mathcal G}

\newcommand{\calH}{\mathcal H}

\newcommand{\calC}{\mathcal C}

\newcommand{\calZ}{\mathcal Z}
\newcommand{\calK}{\mathcal K}
\newcommand{\calU}{\mathcal U}
\newcommand{\calV}{\mathcal V}
\newcommand{\calO}{\mathcal O}
\newcommand{\calP}{\mathcal P}

\newcommand{\calE}{\mathcal E}

\renewcommand{\phi}{\varphi}

\newcommand{\norm}[1]{\left\lVert#1\right\rVert}

\newcommand{\Tr}{\operatorname {Tr}}

\newcommand{\bra}[1]{\ensuremath{\left\langle#1\right|}}
\newcommand{\ket}[1]{\ensuremath{\left|#1\right\rangle}}

\theoremstyle{plain}
\newtheorem{thm}{Theorem}
\theoremstyle{plain}
\newtheorem{lem}{Lemma}

\theoremstyle{definition}
\newtheorem{defn}{Definition}

\theoremstyle{remark}

\newenvironment{skproof}{%
	\proof}{\endproof}

\definecolor{samcolor1}{rgb} {0,0,1}
\definecolor{samcolor2}{rgb} {0,0,0}
\begin{document}
\title{Symmetry protected topological order at nonzero temperature}
\author{Sam Roberts}%
\affiliation{Centre for Engineered Quantum Systems, School of Physics, The University of Sydney, Sydney, NSW 2006, Australia}%
\author{Beni Yoshida}%
\affiliation{Perimeter Institute for Theoretical Physics, Waterloo, ON N2L 2Y5, Canada}%
\author{Aleksander Kubica}%
\affiliation{Institute for Quantum Information \& Matter, California Institute of Technology, Pasadena CA 91125, USA}%
\author{Stephen D. Bartlett}%
\affiliation{Centre for Engineered Quantum Systems, School of Physics, The University of Sydney, Sydney, NSW 2006, Australia}%

\date{7 August 2017}

\begin{abstract} 
	We address the question of whether symmetry-protected topological (SPT) order can persist at nonzero temperature, with a focus on understanding the thermal stability of several models studied in the theory of quantum computation.  We present three results in this direction.  First, we prove that nontrivial SPT-order protected by a global onsite symmetry cannot persist at nonzero temperature, demonstrating that several quantum computational structures protected by such onsite symmetries are not thermally stable.  Second, we prove that the 3D cluster state model used in the formulation of topological measurement-based quantum computation possesses a nontrivial SPT-ordered thermal phase when protected by a generalized (1-form) symmetry.  The SPT-order in this model is detected by long-range localizable entanglement in the thermal state, which compares with related results characterizing SPT-order at zero temperature in spin chains using localizable entanglement as an order parameter.  Our third result is to demonstrate that the high error tolerance of this 3D cluster state model for quantum computation, even without a protecting symmetry, can be understood as an application of quantum error correction to effectively enforce a 1-form symmetry.
\end{abstract}

\maketitle

\section{Introduction}
Topological phases are not only fascinating from the perspective of fundamental physics but are also well-suited for the design of quantum computers, for two essential reasons.  First, such phases possess topology-dependent ground-state degeneracies, into which quantum information can be encoded and which can manifest themselves through boundary degrees of freedom.  That is, qubits arranged on a spin lattice in a topologically-ordered phase are an instance of a quantum error correcting code:  information is encoded in nonlocal degrees of freedom, offering robustness to local errors that can be detected through the measurement of local syndromes.  Second, these desirable properties are robust against perturbations that act locally on the system~\cite{BravyiHastings}, making them ideal for quantum information processing with faulty devices without the requirement of precise control over all microscopic degrees of freedom.   

Although much of the existing work on the study of topological phases is devoted to studying ground state (zero temperature) properties, identifying systems that can maintain their quantum coherence in equilibrium at some nonzero temperature would be highly desirable for quantum computing applications.  Most of the well-studied exactly-solvable models in two or three dimensions (such as Kitaev's toric code) do not maintain their topological order except at zero temperature~\cite{HastingsThermal,brownreview}.  The full range of phenomena of topological models in three or more dimensions has yet to be fully explored, though, so there is plenty of room for optimism.

A promising new direction in recent years is to add a symmetry to the mix. Symmetries have historically proven to be a powerful tool for understanding the structure and thermal stability of many-body phases of matter, for example, Landau's paradigm of symmetry breaking, the Mermin-Wagner theorem \cite{mermin1966absence}, and Elitzur's theorem\cite{elitzur1975impossibility}. More recently, symmetries have been used to characterise the order in systems away from equilibrium, such as periodically driven (Floquet) systems, where the thermalization time can be long \cite{else2017prethermal,potter2016classification}. Even at zero temperature, a rich set of ordered phases can appear even in trivial models when a symmetry is enforced; such symmetry-protected topological (SPT) phases are described by short-range entangled states that cannot be adiabatically connected to a trivial product state while preserving the symmetry~\cite{chen11a,chen11b,schuch2011classifying}.  Like topological phases, these SPT phases can possess ground-state degeneracies manifested through boundary degrees of freedom, and these degeneracies are robust against local symmetry-respecting perturbations.  With symmetry, new avenues open up.  For example, SPT-nontrivial phases can be identified even in spin chains with only one spatial dimension; nontrivial topological phases require at least 2D.

Nontrivial SPT phases are not likely to be useful for defining good quantum codes, mainly because this would require a very strong assumption about the error model (i.e., that it respects the symmetry).  Nonetheless, SPT phases have found several applications in our understanding of other features of quantum computation.  First, the model for measurement-based quantum computation (MBQC)~\cite{Rau01} can be understood in terms of performing computations on fractionalized edge modes associated with the boundaries of symmetry-protected phases of spin chains~\cite{BBMR,Miy10,RMBB}; a very precise relationship between the computational properties of a spin chain and its SPT-order was developed by Else \textit{et al.}~\cite{Else2012,Delse1} and Miller \textit{et al.} \cite{MillerMiyake15}.  Second, a direct relationship between the set of fault-tolerant gates for a topological code, the classification of gapped boundaries of this code, and SPT phases for which these gapped boundaries serve as ground states has been shown~\cite{BYCCSPT,BYFT}. This relationship is useful for the construction of fault-tolerant non-Clifford gates and may have applications in magic state distillation.  These results hint at a new relationship between such gapped domain walls and SPT-ordered phases on the boundary, in particular in higher dimensions.

Very little is known about the thermal stability of SPT-ordered systems, and the possibility is left open that some of the robust properties of SPT-ordered phases for quantum computing may survive at nonzero temperature when the local symmetry is enforced.  The presence of SPT-order in thermal systems is deeply connected to survival of the aforementioned gapped boundaries in a topological code and their associated fault-tolerant non-Clifford gates at nonzero temperature.  The survival of SPT-order for systems excited out of the ground state has been investigated in the context of many-body localization~\cite{bahri2015localization,yao2015many}.

Our first result is a proof that a nontrivial SPT phase protected by a global onsite (zero-form) symmetry cannot exist for any nonzero temperature.  This proof requires us to formulate a definition of nontrivial SPT-order for thermal states, which we do through an appropriate definition of a symmetric Gibbs state together with a definition of nontrivial SPT-order for mixed states based on circuit complexity following a similar approach by Hastings for topological order~\cite{HastingsThermal}. We prove this result for the broad class of models described by group cohomology \cite{ChenGuLiuWen}.

As SPT-order has been shown to be an enabling feature of measurement-based quantum computation, this no-go result would suggest that thermal states at nonzero temperature cannot be used as resource states for such schemes.  Surprisingly, though, we know this to be false, through the existence of several counterexamples.  The topological cluster state scheme of Raussendorf \textit{et al.}~\cite{Rau06} is the basis for essentially all currently-pursued high-error-threshold architectures for quantum computing (its circuit-model implementation gives the well-studied `surface code' architecture~\cite{FMMC12}). Using a cluster state Hamiltonian in three dimensions, the results of Ref.~\cite{RBH} show that the thermal state of this model is a resource for quantum computation below some critical temperature.  This is despite the fact that this cluster model Hamiltonian does not undergo any thermodynamic phase transition, even when protected by an onsite symmetry, and so the physical origin of its thermal stability remains elusive.   Other 3D Hamiltonians have been proposed that are universal for MBQC at nonzero temperature \cite{fujii2012topologically,li2011thermal,fujii2013measurement}, but there is currently no guiding principle explaining the thermal robustness of MBQC.

As our second result, we present and analyse the 3D cluster state model from the perspective of SPT-order, and show that this model possesses a nontrivial SPT phase at nonzero temperature when protected by a 1-form symmetry.  Higher-form symmetries are a natural generalization of the 0-form global symmetry for which the group action is onsite. A $q$-form symmetry can be imposed by an operator acting on a closed codimension-$q$ manifold $\mathcal{M}$.  When $q> 0$, the symmetry imposes much stronger constraints than the onsite, $q=0$ case. Several recent works have investigated SPT phases with higher-form symmetries~\cite{Baez04,Baez10,KThigher,Kapustin14b,GKSWhigher,BYhigher}.  By enforcing a 1-form symmetry on the 3D cluster state model, we prove that SPT-ordering in 1-form symmetric models can be maintained at nonzero temperature.  We explicitly construct types of nonlocal order parameters that characterize this SPT-ordering in the thermal state. These order parameters consist of pairs of membranes, and when equipped with local error-correcting operations on the boundaries serve as a witness of the long-range localizable entanglement that is present in the thermal state.

Our third result is to provide an operational interpretation of this SPT-ordering under the 1-form symmetry, using the concept of localizable entanglement in the thermal state.  This interpretation provides an explanation of the thermal stability of the topological cluster state model for quantum computation, even for the case where symmetries are not enforced.  In one dimension, the SPT-ordering at zero temperature of the cluster state model protected by a global 0-form $\zz_2 \times \zz_2$ symmetry is characterised by the ability to localize entanglement in the ground state between the fractionalized edge modes via symmetry-adapted measurements in the bulk~\cite{Else2012,Delse1}.  By analogy, in the 3D cluster state model, we demonstrate that our order parameter takes near-maximal values for the nontrivial SPT phase at low temperature, which guarantees robustness of the localizable entanglement between two boundary surface codes of this model via symmetry-adapted measurements in the bulk.  In addition to localizing entanglement, the measurements provide complete information about the 1-form symmetry operators.  Therefore, even when the 1-form symmetry is not enforced, measurement of these symmetry operators allows for error correction of the resulting thermal state to the corresponding thermal SPT-ordered state for which entanglement is ensured. Therefore, the scheme can offer thermal stability even without enforcing the symmetry.

The paper is organised as follows. In section \ref{sec2} we formulate and define the types of models and relevant notions of SPT-order for thermal states. We then provide a proof that SPT-order protected by an onsite symmetry cannot exist at nonzero temperature. We prove this first for a well known SPT model in 2D, and then for the more general group cohomology models. In section \ref{sec3} we show that the 3D cluster model possesses SPT-order at nonzero temperature, protected by a 1-form symmetry. We show this firstly through an argument based on gauging and secondly through a nonlocal order parameter. In section \ref{sec4} we discuss the nontrivial SPT protected by 1-form symmetry in the context of measurement-based quantum computation. We conclude with a discussion and outlook in section \ref{sec5}. 

\section{Thermal SPT-order}\label{sec2}
In this section we introduce the types of models we will be treating and the relevant definitions of SPT-order. We then develop a toolset to analyse SPT-order in a thermal setting, making use of the well-known framework of simulating thermalization of quantum many-body systems based on the Davies map \cite{Davies1, Davies2}. Our main result in this section is a proof of the instability of SPT-order at nonzero temperature for models in arbitrary dimension protected by an onsite symmetry.

\subsection{The setting}
 Consider a discrete lattice $\Lambda$ embedded in a $D$-dimensional manifold $M^D$. Spins with local Hilbert space $\calH_i$ are placed at each site $i \in \Lambda$ (`sites' can be chosen to be at vertices, edges, etc., of the lattice), with total Hilbert space $\calH = \otimes_{i \in \Lambda} \calH_i$. The types of models that we are considering can be represented by local, commuting projector Hamiltonians $H = \sum_X h_X$, where each local term $h_X$ is supported on a subset $X \subseteq \Lambda$ with diam$(X)\leq $ const. We assume that the system has some symmetry described by a group $G$, with unitary representation $S$. The symmetries we consider can be onsite symmetries, as well as more general higher-form symmetries, which we now define. An onsite symmetry takes the form 
 \begin{equation}
 S(g) = \bigotimes_{i\in \Lambda} u_i(g),
 \end{equation} 
 where $u_i(g)$ is the representation of $G$ on a single site $i\in \Lambda$. A $q$-form symmetry (for some $q\in \{0,1,\ldots,D-1\}$) consists of operators $S_{\calM}(g)$, supported on codimension-$q$ submanifolds $\calM$ in $M^D$, with $g\in G$ ~\cite{Baez04,Baez10,KThigher,Kapustin14b,GKSWhigher,BYhigher}. In this language, an onsite symmetry may also be referred to as a 0-form symmetry. In such a theory, charged excitations are $q$-dimensional objects and symmetry operators impose conservation laws on higher-dimensional charged objects.
 
A useful way to classify phases of matter at zero temperature is to use circuit complexity~\cite{ChenGuWen}. A quantum circuit may be represented as 
\begin{equation}
U_{\text{circ}} = \prod_{j=1}^{d} \calD_j \quad \text{where} \quad \calD_j = u_1^{(j)} \otimes u_2^{(j)} \otimes\ldots \otimes u_{k_j}^{(j)},
\end{equation}
where each geometrically local gate $u_{k}^{(j)}$ is supported on a region of radius at most $r$, and $d$ is the number of layers. The depth of such a circuit is defined to be the product $rd$, and a circuit is known as \textit{low-depth} if $rd$ is constant in the system size\footnote{Note that it is common to refer to $r$ and $d$ as the \textit{range} and the \textit{depth} of the circuit $U_{\text{circ}}$, respectively, but we do not make this distinction.}. We say a ground state of a gapped Hamiltonian $H$ is \textit{short-range entangled}, if it can be transformed into a product state using a low-depth circuit \cite{ChenGuWen}. In the context of SPT phases, the local gates $u_k^{(j)}$ of a quantum circuit are constrained to commute with the symmetry $S(g)$.

Namely, SPT-order at zero temperature is defined in the following way. Let $\ket{\psi}$ be the unique ground state of a gapped Hamiltonian $H$ on a closed (without boundary) lattice, with symmetry $G$. Then $\ket{\psi}$ belongs to a nontrivial SPT phase if:
 	\begin{enumerate}
 		\item it is short-range entangled,
 		\item any low-depth circuit connecting $\ket{\psi}$ to a product state has gates that break the symmetry. 
 	\end{enumerate}  

We emphasize that while there may exist a low-depth symmetric unitary map that connects a state with nontrivial SPT-order to a product state, the local gates composing it cannot be symmetric. SPT models have trivial bulk properties in the sense that they have no exotic excitations or degeneracies dependent on the topology of the underlying manifold. Despite this absence, interesting protected surface states are known to appear at the boundary of an SPT phase.  For example, in 1D, nontrivial SPT chains can exhibit fractionalized edge modes at their endpoints, such as with spin-1 antiferromagnets in the Haldane phase or Majorana nanowires.  In general, in higher dimensions, it is believed that the 1D surface of a 2D SPT must be gapless or break the symmetry \cite{CZX, LevinGu}, while in three or more dimensions it is believed that the surface must be gapless, break the symmetry or be topologically ordered \cite{BulkBoundarySPT, TranslationalSET}. 

A large and well-known class of SPT models are the group cohomology models protected by onsite symmetries \cite{ChenGuLiuWen}. In terms of circuit depth, using gates of constant range, these wavefunctions require a circuit of depth $\calO(N)$ to symmetrically disentangle, where $N$ is the number of spins (for example, the one-dimensional case is proven in \cite{HCComplexity}). While this class captures a large number of SPT phases, there are known models beyond group cohomology, including 3D models  that are protected by time reversal symmetry~\cite{SVBeyond,MKFBeyond, WSBeyond, BCFVBeyond}. More recently, looking beyond onsite symmetries has led to generalised SPT models protected by higher-form symmetries, both in the continuum and on the lattice \cite{KThigher, GKSWhigher, BYhigher}.

\subsection{Defining SPT-order for thermal states}

As defined above, SPT-order is manifestly a pattern of entanglement in the gapped ground state of a Hamiltonian. In this section we extend this definition to systems at nonzero temperature after briefly reviewing thermalization via the Davies map \cite{Davies1, Davies2}. We will argue that in the presence of symmetry, a natural notion of a thermal state at temperature $T$ is the \textit{symmetric} Gibbs ensemble
\begin{equation}\label{symgibbs}
\rho(\beta) = \lim\limits_{\lambda \rightarrow \infty} \rho_{\lambda}(\beta), 
\end{equation}
where $\beta = T^{-1}$, and $\rho_{\lambda}(\beta)$ is the (usual) Gibbs ensemble of a modified Hamiltonian $H(\lambda)$
\begin{equation}\label{symHam}
\rho_{\lambda}(\beta) = \calZ(\lambda)^{-1} e^{-\beta H(\lambda)}, \qquad H(\lambda) =  H - \lambda \sum_{g \in G} S(g),
\end{equation}
where $\calZ(\lambda) =\text{Tr} e^{-\beta H(\lambda)}$. Note that in the case of a higher-form symmetry, the sum in Eq. (\ref{symHam}) is over all symmetry operators. The symmetric ensemble arises naturally in two different contexts: $(i)$ the fixed point of a system thermalizing in the presence of a symmetry, $(ii)$ the post error corrected state of a thermal ensemble. We will overview the first point $(i)$ in this section, before returning to error correction in detail in section \ref{sec4B}.

To motivate this notion of a symmetric thermal state, consider thermalization as modelled by weakly coupling the system to a bosonic bath
\begin{equation}
H' = H_{S} \otimes {I}_B + {I}_S \otimes H_{B} + H_{int},
\end{equation}
where $H_{S}$ is the system Hamiltonian describing the SPT phase, $H_{B}$ is the bath Hamiltonian, and $H_{int} = \sum_{\alpha} s_{\alpha}\otimes b_{\alpha}$ is the interaction Hamiltonian comprised of the system and bath operators $s_{\alpha}$ and $b_{\alpha}$ respectively. The interaction Hamiltonian is constrained by the symmetry in that it must commute with the symmetry on the joint system $U(g) = S(g) \otimes {I}_{B}$. Note that we require no other symmetry of the bath, other than that the couplings respect the system symmetry $S(g)$. 

In order to realise the usual Gibbs ensemble as the fixed point of the reduced system dynamics, we require the dynamics to be ergodic. This is usually achieved by choosing bath couplings that are as simple as possible while ensuring the system operators address all energy levels of the system Hamiltonian $H_S$. The necessary and sufficient condition for ergodicity is that no nontrivial operators commute with all of the Hamiltonian and system operators \cite{Ergodic1, Ergodic2}. In the presence of symmetry, such a choice in general will not be possible, since the system operators $s_{\alpha}$ must respect the symmetry. Therefore ergodicity can only be achieved on a given sector, and for the sake of concreteness we focus our attention on the symmetric sector (the $+1$-eigenspace of $U(g)$).

We assume that the coupling is chosen such that the only operators that commute with $H_S$ and all of the system operators $s_{\alpha}$ are symmetry operators, and additionally that the initial state belongs to the symmetric sector. Then, following the Davies prescription, the unique fixed point of the dynamics generated by the above interaction will be the \textit{symmetric} Gibbs ensemble of Eq.~(\ref{symgibbs}). We will return to the assumption of the initial state belonging to the symmetric sector in section~\ref{sec4}, specifically in the context of error correction.

Given this ensemble, let us now define what it means to have SPT-order at nonzero temperature by modifying a definition due to Hastings \cite{HastingsThermal}. The notion of a \textit{trivial state} is replaced by a \textit{classical symmetric ensemble}, which is the symmetric Gibbs ensemble of a classical Hamiltonian $H_{cl}$. Here a classical Hamiltonian refers to a Hamiltonian expressible by a sum of terms diagonal in a local product basis. To define SPT at nonzero temperature, we follow Hastings \cite{HastingsThermal} and ask what is the circuit depth required to approximate the symmetric Gibbs ensemble, beginning with a classical Gibbs ensemble.
\begin{defn}
Let $\rho$ be the symmetric Gibbs state of a Hamiltonian $H$ that has symmetry $S(g)$, $g\in G$ and a SRE, unique ground state. We say $\rho$ is $(r,\epsilon)$ SPT-trivial if there exists: 
	\begin{enumerate}
		\item An enlarged Hilbert space $\calH' = \calH\otimes \calK$.
		\item A classical, nondegenerate Hamiltonian $H_{cl}$ defined on $\calH'$ with symmetry 
		\begin{equation}\label{JointSym}
		U(g)= S(g) \otimes {I}_{\calK}.
		\end{equation} 
		\item \vspace{-2.5mm} A circuit $\calU$ with depth $r$ acting on the enlarged space $\calH'$, composed of symmetric gates, such that 
		\begin{equation}
			\norm{ \rho - \Tr_{\calK} \left( \calU \rho_{cl} \calU^{\dagger} \right) }_1 \textless \epsilon,
		\end{equation}
		where $\rho_{cl}$ is the symmetric Gibbs ensemble of $H_{cl}$, and $\norm{\cdot}_1$ denotes the trace norm.
	\end{enumerate} 
\end{defn}
We make a few remarks on this definition. Firstly, we require $H_{cl}$ in the definition to be non-degenerate to exclude spontaneous symmetry breaking, since the symmetric Gibbs state of such a system can be highly nontrivial in terms of circuit depth. Secondly, we make the choice of symmetry in Eq.~(\ref{JointSym}) to avoid the following situation. Suppose the choice of symmetry was given by $U(g) = S(g) \otimes S(g)$. For any SPT-ordered state $\ket{\psi}$ with symmetry $S(g)$, there exists a state $\ket{\psi^{-1}}$ with symmetry $S(g)$ such that $\ket{\psi} \otimes \ket{\psi^{-1}}$ can be prepared from a product state by a constant depth circuit that is symmetric under $U(g)$. This property is referred to as the invertibility of SPT phases \cite{kongInvertible}. After tracing out the second subsystem, this choice of symmetry would imply that $\ket{\psi}$ is $(r,0)$ trivial (even at $T=0$) for some constant $r$.

Operationally, the above definition asserts than an SPT-trivial state is one that can be prepared from a classical ensemble using a low-depth symmetric quantum circuit (potentially with ancillas). An important consequence of this definition is that if a (symmetric) Gibbs ensemble can be expressed (up to error $\epsilon$) as a mixture of $(r,0)$ SPT-trivial states, then it is an $(r,\epsilon)$ SPT-trivial state \cite{HastingsThermal}. Indeed our strategy in the following section will be to show that the symmetric Gibbs ensemble of SPTs protected by onsite symmetries can be approximated by a convex combination of states, each of which is symmetrically low-depth equivalent to a product state.

\subsection{Onsite symmetric models have no SPT-order at nonzero temperature}

We now show that any SPT-ordered Hamiltonian $H$ with an \textit{onsite} symmetry is trivial at any $T \textgreater 0$ according to the above definition. For concreteness, we focus on a particular 2D example with $\zz_2$ onsite symmetry and defer the more general result to the next subsection. The proof proceeds by first constructing a new Hamiltonian $H'$ from $H$ whose Gibbs ensemble approximates that of $H$ and is obtained by removing terms from $H$. By dividing the lattice into disjoint regions of small size (i.e. logarithmic in the system size), the missing terms present within each region allow us to define a circuit with small-depth that transforms $H'$ into a trivial Hamiltonian describing a paramagnet. We find that many tools used to prove that two-dimensional, commuting projector Hamiltonians have trivial (intrinsic) topological order at nonzero temperature (in the absence of symmetry) in Ref.~\cite{HastingsThermal} apply in this context. 

Our proof method has the following physical interpretation. In the SPT-ordered Hamiltonian $H$, excitations are point-like objects and the $\zz_2$ onsite symmetry imposes a conservation law on $H$ that the number of point-like excitations must be even. By removing terms in $H'$, we create sinks where single point-like excitations can be created and destroyed, circumventing the above conservation law. Using these sinks, one can construct a symmetric disentangler out of operators that move point-like excitations into the sinks. This construction leaves open the possibility of thermal SPT-order in the presence of higher-form symmetries, as the removed terms do not act as sinks for the higher-dimensional excitations of these models, as we investigate in the next section. 

The example 2D model we consider was first discussed in \cite{LevinGu} (although it appeared, previously in a different guise in \cite{CZX}), and will capture the key ingredients used to prove the general case. Consider a triangular lattice whose set of vertices, edges and faces is labelled by $\Delta_0$, $\Delta_1$, and $\Delta_2$ respectively. On each vertex $v \in \Delta_0$ resides a qubit as in Fig.~\ref{figTriangleLattice}, and let $N = |\Delta_0|$ be the number of qubits. Consider first the trivial paramagnet
\begin{equation}
H_0 = -\sum_{v \in \Delta_0}X_v,
\end{equation}
where $X_v$ is the Pauli $X$ operator acting on the qubit at vertex $v$. The unique, gapped ground state of this model is the trivial product state $\ket{\psi_0} = \ket{+}^{\otimes N}$, where $\ket{+} = \frac{1}{\sqrt{2}}(\ket{0} +\ket{1})$. The Hamiltonian, and thus the unique ground state, possess an onsite $\zz_2$ symmetry generated by 
\begin{equation}\label{EqSym}
S = \bigotimes_{v \in \Delta_0} X_v.
\end{equation}
We would like to construct a model with the same symmetry, but belonging to a nontrivial SPT phase. We first define the controlled-$Z$ unitary acting on two qubits sharing an edge $e = (v_1,v_2)$ to be
\begin{equation}\label{EqCZ}
CZ_{(v_1,v_2)} = \exp\left(\frac{i \pi}{4}(I-Z_{v_1})(I-Z_{v_2})\right).
\end{equation}
The nontrivial model can be constructed from these operators as a sum of local terms
\begin{equation}
H_1 = -\sum_{v \in \Delta_0} h_v, \qquad h_v = X_v \prod_{e \in \text{Link}_1 (v)} CZ_e,
\end{equation}
where the $\text{Link}_1(v)$ consists of the neighbouring edges of $v$ that do not contain $v$, as depicted by thick blue edges in Fig.~\ref{figTriangleLattice}. We note that $H_1$ is slightly different to the model presented in \cite{LevinGu}, but they are equivalent up to a constant depth quantum circuit comprised of symmetric gates. Each of the terms $h_v$ are commuting and satisfy $h_v^2 = I$, and therefore have eigenvalues $\pm1$. One can confirm that this model shares the same $\zz_2$ symmetry as the trivial paramagnetic model $H_0$.

\begin{figure}%
	\centering
	\subfloat[]{{\includegraphics[width=0.32\linewidth]{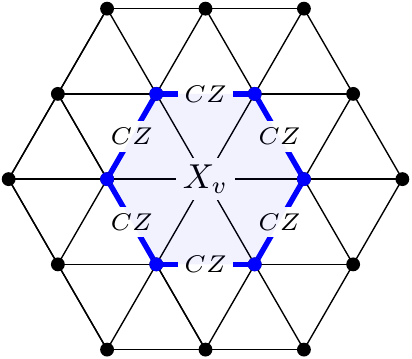} }\label{figTriangleLattice}}%
	\qquad
	\subfloat[]{{\includegraphics[width=0.28\linewidth]{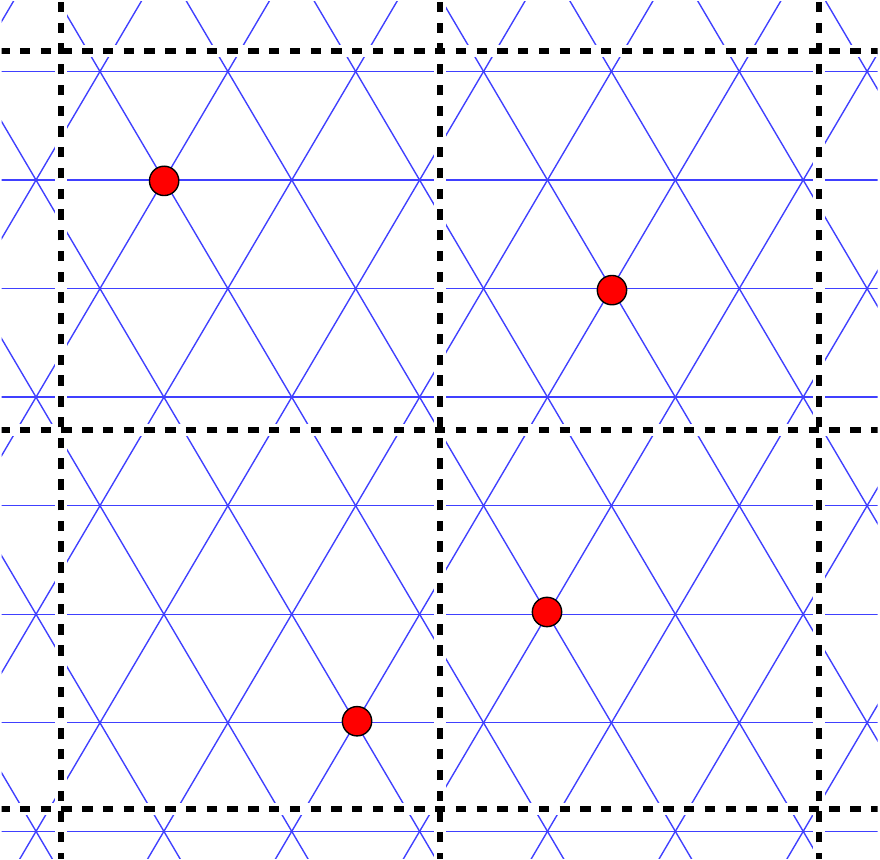} }\label{figLatticeGrid}}%
	\caption{(Color online) (a) The triangular lattice and one of the terms $h_v$ belonging to $H_1$. The 1-link of the vertex $v$ is the set of blue (thick line) edges. (b) A valid configuration has sinks (large dots) for each square region in $\mathcal{P}_l$, where a sink is a vertex $v$ with $k_v = 0$.}%
\end{figure}

The unique ground state $\ket{\psi_1}$ is the $+1$-eigenstate of each of the terms $h_v$. Additionally, one can show that this model is short-range entangled, as it can be connected to the trivial ground state via the following unitary $U_1 = \prod_{t \in \Delta_2}C^{\otimes 2}Z_{t}$, where $C^{\otimes 2}Z_{t}$ the 3-qubit controlled-$Z$ unitary acting on qubits in a triangle $t = (v_1,v_2,v_3)$ as:
\begin{equation}
C^{\otimes 2}Z_{t} = \exp\left(\frac{i \pi}{8}(I-Z_{v_1})(I-Z_{v_2})(I-Z_{v_3})\right).
\end{equation}
The unitary $U_1$ given by the whole circuit commutes with the symmetry, $[U_1, S] = 0$, provided the lattice has no boundary. But importantly, each gate in the circuit is not symmetric on its own, $[C^{\otimes 2}Z_{t}, S] \neq 0$. It is shown in \cite{LevinGu} that $H_1$ cannot be adiabatically connected to the trivial paramagnet $H_0$ without closing the gap or breaking the symmetry, so it is impossible to approximate $U_1$ by a constant depth circuit comprised of symmetric gates. Therefore, $H_1$ has nontrivial SPT-order at zero temperature. 

Now let us show that the model $H_1$ becomes SPT-trivial at nonzero temperature. Similarly to Refs.~\cite{HastingsThermal, SY}, we associate a binary value $k_v \in \{0,1\}$ to each site $v \in \Delta_0$ to indicate the presence or absence of a term in an \textit{imperfect} Hamiltonian:
\begin{equation}
H(\textbf{\textit{k}}) = -\sum_{v\in \Delta_0} k_v h_v,
\end{equation}
where $\textbf{\textit{k}} \in \{0,1\}^N$. For a given imperfect Hamiltonian $H(\textbf{\textit{k}})$ we say a site $v$ is a \textit{sink} if $k_v = 0$, corresponding to a missing term. We now wish to express the Gibbs ensemble in terms of a convex sum of the ground spaces of imperfect Hamiltonians.
Let $\overline{\rho}(\textbf{\textit{k}})$ be the uniform mixture of symmetric ground states of $H(\textbf{\textit{k}})$. Then following~\cite{SY}, we define the free symmetric ensemble at $\beta=T^{-1}$ as
\begin{equation}\label{EqFreeGibbs}
\rho_{f}(\beta) = \sum_{\textbf{\textit{k}} \in \{0,1\}^{N}} \Pr(\textbf{\textit{k}})\overline{\rho}(\textbf{\textit{k}}),
\end{equation}
where $\Pr(\textbf{\textit{k}})$ is a probability distribution
\begin{equation}\label{DefPr}
\Pr(\textbf{\textit{k}}) = (1-p_{\beta})^{w(\textbf{\textit{k}})}p_{\beta}^{N-w(\textbf{\textit{k}})}, \qquad p_{\beta} = \frac{2}{e^{2\beta} +1},
\end{equation}
and $w(\textbf{\textit{k}})$ is the Hamming weight of the vector $\textbf{\textit{k}}$ (the number of nonzero entries).

\begin{lem}\label{SinkDist}
	Let $\rho(\beta)$ be the symmetric Gibbs ensemble of $H_1$ with $T \textgreater 0$, then 
	\begin{equation}
	\norm{\rho(\beta) - \rho_f(\beta)}_1 \leq \calO(e^{-\eta N})
	\end{equation}
	for some constant $\eta \textgreater 0$ (independent of system size). 
\end{lem}
\begin{proof}
	The proof is similar to that in Ref.~\cite{SY}. Consider first the usual Gibbs ensemble $\rho'(\beta)$ of $H_1$ (without enforcing the symmetry). Because $H_1$ is a sum of commuting terms, we have
	\begin{align}
	\rho'(\beta) &= \frac{1}{\calZ'} \prod_{v \in \Delta_0} e^{\beta h_v}.
	\end{align}
	Since each term satisfies $h_v^2 = I$, we have $\exp(\beta h_v) = \cosh(\beta)I + \sinh(\beta) h_v$. Introducing a new normalization factor $\tilde{\calZ} = {(e^{ \beta}+e^{-\beta})^{N}}{\calZ'}$ we have
	\begin{align}
	\rho'(\beta) &= \frac{1}{\tilde{\calZ}}\prod_{v \in \Delta_0}  \left((1-p)\frac{I + h_v}{2} + p\frac{I}{2} \right),
	\end{align}
	where we have set $p = {2}/{(e^{2\beta} +1)}$. Expanding this out and introducing a dummy binary variable $k_v$ for each vertex $v \in \Delta_0$, we have 
	\begin{equation}
	\rho'(\beta) = \frac{1}{\tilde{\calZ}}\sum_{\textbf{\textit{k}} \in \{0,1\}^{N}} \left(\prod_{v \in \Delta_0}(1 - p)^{k_v}p^{1-k_v}\left( \frac{I + k_v h_v}{2}\right)\right),
	\end{equation}
	which we can rewrite as
	\begin{equation}
	\rho'(\beta) = \frac{1}{\tilde{\calZ}}\sum_{\textbf{\textit{k}} \in \{0,1\}^{N}} \Pr(\textbf{\textit{k}}){\rho}(\textbf{\textit{k}}), 
	\end{equation}
	where
	\begin{equation}
	{\rho}(\textbf{\textit{k}}) = \frac{1}{2^{N}}\prod_{v \in \Delta_0}\left(I + k_v h_v \right),
	\end{equation}
	and $\Pr(\textbf{\textit{k}})$ is given by Eq. (\ref{DefPr}). Note that ${\rho}(\textbf{\textit{k}})$ is a uniform mixture of all ground states of the imperfect Hamiltonian $H(\textbf{\textit{k}})$. Let us confirm that the normalization of ${\rho}(\textbf{\textit{k}})$ is correct. For any subset $M \subseteq \Delta_0$, we have 
	\begin{equation}
	\Tr\left(\prod_{v \in M} h_v\right) = \Tr\left(\prod_{v \in M}  U_1X_vU_1^{\dagger}\right) 
	= \Tr\left(\prod_{v \in M} X_v\right)
	= 0,
	\end{equation}
	and therefore $\Tr\left( {\rho}(\textbf{\textit{k}})\right) = 1$. Now notice that 
	\begin{equation}\label{EqSumProb}
	\sum_{\textbf{\textit{k}}\in \{0,1\}^N} \Pr(\textbf{\textbf{\textit{k}}}) =\sum_{l=0}^{N} \binom{N}{l}(1-p)^{l}p^{N-l} =1,
	\end{equation}
	and therefore $\tilde{\calZ} = 1$, which means we have 
	\begin{equation}
	\rho'(\beta) = \sum_{\textbf{\textit{k}} \in \{0,1\}^{N}} \Pr(\textbf{\textit{k}}) {\rho}(\textbf{\textit{k}}).
	\end{equation}
	Having considered the usual Gibbs ensemble without symmetries, we now consider the Gibbs ensemble with the symmetry enforced. Let $P= (I + S)/2$ be the projector onto the $+1$-eigenspace of $S$ (recall, $S$ is the symmetry operator defined in Eq. (\ref{EqSym})). The symmetric Gibbs ensemble $\rho(\beta)$ can be obtained by projecting $\rho'(\beta)$ into the symmetric sector and renormalizing
	\begin{equation}\label{EqSymEnsP}
	\rho(\beta) = \frac{P \rho'(\beta)P}{\Tr(P \rho'(\beta)P)}.
	\end{equation}	
	For $\textbf{\textit{k}}_1  := (1,\ldots,1)$, the Hamiltonian $H(\textbf{\textit{k}}_1)$ has a unique and symmetric ground state and therefore $\overline{\rho}(\textbf{\textit{k}}_1) = {\rho}(\textbf{\textit{k}}_1)$. For $\textbf{\textit{k}} \neq \textbf{\textit{k}}_1$, the imperfect Hamiltonian $H(\textbf{\textit{k}})$ has a $2^{N - w(\textbf{\textit{k}})}$-dimensional ground space, which is partitioned equally into the $+1$- and $-1$-eigenspaces of $S$. Therefore we have 
	\begin{equation}\label{EqTrSym}
	\Tr(P\rho(\textbf{\textit{k}})P) = \frac{1}{2} \Tr(\rho(\textbf{\textit{k}})) = \frac{1}{2} \qquad \forall \textbf{\textit{k}} \neq \textbf{\textit{k}}_1.
	\end{equation}
	The symmetric ground space projectors of the imperfect Hamiltonian $H(\textbf{\textit{k}})$ can be written
	\begin{equation}\label{EqSymGSP}
	\overline{\rho}(\textbf{\textit{k}}) = 
	\begin{cases}
	\rho(\textbf{\textit{k}}) \quad \text{if} \quad \textbf{\textit{k}} = \textbf{\textit{k}}_1, \\
	2 P\rho(\textbf{\textit{k}})P \quad \text{otherwise}
	\end{cases}
	\end{equation}
	Let us evaluate the normalization factor $\calZ = \Tr(P \rho'(\beta)P)$. We obtain 
	\begin{align}
	\calZ&= \sum_{\textbf{\textit{k}}\in \{0,1\}^N} \text{Pr}(\textbf{\textit{k}}) \Tr(P\rho(\textbf{\textit{k}})P) = \sum_{\textbf{\textit{k}}\neq \textbf{\textit{k}}_1} \frac{1}{2}\text{Pr}(\textbf{\textit{k}}) + \text{Pr}(\textbf{\textit{k}}_1) = \frac{1}{2}(1 + \text{Pr}(\textbf{\textit{k}}_1)).  \label{EqZPrk1}
	\end{align}
	In particular, notice that $\calZ \in [\frac{1}{2}, 1]$.
	Then the trace distance between $\rho_f(\beta)$ and $\rho(\beta)$ is given by
	\begin{align}
	\norm{\rho(\beta)- \rho_f(\beta)}_1 &= \norm{\calZ^{-1} \sum _{\textbf{\textit{k}} \in \{0,1\}^N} \Pr(\textbf{\textit{k}}) P \rho(\textbf{\textit{k}}) P - \sum _{\textbf{\textit{k}} \in \{0,1\}^N} \Pr(\textbf{\textit{k}}) \overline{\rho}(\textbf{\textit{k}})}_1
	\end{align}
	Using Eq. (\ref{EqSymGSP}), and the triangle inequality, we get 
	\begin{align}
	\norm{\rho(\beta)- \rho_f(\beta)}_1 &\leq \left(2-\calZ^{-1}\right)\sum_{\textbf{\textit{k}} \neq \textbf{\textit{k}}_1} \text{Pr}(\textbf{\textit{k}})\norm{P\rho(\textbf{\textit{k}})P}_1 + (1-\calZ^{-1})\text{Pr}(\textbf{\textit{k}}_1)\norm{\rho(\textbf{\textit{k}}_1)}_1 \\
	&= \left(2-\calZ^{-1} \right)\sum_{\textbf{\textit{k}} \neq \textbf{\textit{k}}_1} \text{Pr}(\textbf{\textit{k}})\frac{1}{2} + (\calZ^{-1} - 1)\text{Pr}(\textbf{\textit{k}}_1),
	\end{align}
	where we have used Eq. (\ref{EqTrSym}) in the second line. Then making use of Eqs. (\ref{EqSumProb}) and (\ref{EqZPrk1}), we get
	\begin{align}
	\norm{\rho(\beta)- \rho_f(\beta)}_1 &\leq \frac{1}{2}\left(2-\calZ^{-1}\right)(1- \Pr(\textbf{\textit{k}}_1)) + \left(\calZ^{-1} - 1\right) \text{Pr}(\textbf{\textit{k}}_1)\\
	& \leq \frac{2\Pr(\textbf{\textit{k}}_1)}{1+ \Pr(\textbf{\textit{k}}_1)}.
	\end{align}
	Since $\Pr(\textbf{\textit{k}}_1) = (1-p_{\beta})^N$ and $p_{\beta} \in (0,1]$ for $T \textgreater 0$, we therefore have 
	\begin{equation}
	\norm{\rho(\beta)- \rho_f(\beta)}_1 \leq 2 e^{-N \log(1-p_{\beta})}.
	\end{equation}
	Setting $\eta = -\log(1-p_{\beta}) \textgreater 0$, the claim follows.
\end{proof}

We now divide up the lattice into a square grid $\calP_l$ as in Fig.~\ref{figLatticeGrid}, with each square region having side-length $l=\left(c \log(L)\right)^\frac{1}{2}$ for some constant $c$. We will choose $c$ to be sufficiently large to ensure that, with high probability, there will be at least one sink within each square region.  A configuration $\textbf{\textit{k}}$ is called $l$-\textit{valid} if there is a sink in every square region and invalid otherwise. We want to show that the Gibbs state $\rho(\beta)$ at inverse temperature $\beta$ is well approximated by a distribution over $l$-valid configurations.

\begin{lem}\label{validProb}
	For a given grid $\calP_l$, let $\calV \subseteq \{0,1\}^N$ be the set of $l$-\textit{valid} configurations, and let
	\begin{equation}\label{imperfectState}
	\rho_{\calV}(\beta) = \sum_{\textbf{k} \in \calV} \Pr(\textbf{k}) \overline{\rho}(\textbf{k}).
	\end{equation}
	For any $T \textgreater 0$, there exists a constant $\delta\textgreater 0$ (independent of system size) such that $\rho_{\calV}(\beta)$ satisfies
	\begin{equation}
	\norm{\rho_{\calV}(\beta) - \rho(\beta)}_1 \leq \calO \left( L^{-\delta}\right).
	\end{equation}
\end{lem}
\begin{proof}
	Recall that $\Pr(\textbf{\textit{k}})= \prod_{v \in \Delta_0} (1-p)^{k_v}p^{1-k_v}$, such that $1-p$ is the probability of having a sink at a given vertex. Let $\text{P}_{\calV}:=\sum_{\textbf{\textit{k}}\in \calV} \Pr(\textbf{\textit{k}})$, then from Lemma \ref{SinkDist}, we have the following	
	\begin{align}
	\norm{\rho_{\calV}(\beta) - \rho(\beta)}_1 &= \norm{\sum_{\textbf{\textit{k}}\not\in \calV} \Pr(\textbf{\textit{k}}) \bar{\rho}(\textbf{\textit{\textbf{\textit{k}}}})}_1 \\
	&\leq \sum_{\textbf{\textit{k}}\not\in \calV} \Pr(\textbf{\textit{k}}) \\
	&= 1-\text{P}_{\calV} \label{EqValidBound}
	\end{align}	
	Let $B$ be the set of vertices within a square region of the grid $\calP_l$. The contribution of configurations containing at least one sink in each square region is given by 
	\begin{equation}
	\text{P}_{\calV} = \prod_{\text{squares } B} (1 - q_B),
	\end{equation}
	where $q_B  = (1-p)^{|B|}$ is the probability that square region $B$ contains no sink.	  Since the probability of each square having a sink is independent, and there are $n = L^2/c\log(L)$ squares in the grid $\calP_l$, using Bernoulli's inequality, we have
	\begin{equation}
	\text{P}_{\calV} = (1- q_B)^n \geq 1 - nq_B.
	\end{equation}
	Since $|B| = c\log(L)$, we have $q_B = L^{c \log(1-p)}$, and Eq. (\ref{EqValidBound}) becomes 
	\begin{align}
	\norm{\rho_{\calV}(\beta) - \rho(\beta)}_1 &\leq \frac{L^{2+c\log(1-p)}}{c\log(L)}\\
	&\leq \frac{1}{c}L^{-\delta},
	\end{align}
	where we have defined $\delta = -2 - c\log(1-p)$. Notice that for $T\textgreater 0$, we have $\log(1-p) \textless 0$. Therefore, choosing $c\textgreater -2/\log(1-p)$ gives $\delta \textgreater 0$.
\end{proof}

We can now show that the symmetric Gibbs ensemble $\rho(\beta)$ is SPT-trivial by constructing a symmetric disentangling circuit that maps $\overline{\rho}(\textbf{\textit{k}})$ to a product state, for each valid configuration $\textbf{\textit{k}}$. Then since $\rho(\beta)$ is approximated by a sum of SPT-trivial states, it follows that $\rho(\beta)$ is SPT-trivial. We note that the following theorem also applies if we replace the symmetric Gibbs ensemble by the usual Gibbs ensemble.

\begin{thm}\label{thm2dTriv}
	For any $T \textgreater 0$, the symmetric Gibbs ensemble $\rho(\beta)$ of $H_1$ is $(r, \epsilon)$ SPT-trivial, where $r=\calO(\log^{\frac{1}{2}}(L))$, and $\epsilon=\calO\left(L^{-\delta}\right)$.
\end{thm}
\begin{proof}
	Let $\textbf{\textit{k}}$ be a valid configuration. To construct a disentangling circuit $\calD_{\textbf{\textit{k}}}$ for $\overline{\rho}(\textbf{\textit{k}})$ we define the elementary gates of the circuit
	\begin{equation}\label{disentangler}
	U_{(v,w)} = \exp\left(\frac{\pi}{4} h_v Z_vZ_{w}\right), \quad W_{(v,w)} = \exp\left(\frac{\pi}{4}X_v Z_vZ_{w} \right), \quad v,w \in \Delta_0.
	\end{equation}
	Notice that both $U_{(v,w)}$ and $W_{(v,w)}$ are symmetric.	Moreover, for any vertex $v$, and any sink $h$, the operator $Z_v Z_h$ has the following commutation and anti-commutation relations
	\begin{equation}\label{ExcitCom}
	\qquad \{ h_v, Z_vZ_{h} \} = 0, \quad [h_w , Z_vZ_{h}] = 0 \quad \forall w\neq v.
	\end{equation}
	Because of the above relations, we can interpret $Z_vZ_{h}$ as an operator which creates an excitation at vertex $v$ in the imperfect Hamiltonian $H(\textbf{\textit{k}})$.
	
	The disentangling circuit $\calD_{\textbf{\textit{k}}}$ is composed of a number of layers $\calD_{\textbf{\textit{k}}} = \prod_{j=1}^d\calD_j$, such that each $\calD_j$ is comprised of gates with constant range, and $d= (c'\log(L))^{\frac{1}{2}}$ for some constant $c'$. The goal is to first disentangle terms near each sink, and then then inductively the next nearest neighbours, moving outwards as depicted in Fig.~\ref{figDisentangler}. We define sets of vertices which determine the order that we perform the gates. Let the initial set of vertices $V(0)$ contain exactly one sink in each square region (if there are many in each square region, choose any of them). Then for $j \geq 1$, let 
	\begin{equation}
	V(j) = \{v \in \Delta_0 ~|~ \text{dist}(v,w) \leq j \text{ for some } w \in V(0) \},
	\end{equation}
	be the union of balls of radius $j$ around each sink, where $\text{dist}(v,w)$ is the shortest path between vertices $v$ and $w$. We also define 
	\begin{equation}
	\overline{V}(j) = V(j) \setminus V(j-1),
	\end{equation}
	to be the set of vertices in $V(j)$ that are not in $V(j-1)$. Notice that for increasing $j$, $V(j)$ defines neighbourhoods of increasing size around each of the sinks, and that $\overline{V}(j)$ can be considered the boundary set of vertices of $V(j)$. 

	\begin{figure}[htb!]
		\centering
		\includegraphics[width=0.32\linewidth]{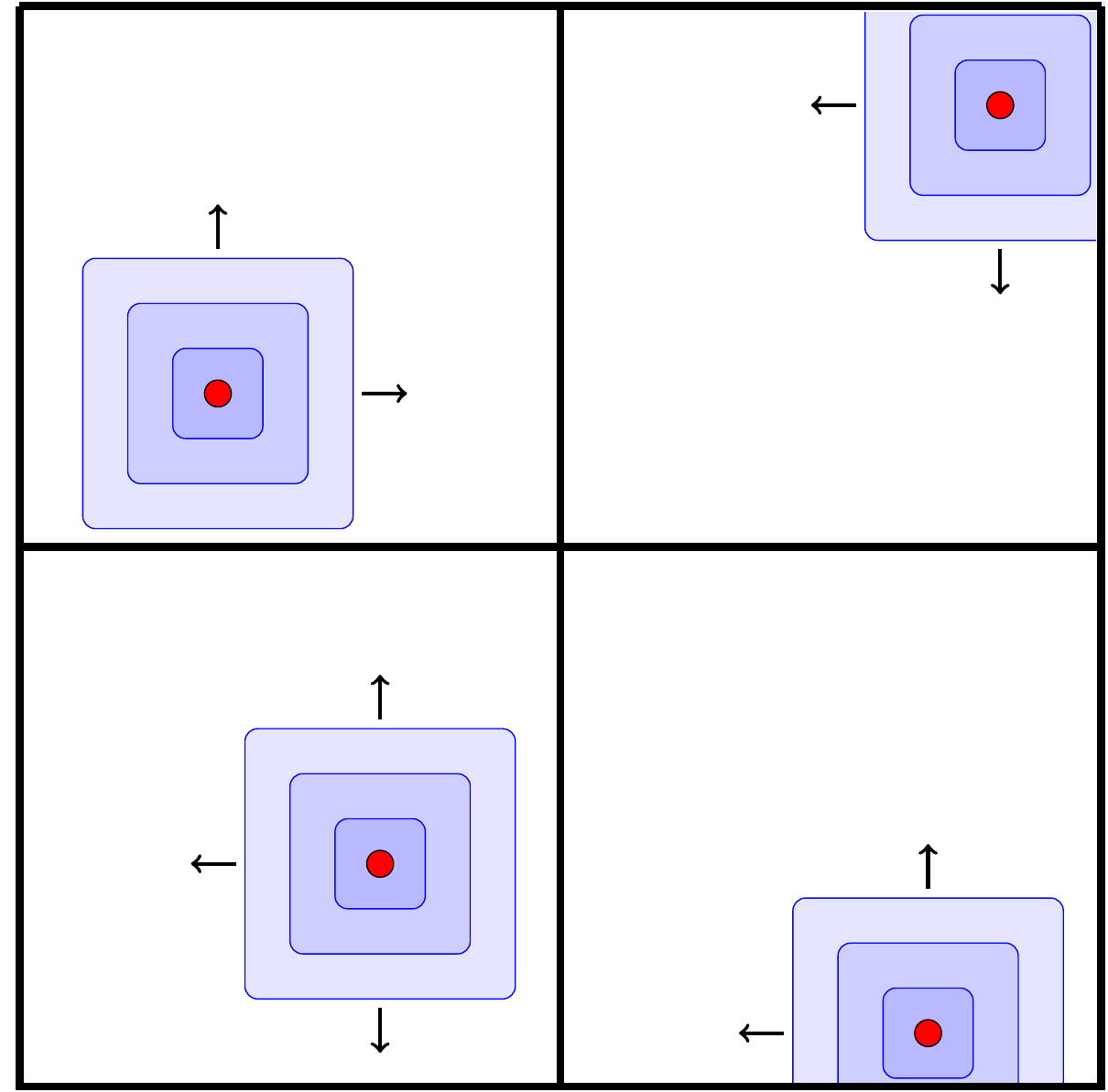}
		\caption{(Color online) The disentangler acts first on terms in $H(\textbf{\textit{k}})$ neighbouring the sinks (large dots), then moves outward. The set $V(0)$ consists of the sinks, depicted as large dots, and the successive shaded (blue) discs represent the sets $V(1)$, $V(2)$ and $V(3)$.
		} 
		\label{figDisentangler}
	\end{figure}
	
	For any vertex $v$, let $h_j(v)\in V(j)$ be the nearest vertex to $v$ that belongs to $V(j)$ (if there are multiple, choose any of them). Then the $j$'th layer of the circuit is defined by 
	\begin{equation}
	\calD_j' = \prod_{v \in \overline{V}(j)} U_{(v,h_{j-1}(v))} .
	\end{equation}
	Now $\calD_j'$ has constant depth for each $j$, because it is comprised of gates supported on a small neighbourhood of $\overline{V}(j)$. The gates can be divided into non-overlapping sets, each of which can be performed simultaneously (for example, the lattice is 3-colorable, and all gates $U_{(v,h_{j-1}(v))}$ with $v$ a fixed colour can be performed in parallel).  	

	Each gate $U_{(v,w)}$ has the following action under conjugation:
	\begin{align}\label{mapU}
	h_v \mapsto -Z_vZ_{w}, 
	\end{align}
	and commutes with $h_l$ for all $l\neq v,w$, and $Z_xZ_{y}$ for all $x,y\neq v$. Notice that for the first layer, $\calD_1$ conjugates all the terms $h_v$ sharing an edge with a sink into terms $-Z_vZ_{h_0(v)}$, where $h_0(v)$ is the sink adjacent to the vertex $v$. Subsequent layers $\calD_j$ map all the terms $h_v$ inside $V(j)$ to terms of the form $Z_vZ_w$. Let the constant $c'$ be chosen such that $d=(c'\log(L))^{\frac{1}{2}}$ is the diameter of each square region. Since each square region has a sink in it, we have $V(d)=\Delta_0$. Therefore, after at most $d$ layers, the circuit $\calD_{\textbf{\textit{k}}}' = \prod_{j=1}^d \calD_j'$ conjugates the imperfect Hamiltonian $H(\textbf{\textit{k}})$ into a sum of terms of the form $Z_vZ_w$. 
	
	Next, we make use of the gates $W_{(v,w)}$. In a similar way, we define the $j$'th layer of a second circuit by 
	\begin{equation}
	\calD_j'' = \prod_{v \in \overline{V}(j)} W_{(v,h_{j-1}(v))}.
	\end{equation}	
	The circuit $\calD_{\textbf{\textit{k}}}'' = \prod_{j=1}^d \calD_j''$ has depth $d$, as can be shown by the same argument given for $\calD_{\textbf{\textit{k}}}'$. Each gate $W_{(v,w)}$ has the following action under conjugation:	
	\begin{align}\label{mapU}
	Z_vZ_{w} \mapsto X_v, 
	\end{align}	
	and commutes with $Z_lZ_m$ for all $l,m \neq v$. Defining $\calD_{\textbf{\textit{k}}} = \calD_{\textbf{\textit{k}}}'' \circ\calD_{\textbf{\textit{k}}}'$, the circuit $\calD_{\textbf{\textit{k}}}$ applied to the imperfect Hamiltonian has the following action
\begin{equation}\label{aferro}
\calD_{\textbf{\textit{k}}} H(\textbf{\textit{k}}) \calD_{\textbf{\textit{k}}}^{\dagger} = \sum_{v \in \Delta_0} k_vX_v := H_0(\textbf{\textit{k}}).
\end{equation}	
Therefore the circuit $\calD_{\textbf{\textit{k}}}$ maps $\overline{\rho}(\textbf{\textit{k}}) \mapsto \overline{\rho}_0(\textbf{\textit{k}})$, where $\overline{\rho}_0(\textbf{\textit{k}})$ is the (normalised) symmetric ground space projector of $H_{0}(\textbf{\textit{k}})$, which is a product state. This holds for each valid configuration $\textbf{\textit{k}}\in \calV$ and therefore each $\overline{\rho}(\textbf{\textit{k}})$ is a $(2d,0)$-trivial state, where $d=(c'\log(L))^{\frac{1}{2}}$ for some constant $c'$.

A state $\rho(\beta)$ is $(r,\epsilon)$ SPT-trivial if and only if it can be approximated up to error $\epsilon$ (in trace norm) by a convex combination of $(r,0)$ SPT-trivial states. Since from Lemma~\ref{validProb} we have that $\rho(\beta)$ is approximated to within $\epsilon = \calO \left( L^{-\delta}\right)$ error by the imperfect state in Eq.~(\ref{imperfectState}), and the imperfect state is a convex combination of $(2d,0)$ SPT-trivial states, the result then follows.
\end{proof}	

The existence of a symmetric unitary $\calD$ that disentangles terms is closely related to the existence of sinks at some sites $k_v=0$, where point-like excitations can be locally created and destroyed. The existence of such excitations is a generic feature of Hamiltonians describing SRE phases, which suggests how the proof can be generalized to arbitrary dimension. In the next subsection, we sketch the more general proof for any SPT models based on group cohomology, using the tools developed in this section.  

\subsection{Thermal instability of SPT for group cohomology models}

We now prove the more general formulation of Theorem \ref{thm2dTriv}: that SPTs in arbitrary dimension, protected by onsite symmetries are trivial at nonzero temperature. We prove this statement for a class of models based on the group cohomology formalism~\cite{ChenGuLiuWen}. This class captures many of the known SPT phases protected by onsite symmetries, and we believe the arguments presented here can be generalised to models with onsite symmetries outside of the formalism. The construction of these models involves some technical details which we briefly review. 

The models are constructed in terms of special functions known as cocycles of the group $G$. A $d$-cochain of the group $G$ over $U(1)$ is a function $\nu_{d}: G^{d+1} \rightarrow U(1)$ that satisfies
\begin{equation}
\nu_d(g_0,g_1,\ldots,g_d)  = \nu_d(gg_0,gg_1,\ldots,gg_d) \qquad \forall g, g_k \in G.
\end{equation}
An important set of $d$-cochains are the $d$-cocycles, which satisfy the additional cocycle condition for any $d+2$ elements $g_0,\ldots,g_{d+1}$ of $G$, namely
\begin{equation}
\prod_{j=0}^{d+1} \nu_d (g_0, \ldots, g_{j-1}, g_{j+1} , \ldots , g_{d+1} )^{(-1)^j} = 1 \qquad \forall g_k \in G.
\end{equation}
An equivalence relation on the set of $d$-cocycles is given by multiplication by a $d$-coboundary. A $d$-coboundary $\lambda_d$ is a $d$-cochain that can be expressed as 
\begin{equation}
\lambda_d(g_0, g_1, \ldots, g_d) = \prod_{j=0}^d \mu_{d-1}(g_0,\ldots,g_{j-1},g_{j+1},\ldots, g_{d})^{(-1)^j},
\end{equation}
for some $(d-1)$-cochain $\mu_{d-1}$. Note that every $d$-coboundary is a $d$-cocycle, but not necessarily the other way around. The equivalence classes of $d$-cocycles are labelled by elements of the $d$-cohomology group $\calH^{d}(G, U(1))$.

For a system with symmetry group $G$ in $d$ spatial dimensions, consider a triangulation $T^{\Delta}$ of a $d$-dimensional manifold.  We label the $k$-simplexes of the triangulation by $\sigma_k$, and the set of all $k$-simplexes by $\Delta_k$. We assume that $T^{\Delta}$ has a bounded degree (the number of edges containing any given vertex must be constant). Additionally, we require that the triangulation has a branching structure (an orientation on each edge such that there is no oriented loop on any triangle) which allows us to give a parity $P(\sigma_d) = \pm 1$ to each $d$-simplex. To each vertex $v \in \Delta_0$, we associate a $|G|$-dimensional Hilbert space, a basis for which is given by $\{ \ket{g} , g \in G \}$. Let $N = |\Delta_0|$ be the number of spins. The symmetry action is given by the left regular representation
\begin{equation}
S(g) \ket{g_1,\ldots,g_N} = \ket{gg_1,\ldots,gg_N}.
\end{equation}
Consider first the trivial product state
\begin{equation}
\ket{\psi_0} =  \ket{+}^{\otimes N}, \quad \text{where} \quad \ket{+} = \frac{1}{\sqrt{|G|}} \sum_{g\in G} \ket{g},
\end{equation}
which is the ground state of the trivial Hamiltonian
\begin{equation}\label{trivcohoSPT}
H_0 = \sum_{v \in \Delta_0}\left(I - 2 |{+}\rangle\langle{+}|_v\right),
\end{equation}
where the notation $|{+}\rangle\langle{+}|_v$ means the projector $|{+}\rangle\langle{+}|$ at site $v$, and identity elsewhere. Notice that $\left(I - 2 |{+}\rangle\langle{+}|_v\right)^2 = I$. For any $(d+1)$-cocycle $\nu_{d+1}$, one can construct the unitary 
\begin{align}
U = \prod_{\sigma_d \in \Delta_d} {(U^{\sigma_d}_{\nu_{d+1}})}^{P(\sigma_d)}, \qquad P(\sigma_d)=\pm 1,
\end{align}
where $U^{\sigma_d}_{\nu_{d+1}}$ acts on spins that are vertices of $\sigma_d$ and $P(\sigma_d)$ represents the orientation of $\sigma_d$. Here $U_{\nu_{d+1}}$ is a $(d+1)$-body diagonal phase operator that acts as
\begin{align}
U_{\nu_{d+1}}|g_{1},\ldots,g_{d+1}\rangle = \nu_{d+1}(1,g_{1},\ldots,g_{d+1})|g_{1},\ldots,g_{d+1}\rangle.
\end{align}
Consider the Hamiltonian $H(\nu_{d+1}) = UH_0U^{\dagger}$ with ground state $\ket{\psi(\nu_{d+1})} = U \ket{\psi_{0}}$. Two important results in Ref.~\cite{ChenGuLiuWen} are the following:
\begin{enumerate}
	\item The unitary $U$, Hamiltonian $H(\nu_{d+1})$, and state $\ket{\psi(\nu_{d+1})}$ are symmetric under the onsite symmetry of $G$.
	\item If $\nu_{d+1}$ is nontrivial (i.e. it is not equivalent to the constant $\nu_{d+1}' = 1$), then $\ket{\psi(\nu_{d+1})}$ belongs to a nontrivial SPT phase.
\end{enumerate}
An important consequence of the cocycle functions that will be used in our proof, is their invariance under the so-called Pachner moves. Pachner moves are local operations that convert one triangulation into another one. Any two triangulations of a (piecewise linear) manifold can be related by a sequence of Pachner moves. In two-dimensions, the two basic Pachner moves are shown in Fig.~\ref{figPachner}. If two triangulations are related by a sequence of Pachner moves, then the SPT wavefunctions on these triangulations are related by a symmetric unitary combined with the addition/removal of ancillas in the $\ket{+}$ state. Since a sequence of Pachner moves corresponds to a symmetric unitary, we can define the depth of this sequence. Namely, we define a parallel Pachner move as any sequence of Pachner moves performed on disjoint $d$-simplexes. Then the depth of a Pachner sequence is the number of parallel Pachner moves, multiplied by the (max) diameter of the $d$-simplexes that are acted upon (this equals the depth of the corresponding symmetric unitary). 

\begin{figure}
	\centering
	\includegraphics[width=0.38\linewidth]{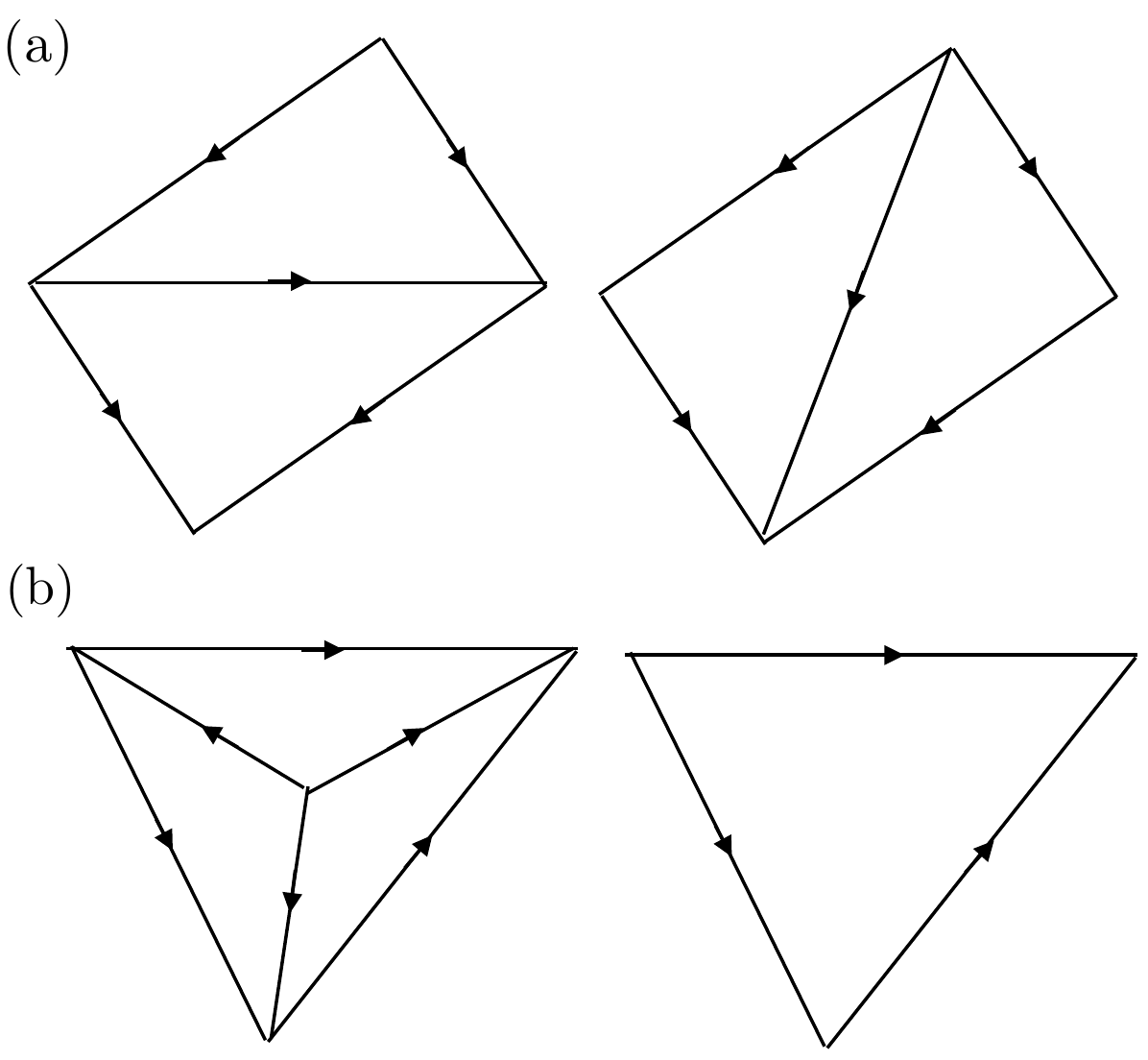}
	\caption{In two dimensions, there are two distinct Pachner moves: (a) two triangles are replaced by two triangles and (b) three triangles are replaced by one triangle and the number of vertices changes by 1. The arrows represent the orientation of each edge. Notice that there are no oriented loops on any triangle.
	} 
	\label{figPachner}
\end{figure}

We now prove that for any group $G$ and in any dimension $d$, the above Hamiltonian must have trivial SPT-order at nonzero temperature. The proof proceeds in a similar way to Theorem \ref{thm2dTriv}, where we first approximate the Gibbs ensemble $\rho(\beta)$ by a convex combination of valid configurations, and then show that each valid configuration is low-depth equivalent to a classical ensemble. Since a combination of trivial ensembles is trivial, the result follows.
\begin{thm}\label{onsiteTrivial}
	For any $T \textgreater 0$, the symmetric Gibbs state $\rho(\beta)$ of $H(\nu_{d+1})$ is $(r, \epsilon)$ SPT-trivial, where $r= \calO(\log(L) \log\log(L))$, and $\epsilon =  1/ \text{poly}(L)$, where $poly(L)$ is a polynomial in the lattice linear size $L$.
\end{thm}
\begin{skproof} 
	Let $T^{\Delta}$ be the triangulation upon which $H(\nu_{d+1})$ is defined. For simplicity of presentation, we assume the triangulation is translationally invariant on some scale (although non-essential, this allows us to use a lattice renormalization argument). We assume that for each hypercubic region of side-length $l$, there is a constant number $N_c = \calO(l^d)$ of vertices in $T^{\Delta}$. Since $T^{\Delta}$ has bounded degree (by assumption), we have that each vertex belongs to a constant number of $d$-simplexes. 
	
	Similarly to the two-dimensional case, we divide up the lattice into a hypercubic grid $\calP_l$ such that each hypercubic region has side-length $l=\left(c \log(L)\right)^\frac{1}{d}$ for some constant $c$. For each $\textbf{\textit{k}} \in \{0,1\}^N$, we can define an imperfect Hamiltonian $H(\textbf{\textit{k}})$. Let $\overline{\rho}(\textbf{\textit{k}})$ be the normalized, symmetric ground-space projector of $H(\textbf{\textit{k}})$. Any configuration that has at least one sink in every hypercubic region is called valid. By a straightforward generalization of Lemma \ref{validProb}, we can approximate $\rho(\beta)$, up to an error that is an inverse polynomial in the system size $L$ by a weighted combination of valid configurations $\overline{\rho}(\textbf{\textit{k}})$.
	
	Fix a valid configuration $\textbf{\textit{k}}$, and let $S$ be a subset of vertices containing precisely one sink in each hypercubic region. The goal is to find a sequence of Pachner moves taking $T^{\Delta}$ to a different triangulation $T^S$, whose vertex set is the chosen set of sinks $S$ (note that $T^S$ is not uniquely determined, but any choice will suffice). This sequence of Pachner moves gives a corresponding symmetric unitary $\calD_{\textbf{\textit{k}}}$, taking the imperfect Hamiltonian to a trivial Hamiltonian. In particular, let $T^{S}$ be any triangulation with vertices that are sinks and whose set of $k$-simplexes is labelled by $\Delta_k^{S}$ (see Fig. \ref{figPachnerShift2}). Then 
	\begin{equation}
	U_{\textbf{\textit{k}}}^{(S)}=\prod_{\sigma_d\in \Delta_d^S} {(U^{\sigma_d}_{\nu_{d+1}})}^{P(\sigma_d)}
	\end{equation}
	is a symmetric unitary that is supported entirely on the set of sinks and therefore the imperfect Hamiltonian $H_0(\textbf{\textit{k}})$ of the trivial model $H_0$ in Eq.~(\ref{trivcohoSPT}) is invariant under $U_S$. Then since $T^S$ and $T^{\Delta}$ are Pachner equivalent, there exists a symmetric unitary $\calD_{\textbf{\textit{k}}}$ such that
	\begin{equation}
	\calD_{\textbf{\textit{k}}} H(\textbf{\textit{k}}) \calD_{\textbf{\textit{k}}}^{\dagger} = U_{\textbf{\textit{k}}}^{(S)}H_0(\textbf{\textit{k}})U_{\textbf{\textit{k}}}^{(S)\dagger} = H_0(\textbf{\textit{k}}),
	\end{equation}
	from which it follows that $\calD_{\textbf{\textit{k}}}\overline{\rho}(\textbf{\textit{k}}) \calD_{\textbf{\textit{k}}}^{\dagger}$ is a trivial product state. 
	
	Now it only remains to determine an upper bound on the depth of the circuit $\calD_{\textbf{\textit{k}}}$ corresponding to this sequence of Pachner moves. We now describe a sequence of Pachner moves taking $T^{\Delta}$ to $T^S$ that upper bounds the depth of $\calD_{\textbf{\textit{k}}}$ by $\calO(\log(L) \log\log(L))$.
	
	The sequence of Pachner moves taking $T^{\Delta}$ to $T^S$ can be divided into two steps: a renormalization sequence, followed by a small vertex shifting. We present the argument in 2 dimensions, as the case for higher dimensions works analogously, where the 2 dimensional Pachner moves are replaced with the corresponding higher dimensional Pachner moves. The steps are depicted in Fig. \ref{figPachnerShift}. Note that we will keep track of the original vertices throughout, which we refer to as \textit{ambient vertices} (as they correspond to the original degrees of freedom). Any ambient vertices that are not part of new triangulations correspond to spins in the $\ket{+}$ state. 
	
	\begin{figure}[htb!]%
		\centering
		\subfloat[]{{\includegraphics[width=0.26\linewidth]{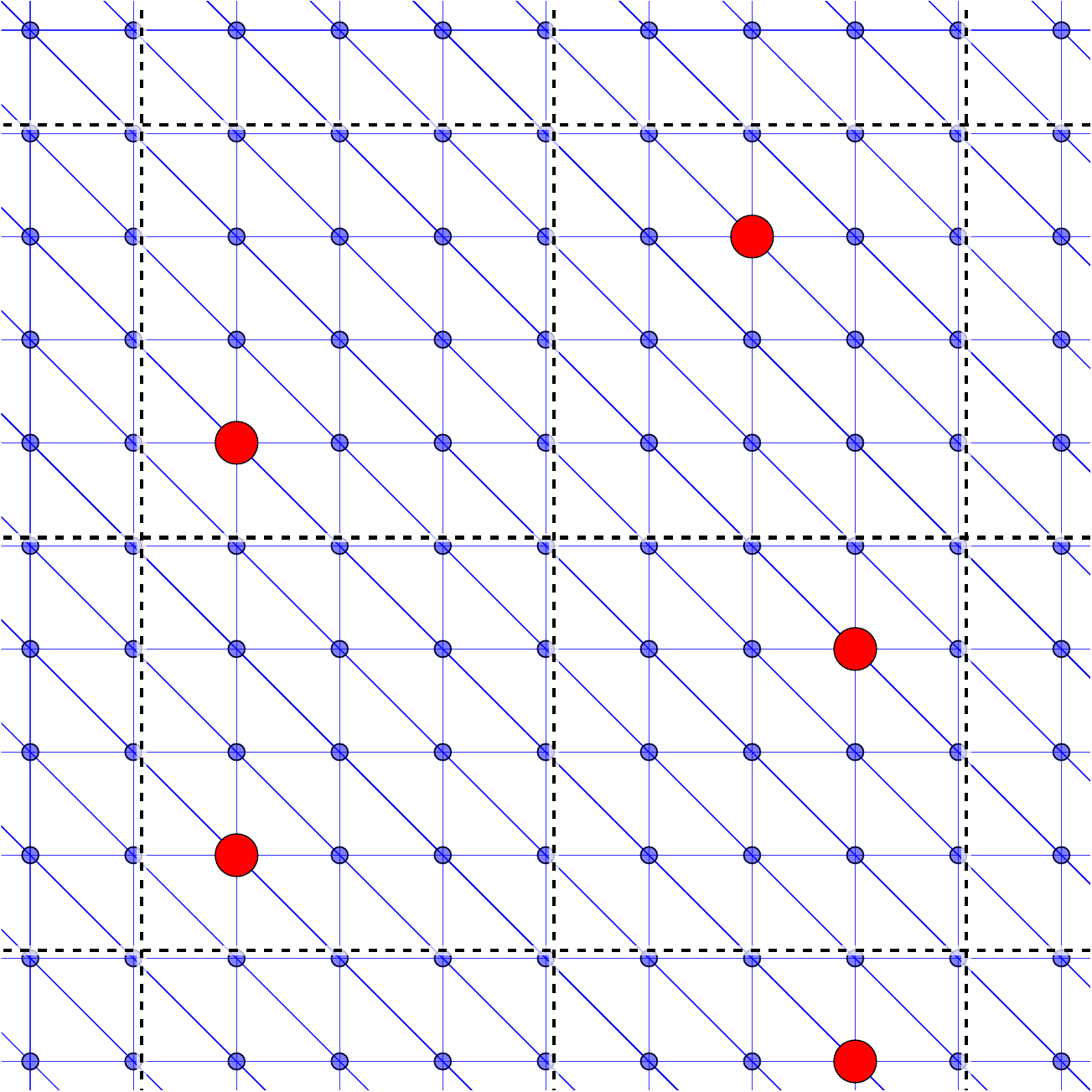} }\label{figPachnerShift0}}%
		\qquad
		\subfloat[]{{\includegraphics[width=0.26\linewidth]{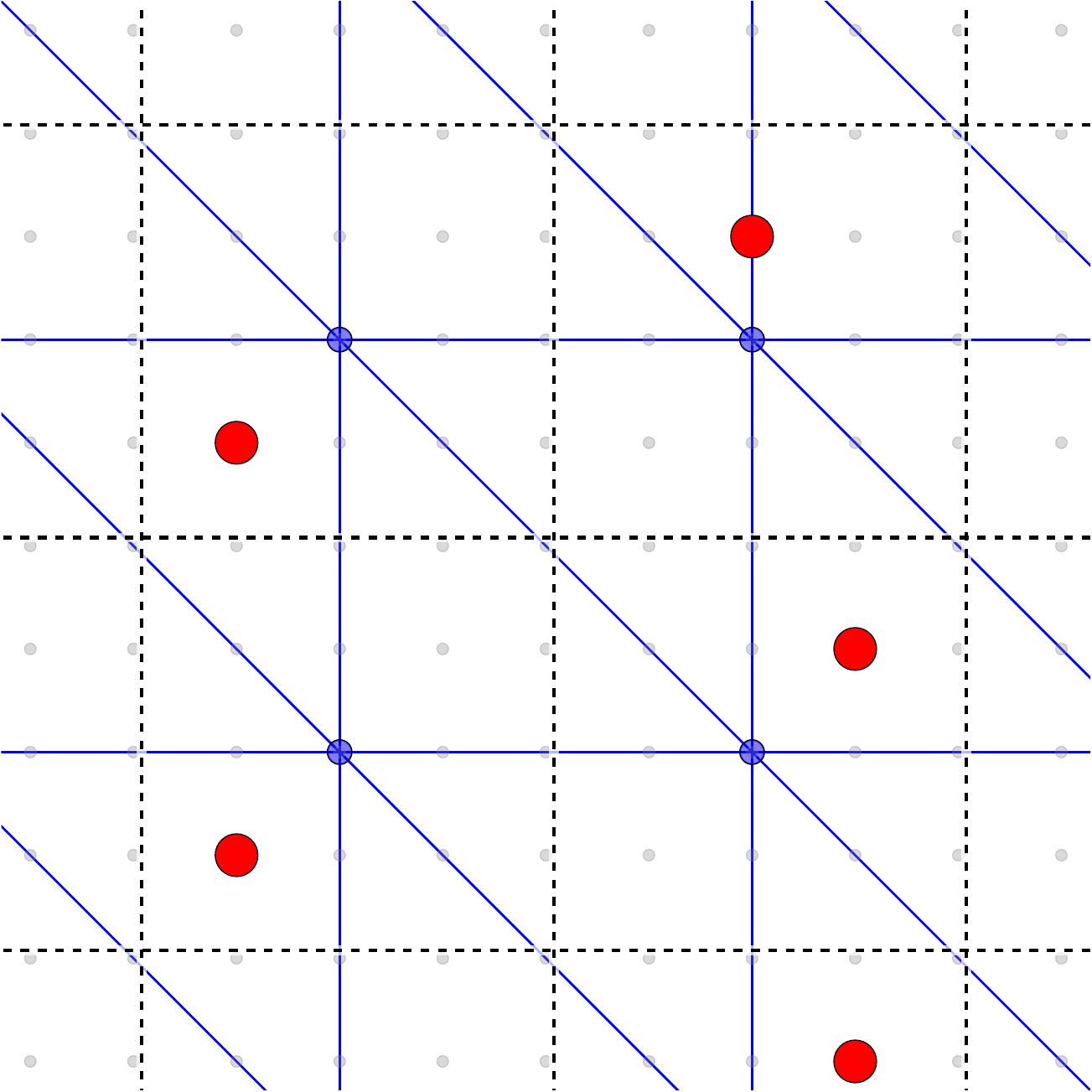} }\label{figPachnerShift1}}%
		\qquad
		\subfloat[]{{\includegraphics[width=0.26\linewidth]{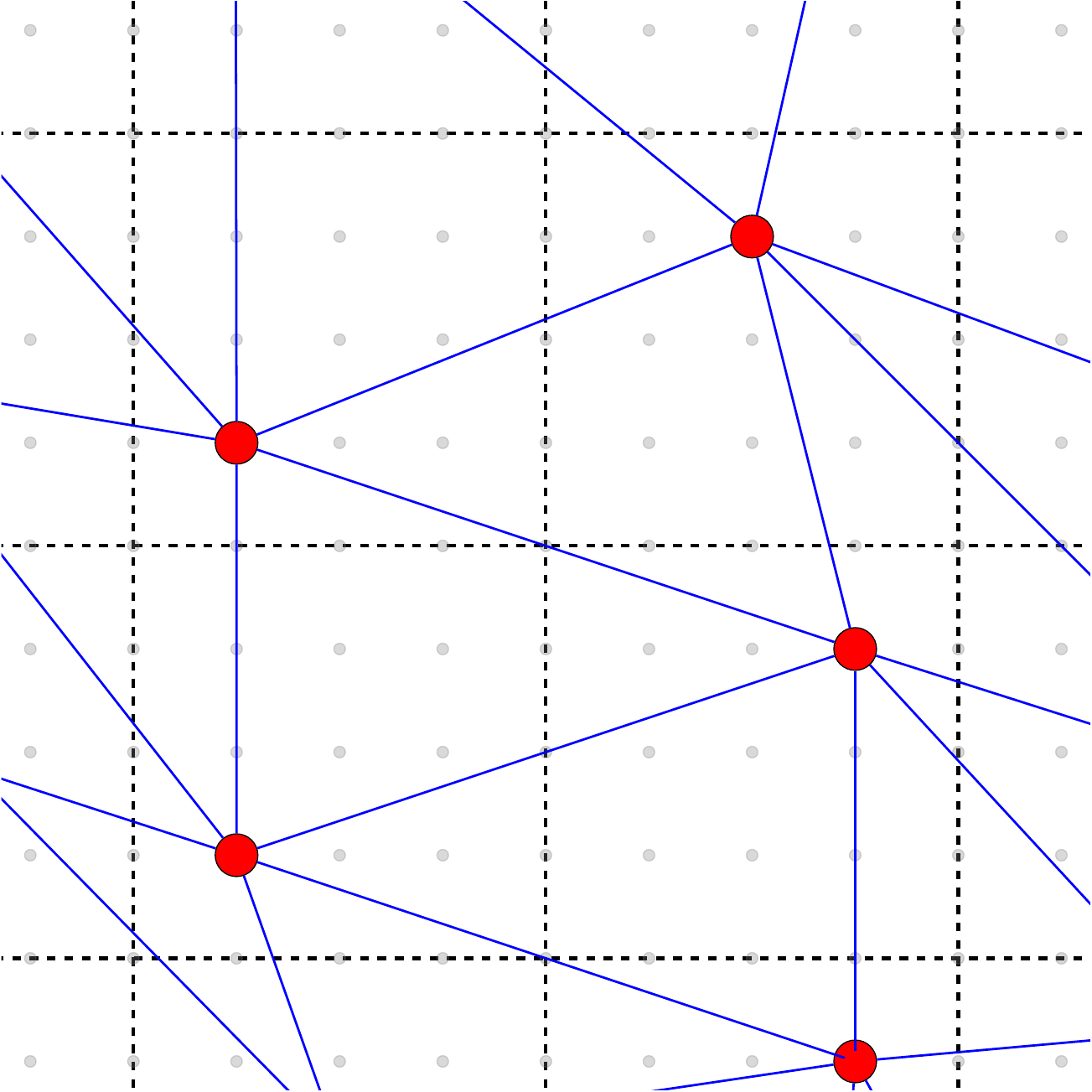} }\label{figPachnerShift2}}%
		\caption{(Color online) The two principle steps to take the original triangulation $T^{\Delta}$ in (a) to the triangulation $T^S$ in (c), whose vertices are all sinks (large dots). The first step is to renormalize the $T^{\Delta}$, resulting in the triangulation in (b). The second step is a vertex shifting, resulting in the triangulation in (c). The grid $\calP_l$ with side-lengths $(c\log(L))^{\frac{1}{d}}$ is displayed by the dashed lines. Faded grey nodes denote ambient vertices no longer part of the triangulation, which correspond to decoupled spins in the $\ket{+}$ state after the circuit $\calD$ has been applied. 
		}%
		\label{figPachnerShift}%
	\end{figure}
	
	Firstly, we perform a sequence of renormalization steps, which increases the original length of the edges in $T^{\Delta}$ from $\calO(1)$ to $l$, and in doing so reduces the number of vertices down to one per cubic region. Firstly, we claim that to renormalize the length of all edges by a factor of 2 takes a constant number of parallel Pachner moves. Indeed for a triangular lattice, it takes 12 parallel Pachner moves to scale the lattice by a factor of 2, as depicted in Fig. \ref{figPachnerRenormalization}. In general, the number of moves will be proportional to the maximum degree of a vertex. Since we wish to rescale the edge length to $l = (c\log(L))^{\frac{1}{d}}$, we need to do $\calO(\log \log(L))$ renormalization steps. Each Pachner move acts on a simplex of size at most $l^d = c\log(L)$, and therefore the depth of this Pachner sequence is $\calO(\log(L) \log \log(L))$. 
	
	Secondly, we need to transform the renormalized triangulation (depicted in Fig. \ref{figPachnerShift1}) to $T^{S}$. Since there is only one vertex per cubic region in the renormalized triangulation, this  process can be considered as a shifting of the vertices. This can be achieved by firstly reintroducing the sinks as vertices using the second Pachner move in Fig. \ref{figPachner}, then removing the remaining ambient vertices using a combination of Pachner moves\footnote{Note that we assume that we have a sufficiently large system such that there exists enough sinks to perform the required Pachner moves. This is without loss of generality as we are concerned with the scaling rather than small system details.}. Since Pachner moves in disjoint simplexes can be performed in parallel, the depth of this sequence is proportional to the degree which (by assumption) is bounded in the original triangulation, and therefore also the renormalized triangulation. Then as each move acts on a simplex of size $l^d = c\log(L)$ the depth of this sequence is $\calO(\log(L))$. 
	
	\begin{figure}
		\centering
		\subfloat[]{{\includegraphics[width=0.22\linewidth]{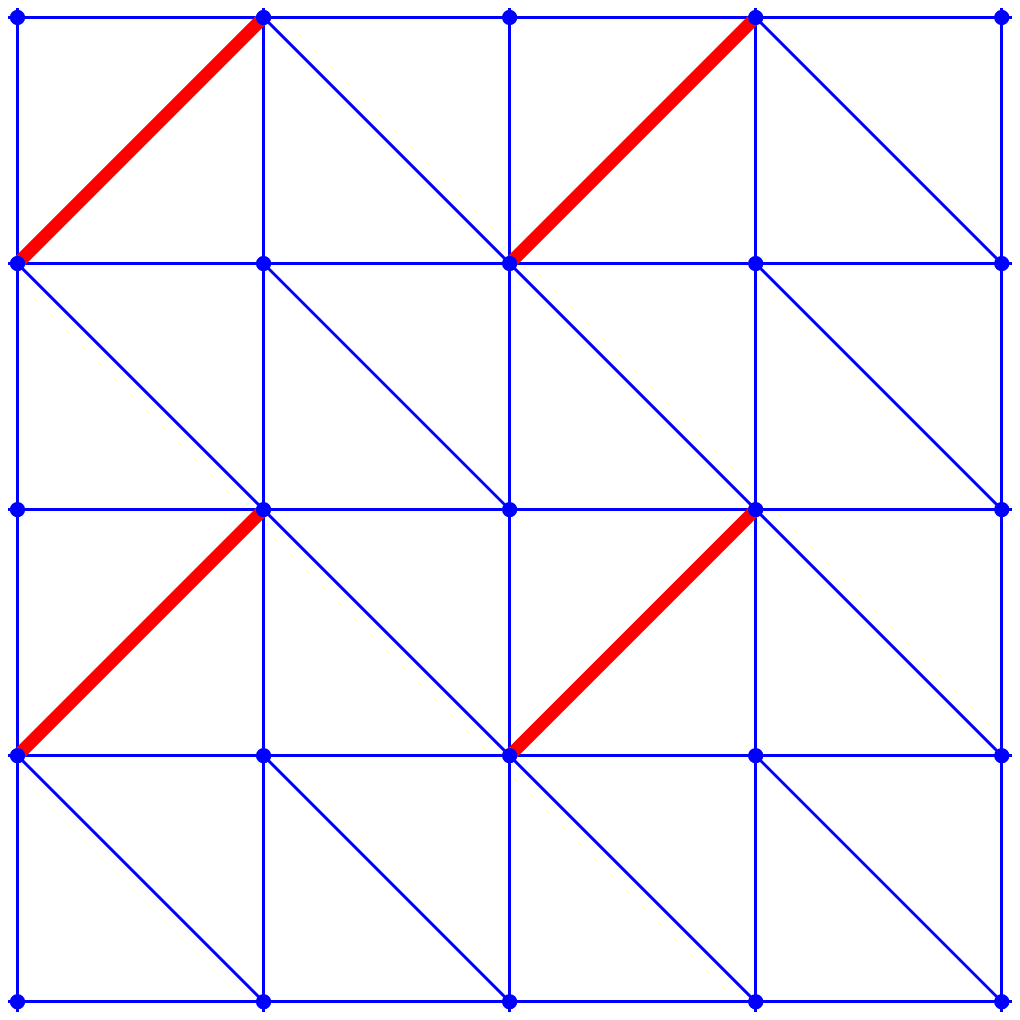} }\label{figPachnerSequence2}}%
		\qquad
		\subfloat[]{{\includegraphics[width=0.22\linewidth]{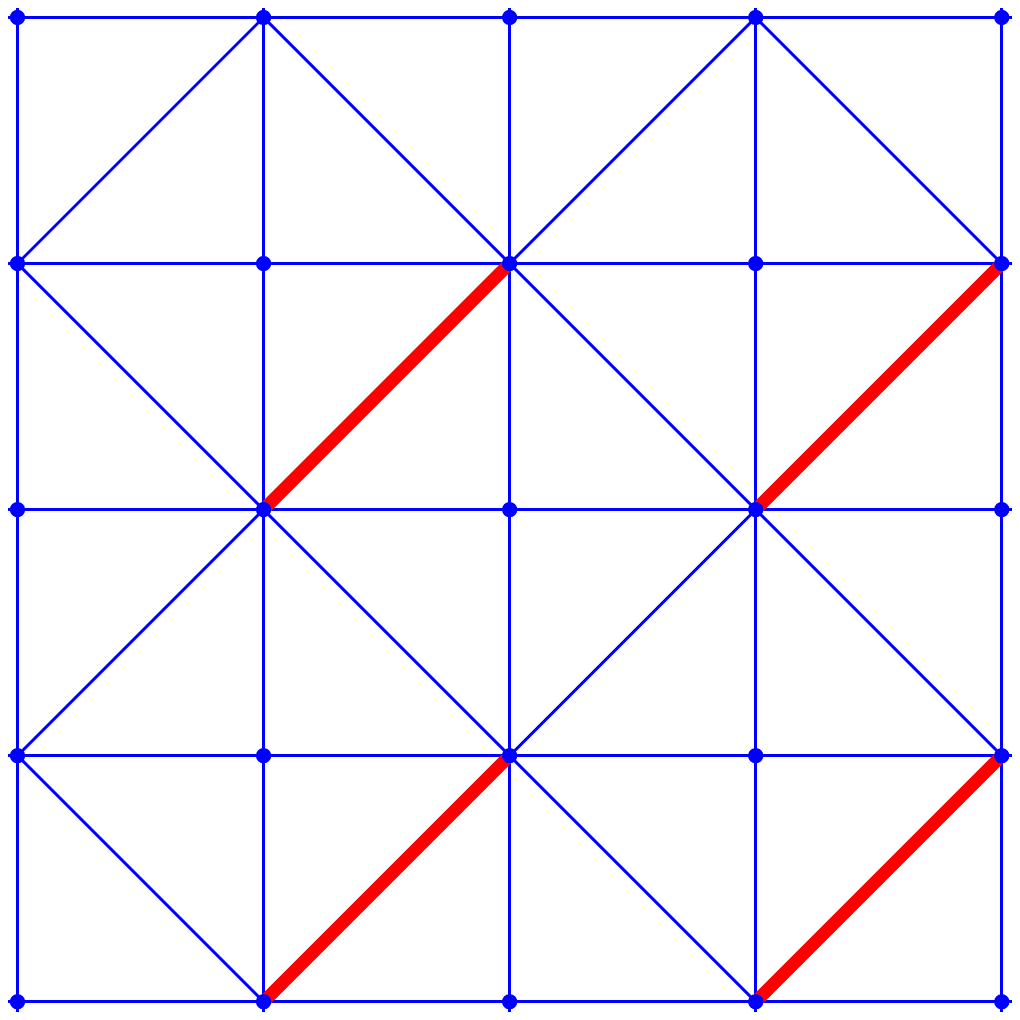} }\label{figPachnerSequence3}}%
		\qquad
		\subfloat[]{{\includegraphics[width=0.22\linewidth]{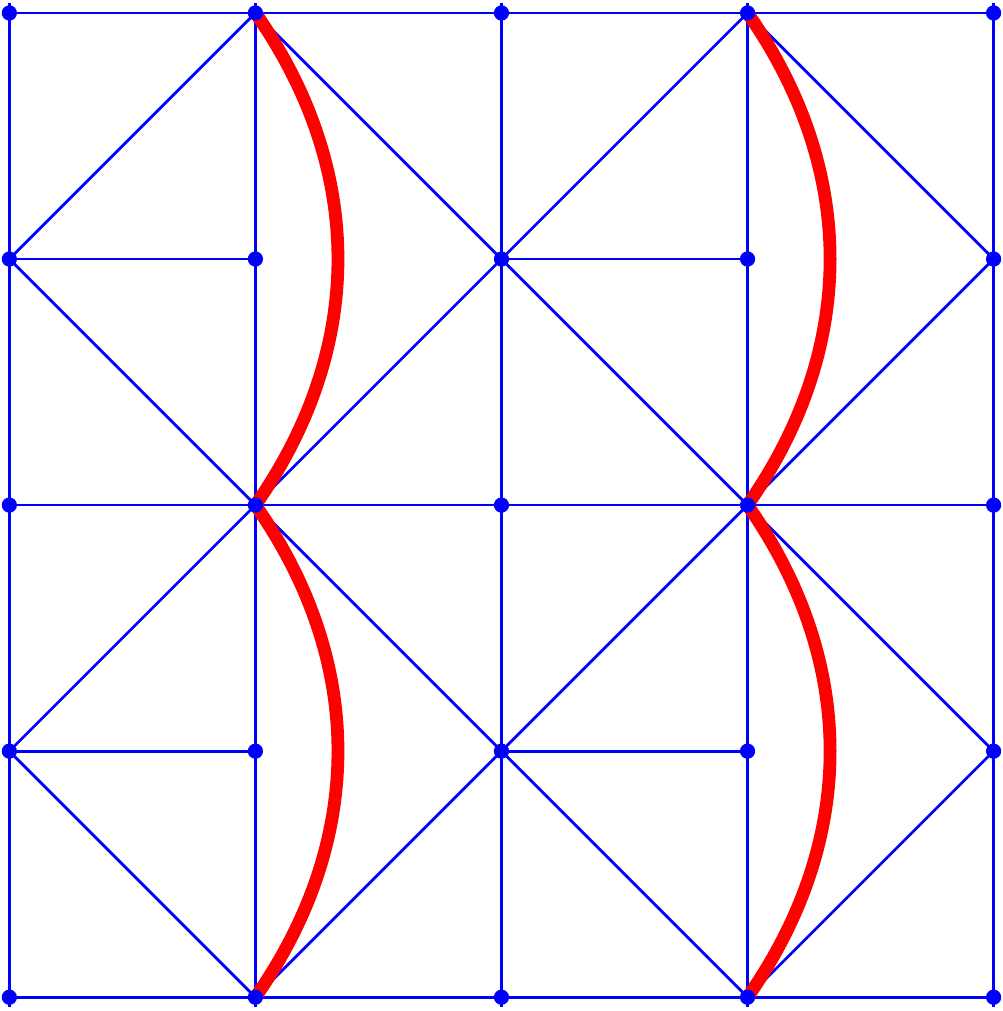} }\label{figPachnerSequence4}}%
		\qquad
		\subfloat[]{{\includegraphics[width=0.22\linewidth]{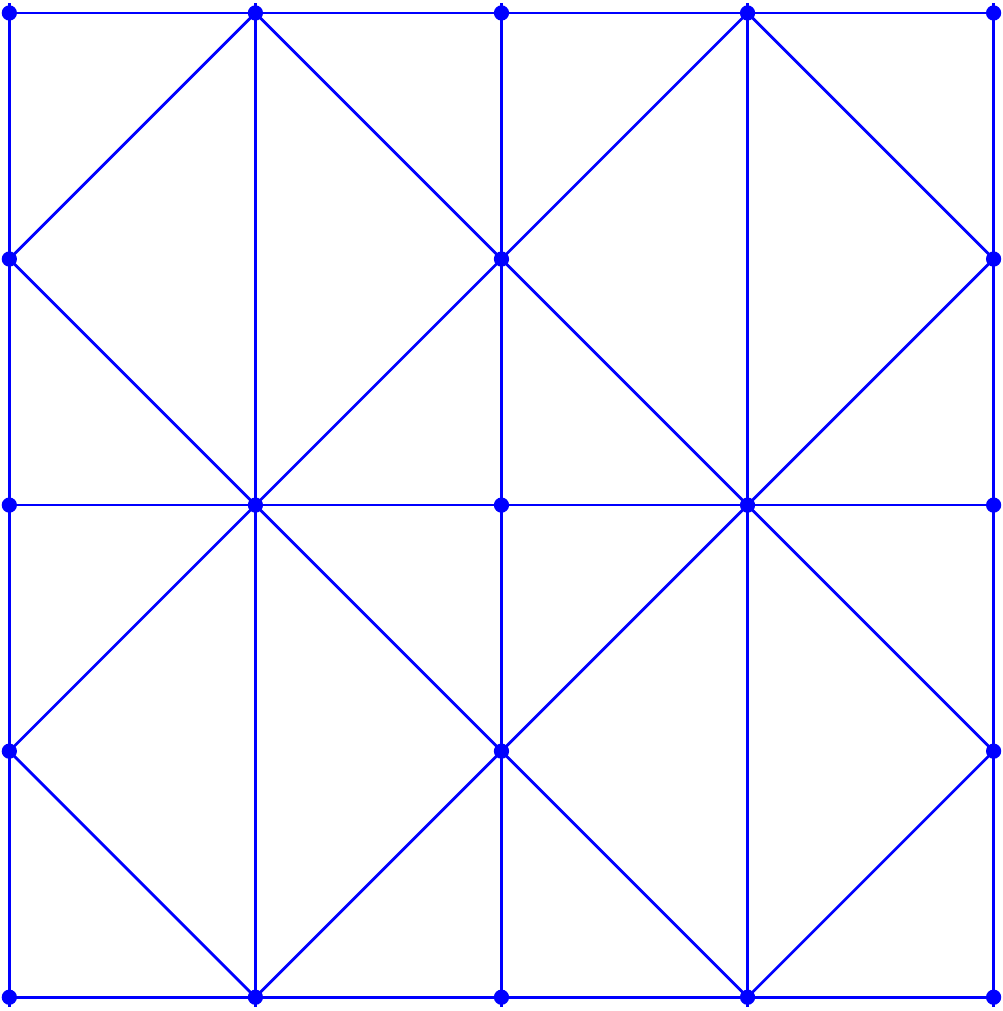} }\label{figPachnerSequence5}}%
		\qquad	
		\subfloat[]{{\includegraphics[width=0.22\linewidth]{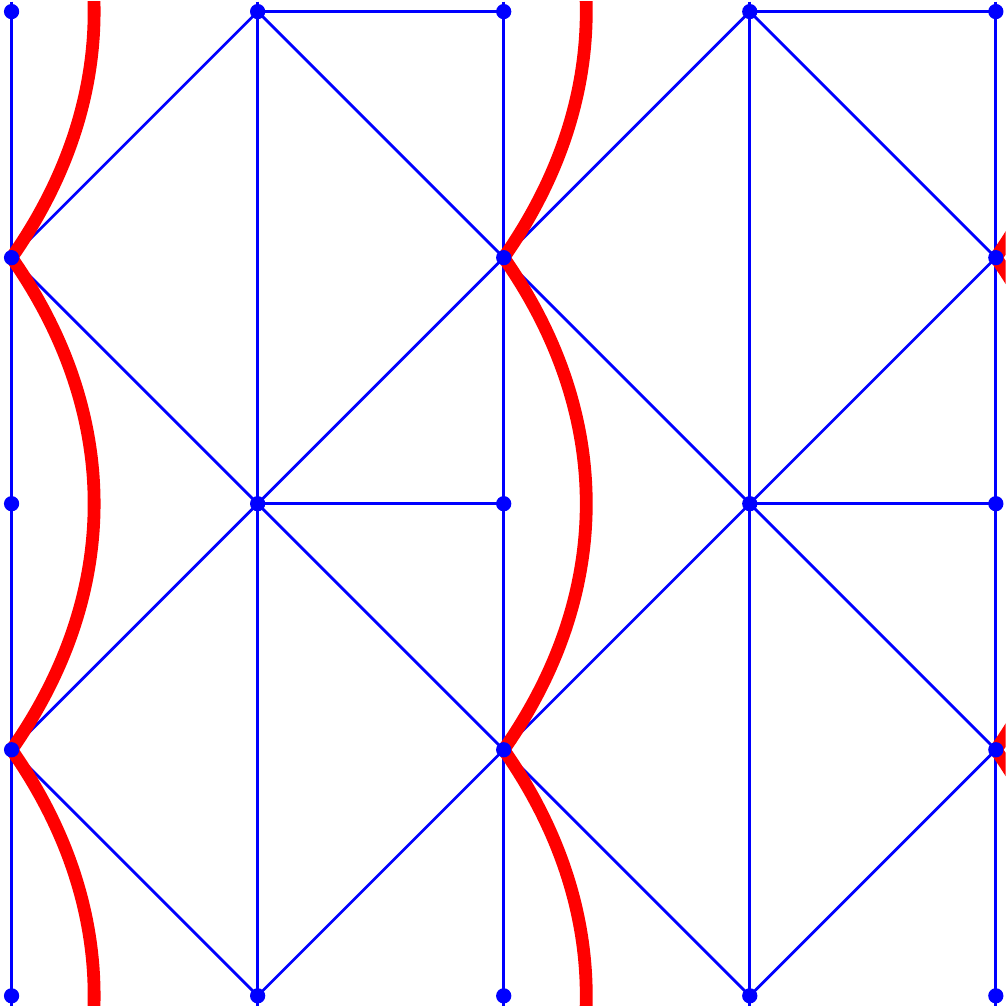} }\label{figPachnerSequence6}}%
		\qquad	
		\subfloat[]{{\includegraphics[width=0.22\linewidth]{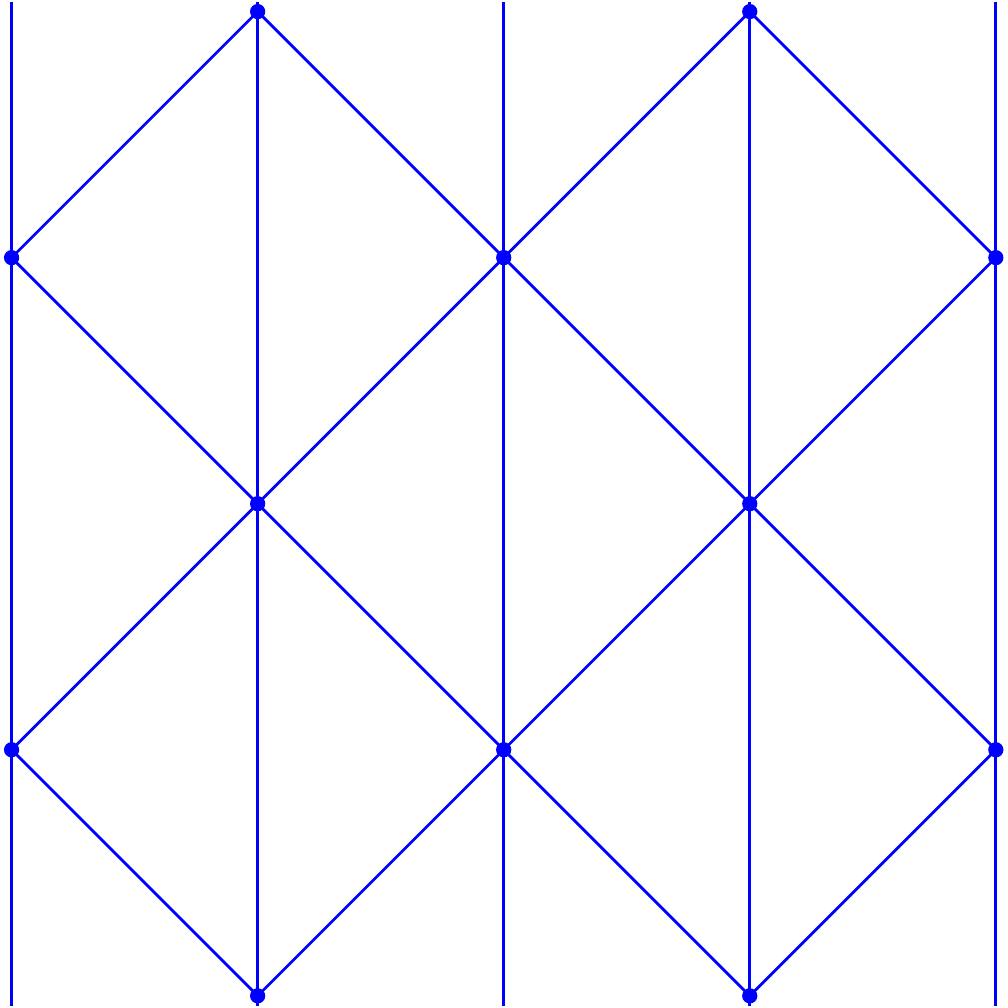} }\label{figPachnerSequence7}}%
		\caption{(Color online) The first six parallel Pachner moves for a single renormalization step that scales the edge lengths of the triangular lattice by a factor of two. New edges arising from Pachner moves are depicted by thick (red) lines. Notice that some of the edges are unchanged (namely the diagonal ones), but this process can be repeated to rescale them too. Sinks are not displayed in this figure as they do not yet play a role. This process can be repeated $\calO(\log\log(L))$ times to rescale $T^{\Delta}$ to the renormalized triangulation in Fig. \ref{figPachnerShift1}. 
		}%
		\label{figPachnerRenormalization}%
	\end{figure}
	
	Putting these two steps together we have the depth of the sequence of Pachner moves taking $T^{\Delta}$ to $T^S$ is $\calO(\log(L) \log\log(L))$. This sequence of Pachner moves gives rise to a symmetric circuit $\calD_{\textbf{\textit{k}}}$ taking $\overline{\rho}(\textbf{\textit{k}})$ to a trivial state.  Since this argument works for every valid configuration, we have that $\rho(\beta)$ is polynomially approximated by a sum of $(\calO(\log(L) \log\log(L)),0)$ SPT-trivial states, and therefore $\rho(\beta)$ is $(r,\epsilon)$  SPT-trivial, with $r=\calO(\log(L) \log\log(L))$, and $\epsilon = 1/poly(L)$. Note that the key ingredient in this sequence of Pachner moves is that the degree remains bounded at all stages, and therefore to disentangle any spins requires only a constant number of Pachner moves. The exponent of the $\log$ may be improved for example by keeping the sinks in the triangulation during the renormalization steps. 
\end{skproof}

\section{A model with a thermal SPT phase}\label{sec3}
Despite proving that a thermal SPT phase is impossible in models with only onsite symmetries, we now provide an example of a model with thermal SPT-order by enforcing a stronger, higher-form symmetry. The model we consider is the cluster model on a particular three-dimensional lattice, first introduced by Raussendorf, Bravyi and Harrington (RBH) \cite{RBH}, protected by a $\zz_2 \times \zz_2$ 1-form symmetry. While the discussion here is specific to the RBH model, the tools developed and the analysis is quite general, and can be extended to other higher form models.

Cluster states are well known within the quantum information community for their importance as a resource for measurement based quantum computation (MBQC) \cite{ClusterMBQC}. They can be defined on any graph or lattice, and their usefulness for computation is strongly dependent upon the underlying graph or lattice dimension \cite{QuantumWire, ClusterMBQC}. In the context of SPT phases, the 1D cluster model is known to belong to a nontrivial phase with a $\zz_2 \times \zz_2$ onsite symmetry \cite{1dCluster}, and states within this phase have been shown to be useful as quantum computational wires \cite{Delse1}. Additionally, certain states in 2D possessing SPT-order protected by onsite symmetries have been shown to be universal resources for MBQC \cite{MMMBQC, NWMBQC}.

From an information processing standpoint, the RBH model is very compelling. The model forms a basis for the topological MBQC scheme, a universal model of quantum computation with a very high threshold arising from topological considerations \cite{Rau06,TopoClusterComp}. We wish to understand the physical origin and underlying quantum order that underpins the high threshold of this scheme. We begin by reviewing the RBH model.

\subsection{The RBH model}

In order to present the RBH model, it will be helpful to review some homological terminology, which will allow us to specify all relevant operators and make the following analysis simpler. 
The lattice we consider is a cubic lattice $\calC$ of linear size $d$. For simplicity, we consider periodic boundary conditions in each direction such that $\calC$ has topology of a 3-torus. We label by $\Delta_3$, $\Delta_2$, $\Delta_1$, and $\Delta_0$ the set of all cubes, faces, edges, and vertices of $\calC$, respectively. Elements of $\Delta_k$ are called $k$-cells and denoted by $\sigma_k$ for $k \in \{0,1,2,3\}$. 

\subsubsection{Some homological notation}

The lattice $\calC$ naturally gives rise to a chain complex
\begin{equation}\label{eqChCo}
\xymatrix{
	C_3 \ar@<0.5ex>[r]^{\partial_3} & C_2 \ar@<+0.5ex>[r]^{\partial_2} & C_1 \ar@<+0.5ex>[r]^{\partial_1} &  C_0,
}
\end{equation}
which is a set of vector spaces $C_k$ and linear maps $\partial_k: C_k\rightarrow C_{k-1}$ between them called boundary maps, which we now define. Each vector space $C_k \equiv C_k(\calC;\zz_2)$  has elements consisting of formal sums of the basis elements $\sigma_k \in \Delta_k$ with coefficients from the field  $\zz_2$. A general vector $c_k$ in $C_k$ is called a $k$-chain, and can be uniquely written as $c_k = \sum_{\sigma_k \in \Delta_k} a(\sigma_k) \sigma_k,$ with $a(\sigma_k) \in \zz_2.$ Intuitively, a $k$-chain can be one-to-one identified with a subset of $k$-cells of $\Delta_k$, so a 3-chain $c_3\in C_3$ represents a subset of volumes (i.e. $c_3\subset \Delta_3$), a 2-chain represents a subset of surfaces, and so on. Between vector spaces $C_k$ we have the boundary map $\partial_k: C_k \rightarrow C_{k-1}$, defined on each basis element as
\begin{equation}
\partial_k (\sigma_k) = \sum_{\substack{\sigma_{k-1} \in \Delta_k \\ \sigma_{k-1} \subset \sigma_k}} \sigma_{k-1}
\end{equation}
and extended to an arbitrary $k$-chain by linearity. Here, the sum is over all $(k-1)$-cells $\sigma_{k-1}$ that are contained in $\sigma_k$. 

There are two important classes of chains known as cycles and boundaries. The \textit{$k$-cycle group} $Z_k = \text{ker}(\partial_k)$ is the vector space (which can be regarded as a group) consisting of $k$-chains that have no boundary. Elements of $Z_k$ are known as $k$-cycles. Similarly, the \textit{$k$-boundary group} $B_k = \text{im}(\partial_{k+1})$ is the vector space consisting of $k$-chains that are the boundary of a $(k+1)$-chain. Elements of $B_k$ are known as $k$-boundaries. Importantly, the boundary maps satisfy $\partial_{k-1} \circ \partial_k = 0$, which implies that every boundary is a cycle, but in general not every cycle is a boundary. A cycle that is not a boundary is referred to as nontrivial or noncontractible.

One can define the dual lattice $\calC^*$ of the cubic lattice $\calC$, which is obtained by replacing volumes by vertices, faces by edges, edges by faces, and vertices by volumes. The dual lattice $\calC^*$ is also a cubic lattice, but shifted with respect to the primal (initial) lattice. We can define a chain complex associated with the dual lattice in a similar way to Eq.~(\ref{eqChCo}), where $C_k^*$ are vector spaces with $k$-cells of the dual lattice as basis vectors, and corresponding boundary maps $\partial_k^*$. We denote the dual cycle groups by $Z_k^{*}$, and dual boundary groups by $B_k^{*}$. 

Since each $k$-chain corresponds to a unique dual-$(3-k)$-chain, the dual boundary map $\partial_k^* : C_k^* \rightarrow C_{k-1}^*$ can be thought of as a map $\partial_k^*: C_{3-k} \rightarrow C_{4-k}$. Namely, since any $(3-k)$-chain $c_{3-k}$ is dual to a unique dual-$k$-chain $c_{k}'$, we define $\partial_k^* c_{3-k}$ to be the unique $(4-k)$-chain dual to $\partial_k^* c_{k}'$. In the following, we suppress the subscript on the boundary and dual boundary maps, and we will freely apply the dual boundary map on both chains and dual chains using the previous correspondence. This allows us to regard 1-cycles and dual-1-cycles as closed loop-like subsets of the lattice $\calC$, and 2-cycles and dual-2-cycles as closed surface-like subsets of the lattice $\calC$.

\subsubsection{The RBH Hamiltonian}
With this terminology, we can now present the RBH model in a useful homological formulation. The Hilbert space can be constructed by placing a qubit on every 2-cell $\sigma_2 \in \Delta_2$ and every 1-cell ${\sigma}_1 \in \Delta_1$, which we will refer to as the primal and dual qubits respectively (we think of dual qubits as residing on the 2-cells of the dual lattice). The Hilbert space is given by $\calH = \calH_1 \otimes \calH_2$, where $\calH_1$ is the Hilbert space of the dual qubits, and $\calH_2$ is the Hilbert space of the primal qubits.

For a given $2$-chain $c_2 = \sum_{\sigma_2 \in \Delta_2} a(\sigma_2) \sigma_2,$ with $a(\sigma_2) \in \zz_2$, define the Pauli operator 
\begin{equation}
X(c_2) = \prod_{\sigma_2 \in c_2} X_{\sigma_2},
\end{equation}
where $X_{\sigma_2}$ is the Pauli $X$ supported on the qubit at $\sigma_2$. One can similarly define operators for Pauli $Z$ as well as for the dual qubits. A general Pauli operator $P$ then has the following decomposition
\begin{equation}
P = i^{\alpha} X(c_2)Z(c_2')X(c_1) Z(c_1'),
\end{equation}
for some $\alpha \in \{0,1,2,3\}$, 2-chains $c_2$, $c_2'$ and 1-chains $c_1$ and $c_1'$. One could equivalently decompose the operator $P$ in terms of dual chains.

In this notation we can now describe the RBH Hamiltonian on this lattice. The Hamiltonian is given by a sum of local, commuting (5-body) terms
\begin{equation}
H_{\calC} = -\sum_{\sigma_1 \in \Delta_1} K(\sigma_1) -\sum_{\sigma_2 \in \Delta_2} K(\sigma_2),
\end{equation}
where 
\begin{equation}\label{clusterTerms}
K(\sigma_1) = X(\sigma_1) Z(\partial^{*} \sigma_1), \quad \text{and} \quad K(\sigma_2) = X(\sigma_2) Z(\partial \sigma_2),
\end{equation}
as depicted in Fig.~\ref{figclusterlattice3}. We note that $K(\sigma_1)$ and $K(\sigma_2)$ are the standard cluster state stabilizer generators. The \textit{cluster state} $\ket{\psi_{\calC}}$ is the unique ground state of $H_{\calC}$ which is the $+1$-eigenstate of each of the cluster terms $ K(\sigma_1)$ and $K(\sigma_2)$. 

An alternative description in terms of a circuit description shows that the cluster state is short-range entangled. Consider the circuit $\calU_{CZ}$ comprised of controlled-$Z$ gates between every neighbouring primal and dual qubit
\begin{equation}
\calU_{CZ} = \prod_{\substack{\sigma_1 \in \Delta_1 \\ \sigma_2 \in \Delta_2}} \Bigg(\prod_{\sigma_1'\in\partial\sigma_2}CZ_{(\sigma_2,\sigma_1')}\Bigg) \Bigg(\prod_{\sigma_2'\in\partial^{*}\sigma_1}CZ_{(\sigma_1,\sigma_2')}\Bigg),
\end{equation}
where the controlled-$Z$ operator is defined in Eq. (\ref{EqCZ}). One can confirm that 
\begin{equation}\label{hamTriv}
 \calU_{CZ}^{\dagger} H_{\calC} \calU_{CZ} = -\sum_{\substack{\sigma_1 \in \Delta_1}}X(\sigma_1)-\sum_{\substack{\sigma_2 \in \Delta_2  }}X(\sigma_2)=:H_X.
\end{equation}
From this relation we see that the cluster state can be prepared from a product state by the circuit $\calU_{CZ}$, as 
\begin{equation}
\ket{\psi_{\calC}} = \calU_{CZ} \ket{+}^{\otimes |\Delta_2 \cup {\Delta}_1|},
\end{equation}
where $\ket{+}$ is the $+1$-eigenstate of Pauli $X$. Since $\calU_{CZ}$ can be represented by a constant depth quantum circuit, the cluster state is short-range entangled. We now proceed to identify a 1-form $\zz_2 \times \zz_2$ symmetry of the model and show that $\ket{\psi_{\calC}}$ resides in a nontrivial SPT phase at zero temperature when this symmetry is enforced. 

\subsubsection{1-form symmetry}
The cluster state is a short-range entangled state and so in the absence of a symmetry it belongs to the SPT-trivial phase. One can show that with only an onsite symmetry, this model remains in the SPT-trivial phase\footnote{Indeed, for a cluster state in any dimension $D\geq 2$, with onsite symmetry, one can generalise the two-dimensional result of \cite{MMMBQC} and construct a disentangling circuit involving symmetric gates comprised of controlled-$Z$ operations.}. We introduce a $\zz_2 \times \zz_2$ 1-form symmetry of the model and show that the cluster state is in a nontrivial SPT phase when this symmetry is enforced. Formally, we have two copies of a $\zz_2$ 1-form symmetry: one for each lattice (primal and dual).  The symmetry actions are given by a unitary representation $S$ of the 2-boundary and dual-2-boundary groups as
\begin{align}\label{1formsym}
S(b_2) := X(b_2), \quad S(b_2') := X(b_2'),
\end{align}
for any 2-boundary $b_2\in B_2$ and dual-2-boundary $b_2'\in B_2^*$. Any 2-boundary or dual-2-boundary corresponds to a closed, two-dimensional surface $\calM$ of the primal or dual lattice, respectively. The 1-form symmetry can therefore be viewed as being imposed by symmetry operators supported on qubits residing on closed, contractible two-dimensional submanifolds of $\calC$.

A local, generating set of symmetry operators is given by the following elementary operators 
\begin{equation}\label{genset}
\tilde{G} = \{S(\partial \sigma_3), S(\partial^{*}\sigma_0) ~|~ \sigma_3 \in \Delta_3, \sigma_0 \in \Delta_0 \},
\end{equation}
which are all 6-body. For example, an elementary 1-form operator $S(\partial \sigma_3)$ is supported on the surface of a single cube as depicted in Fig.~\ref{fig1FormSym}. Multiplying two neighbouring symmetry operators $S(b_2)S(\overline{b}_2) = S(b_2 + \overline{b}_2)$ can be viewed as gluing together the pair of surfaces that they correspond to. We conclude that the symmetry is a representation of the boundary groups $B_2 \times {B}_2^{*}$.

\begin{figure}[htb!]%
	\centering
	\subfloat[]{{\includegraphics[width=0.33\linewidth]{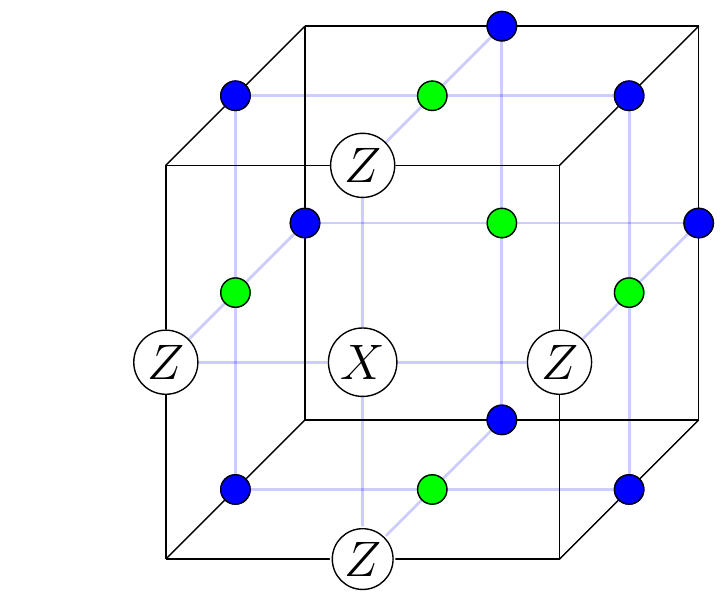} }\label{figclusterlattice3}}%
	\qquad
	\subfloat[]{{\includegraphics[width=0.33\linewidth]{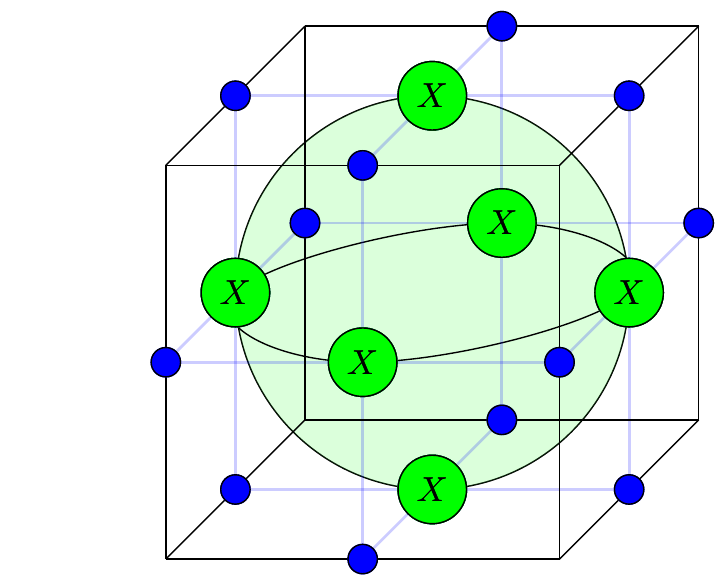} }\label{fig1FormSym}}%
	\caption{(Color online) (a) A unit cell of the cluster lattice $\calC$ with a single cluster term $K(\sigma_2)$. (b) An elementary 1-form operator $S(\partial \sigma_3)$. The primal qubits are depicted as light (green) circles and the dual qubits are depicted in dark (blue) circles. }%
	\label{figcellsym}%
\end{figure}

An important feature of the 1-form symmetry operators is that they can be expressed as products of cluster terms
\begin{equation}
S(b_2) = \prod_{\sigma_2 \in b_2} K(\sigma_2)\quad \text{and} \quad S(b_2') =  \prod_{\sigma_1^* \in b_2'} K(\sigma_1),
\end{equation}
where the second product is over all 1-cells $\sigma_1$ whose dual belong to $b_2'$. For example, this is easily verified for the elementary 1-form operator in Fig.~\ref{figcellsym}. It follows that these operators commute with $H_{\calC}$, and thus are symmetries of the cluster model. Additionally, the cluster state is a $+1$-eigenstate of these symmetry operators. Interestingly, such operators arise naturally in the context of topological MBQC and error correction \cite{TopoClusterComp,RBH} and we will return to this connection in the following section. 

\subsubsection{Thermal state of the 1-form symmetric RBH model}

We now consider the symmetric Gibbs state of the RBH model Hamiltonian $H_{\calC}$. In the presence of the 1-form symmetry, excitations in the RBH model take the form of one-dimensional, loop-like objects, which can be seen as follows. Excitation operators can be constructed out of Pauli-$Z$ operators, but the 1-form symmetry demands they form closed loops in the following way. Consider the operator $Z(c_1)$ for any 1-chain $c_1\in C_1$. This operator anti-commutes with cluster terms along the cycle
\begin{equation}
\{K({\sigma}_1), Z(c_1)\} = 0 \quad \iff \quad {\sigma}_1 \in c_1,
\end{equation} 
and will commute with the 1-form symmetry operators if and only if it has no boundary $\partial c_1 = 0$. Therefore excitation operators on the dual lattice are given by $Z(\gamma)$ where $\gamma \in Z_1$ is a 1-cycle. Similarly, excitation operators on the primal lattice are given by $Z(c_1')$, for any dual-1-cycle $c_1' \in Z_1^{*}$. Recall, 1-cycles and dual-1-cycles look like loop-like objects, and example excitation operators are shown in Fig.~\ref{figExcitations}.

A general symmetric excitation is given by $\ket{\psi(\gamma, \gamma')} = Z(\gamma)Z(\gamma') \ket{\psi_{\calC}}$ with $\gamma \in Z_1$, $\gamma' \in Z_1^*$, and the energy cost of introducing this excitation is $E{(\gamma, \gamma')} = 2 (|\gamma|+|\gamma'|)$. Notice that excitations created by Pauli $X$ operators can be converted into the above form, since they are equivalent up to products of cluster terms (of which the cluster state is a $+1$-eigenstate). As such, excited states are in one-to-one correspondence with elements of the 1-cycle and dual-1-cycle groups $Z_1 \times {Z}_1^{*}$.

The symmetric Gibbs state under this 1-form symmetry is given by a distribution over loop configurations 
\begin{equation}\label{symgibbscluster}
\rho_{\calC}(\beta) = \sum_{(\gamma, \gamma')\in Z_1 \times {Z}_1^{*}} \text{Pr}_{\beta}(\gamma, \gamma') \ket{\psi(\gamma, \gamma')}\bra{\psi(\gamma, \gamma')},
\end{equation}
where the sum is over all primal and dual 1-cycles, and 
\begin{equation}\label{Pgamma}
\text{Pr}_{\beta}(\gamma, \gamma') = \frac{1}{\calZ}e^{{-\beta E{(\gamma, \gamma')}}}, \quad \text{where} \quad \calZ = \sum_{(\gamma, \gamma') \in Z_1 \times Z_1^*} \text{Pr}_{\beta}(\gamma, \gamma').
\end{equation} 
Here, $\text{Pr}_{\beta}(\gamma, \gamma')$ is a probability distribution over loop-like configurations. In the following subsection, we show that this ensemble has nontrivial SPT-order under the 1-form symmetry, using a duality map known as gauging. Then in the subsequent section, we will provide a proof of the nontrivial SPT-ordering of the thermal state using a set of non-local order parameters as witnesses of the SPT-order.  

\begin{figure}%
	\centering
	\includegraphics[width=0.27\linewidth]{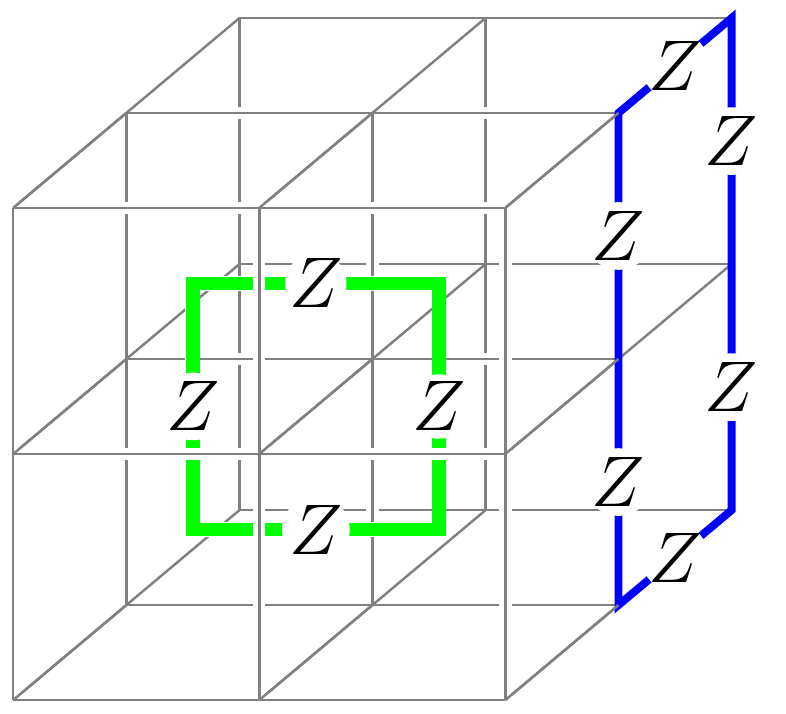} 
	\caption{(Color online) Examples of excitation operators. A 1-cycle is depicted by the thiner dark (blue) line, while a dual-1-cycle is depicted by the thicker light (green) line.}%
	\label{figExcitations}%
\end{figure}

\subsection{SPT-order of the RBH model}
We now show that the RBH model possesses nontrivial SPT-order under the 1-form symmetry by means of a duality map known as \textit{gauging}. Gauging is a procedure widely used throughout the study of many-body physics \cite{LevinGu, HWgauge, JHgauge, DWgauge}, and has recently found application in the study of fault-tolerant logical gates in topological quantum codes \cite{BYhigher, BYFT}. Gauging is the process of transforming a global symmetry $G$ into a local symmetry by minimally coupling the system to gauge fields. We will use an argument originally proposed by Levin and Gu \cite{LevinGu} that two Hamiltonians must belong to distinct SPT phases if the gauged versions belong to distinct topological phases. 

We will take the approach of \cite{BYhigher,BYFT} and view the gauging procedure as a duality map between SPT-ordered Hamiltonians and topologically ordered Hamiltonians, a correspondence known to hold for many models \cite{HWgauge}. By showing that the gauged RBH model belongs to a different phase than the gauged trivial model, we can deduce that the RBH model belongs to a nontrivial SPT phase. Furthermore, thermal stability of the SPT-order can be demonstrated by showing that the RBH cluster state corresponds to a nontrivial gapped domain wall in the 4D toric code, which is known to have thermally stable topological order~\cite{dennisTopo}.

\subsubsection{Gauging the 1-form symmetry}\label{secGauging}

We now outline the procedure of gauging the $\zz_2 \times \zz_2$ 1-form symmetry. More details of gauging models possessing higher-form symmetries can be found in \cite{BYFT}. We start with a basis for the primal and dual Hilbert spaces $\calH_1$ and $\calH_2$ given by vectors of the 1-chain and 2-chain groups respectively. For any 1-chain $c_1 \in C_1$, we can uniquely specify a computational basis state
\begin{equation}
c_1 = \sum_{\sigma_1\in \Delta_1}a(\sigma_1) \sigma_1, \quad \implies \quad \ket{c_1} = \ket{ \{a(\sigma_1)\}},
\end{equation}
where $a(\sigma_1) \in \zz_2$. A similar identification holds for the computational basis states in $\calH_2$ and the 2-chain group. The gauging map $\calG$ on the level of states takes states in $\calH_1$ to $\calH_2$, and states in $\calH_2$ to $\calH_1$ and can be concisely defined by the boundary and dual boundary maps, as follows. On the computational basis, the map $\calG: \calH_1 \otimes {\calH_2} \rightarrow \calH_2 \otimes {\calH_1}$ is defined by
\begin{align}
\calG(\ket{c_1} \otimes \ket{c_2}) = \ket{\partial^{*} c_1} \otimes \ket{\partial{c_2}},
\end{align}
and extended to $\calH=\calH_1 \otimes \calH_2$ by linearity. For example, on a computational basis state, $\calG$ is depicted in Fig.~\ref{figGaugeMap}. 
\begin{figure}[h]%
	\centering
	\includegraphics[width=0.6\linewidth]{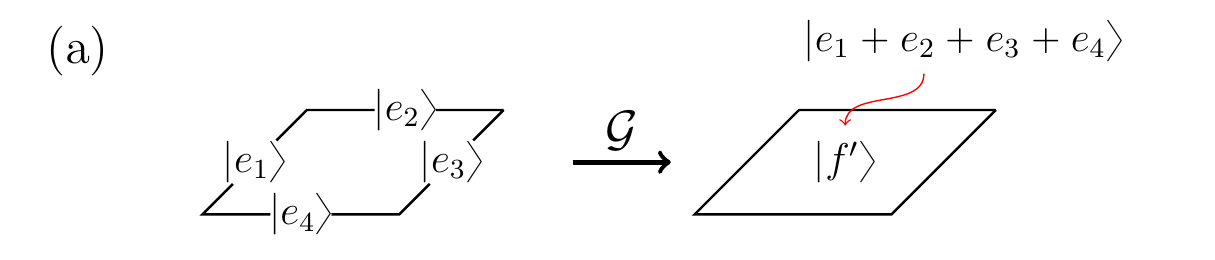}
	\qquad
	\includegraphics[width=0.6\linewidth]{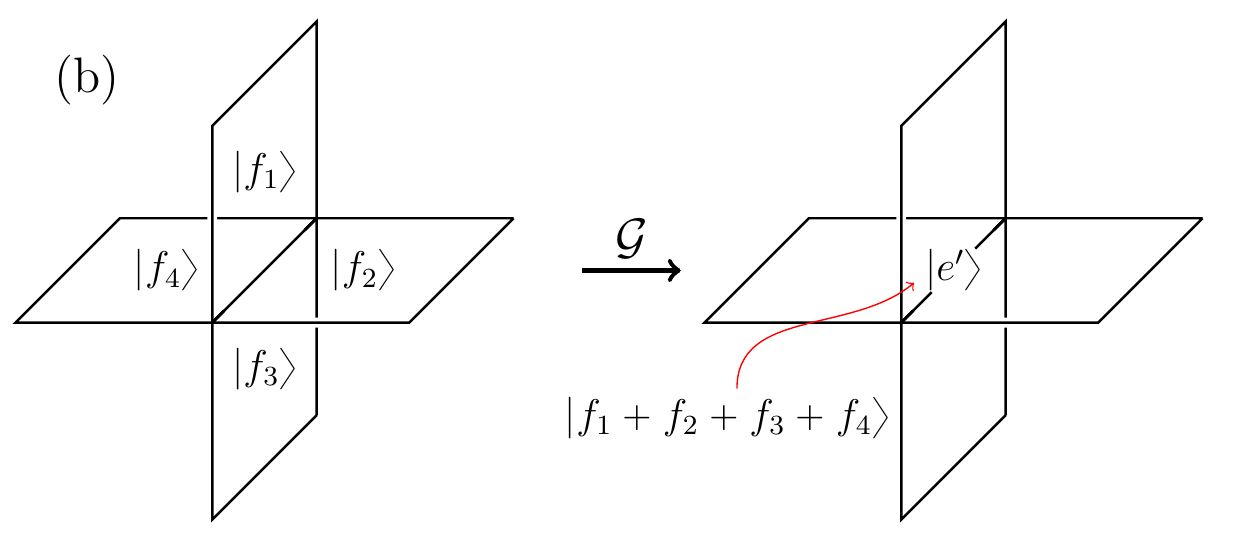} 
	\caption{(Color online) The gauging map on computational basis states. (a) States on the dual sublattice map to states on the primal sublattice. (b) States on the primal lattice map to states on the dual sublattice. The sums are performed mod 2.}%
	\label{figGaugeMap}%
\end{figure}

Importantly, any state $\ket{\psi_{\calG}}$ in the image of $\calG$ satisfies the gauge symmetry condition
\begin{equation}\label{gaugeSym}
Z({z}_2) Z(z_2') \ket{\psi_{\calG}} =  \ket{\psi_{\calG}},
\end{equation}
for any 2-cycle ${z}_2$ and dual-2-cycle $z_2'$. These gauge symmetry operators are similar to the 1-form operators in the RBH model, only they are now in the Pauli-$Z$ basis, and there are additional gauge symmetry operators for nontrivial and 2-cycles dual-2-cycles. Since higher-form symmetries can be viewed as gauge symmetries in a dual description, the distinction between the two types of symmetries is not a definitive one. In this paper, we treat higher form symmetries as symmetries which exist before the gauging map, and gauge symmetries as those which emerge after the gauging map. 

The gauging map $\calG$ can be extended to a map on symmetry respecting operators. For any symmetric operator $A$, the gauged operator $A'$ is defined implicitly by the following equation
\begin{equation}\label{gaugeIdentity}
\calG (A \ket{\psi}) = A' \calG(\ket{\psi}).
\end{equation}
Importantly, the 1-form symmetry operators are mapped to the identity. Note that $A'$ is only defined up to gauge symmetry operators in Eq. (\ref{gaugeSym}). One can use Eq. (\ref{gaugeIdentity}) to verify that gauging the trivial Hamiltonian $H_X$ of Eq. (\ref{hamTriv}) gives the following Hamiltonian
\begin{equation}
H_{X}^{(\calG)} = -\sum_{\sigma_1 \in \Delta_1} X(\partial^{*} \sigma_1) -\sum_{\sigma_2 \in \Delta_2} X(\partial \sigma_2).
\end{equation}
Since the gauged Hilbert space satisfies the gauge symmetry condition in Eq. (\ref{gaugeSym}), one can add $Z$-type terms $Z(\partial \sigma_3)$ and $Z(\partial^*\sigma_0)$ to the gauged Hamiltonian $H_{X}^{(\calG)}$ to fix out the gauge invariant ground space. Therefore, gauging the trivial Hamiltonian gives rise to two decoupled three-dimensional toric code Hamiltonians with qubits on faces and edges, respectively. Each of the toric codes belong to a nontrivial (intrinsic) topologically ordered phase at zero temperature. 

On the other hand, gauging the RBH Hamiltonian gives
\begin{equation}
H_{\calC}^{(\calG)} = -\sum_{\sigma_1 \in \Delta_1} K^{(\calG)}(\sigma_1) -\sum_{\sigma_2 \in \Delta_2} K^{(\calG)}(\sigma_2),
\end{equation}
where 
\begin{equation}
K^{(\calG)}(\sigma_1) = Z(\sigma_1) X(\partial^{*} \sigma_1), \quad K^{(\calG)}(\sigma_2) = Z(\sigma_2) X(\partial \sigma_2).
\end{equation}
This is equivalent to the original RBH Hamiltonian up to a Hadamard transformation $H^{\otimes |\Delta_1 \cup \Delta_2|}$, where $H$ is the Hadamard gate, exchanging the Pauli $X$ and $Z$ operators. Therefore the ground state of $H_{\calC}^{(\calG)}$ remains short-range entangled. As $\calG$ is locality preserving and gap preserving, the inequivalence of the two gauged models shows that the RBH model belongs to a nontrivial SPT phase under the 1-form symmetry. 

\subsubsection{Gapped domain wall at nonzero temperature}\label{secGappedWall}

An interesting and perhaps surprising application of the classification of SPT-ordered phases is in the construction of gapped domain walls in topological models \cite{BYFT}.  Here, we show that the RBH model with $\zz_2 \times \zz_2$ 1-form symmetry can be used to construct a nontrivial domain wall in two copies of the four-dimensional toric code. The fact that the domain wall implements a nontrivial automorphism of the excitation labels in the 4D toric codes demonstrates that the RBH model has nontrivial SPT-order at zero temperature~\cite{BYCCSPT,BYhigher}. We will in addition use this argument to demonstrate that the RBH model with 1-form symmetry retains its SPT-order at nonzero temperature, by leveraging the thermal stability of the 4D toric code. 

To illustrate this procedure, let us first consider the simpler case of a two-dimensional system with $\zz_2 \times \zz_2$ $0$-form symmetry. Namely, consider a square lattice $\Lambda$ with boundary and place qubits on vertices of $\Lambda$. Qubits can be labelled by one of two colors in such a way that neighbouring qubits are of different colors. We consider a system consisting of a trivial Hamiltonian in the bulk and the cluster state Hamiltonian on the boundary:
\begin{align}
H_{0} = - \sum_{u \in \text{bulk}(\Lambda)} X_{u} + H^{1D}_{\text{cluster}},
\end{align}
where $H^{1D}_{\text{cluster}}$ consists of terms supported on the boundary of $\Lambda$ in the following way,
\begin{equation}
H^{1D}_{\text{cluster}} = -\sum_{j \in \partial(\Lambda)} Z_{j-1}X_jZ_{j+1},
\end{equation} 
and the sum is over qubits on the boundary (which have been given a linear ordering). 

The whole Hamiltonian has a $\zz_2 \times \zz_2$ 0-form symmetry, generated by tensor product of Pauli $X$ on each sublattice of a given color. One can apply the gauging map to obtain a gauged Hamiltonian which possesses intrinsic topological order with gapped boundary. In this example, we will have two copies of the toric code with twisted gapped boundaries, where the two copies of the toric code are coupled by terms acting on the boundary. On this gapped boundary, pairs of point-like excitations $e_{1}m_{2}$ and $e_{2}m_{1}$ may condense, where $e_{i}$ and $m_{i}$ ($i=1,2$) represent electric charges and magnetic fluxes from each copy of the toric code. The $e_{i}$ and $m_{i}$ excitations correspond to violated $X$-type and $Z$-type stabilisers respectively, and occur at the end of strings of $Z$-type and $X$-type operators respectively.

By unfolding the lattice (see Fig.~\ref{figGaugeWall}, also Ref.~\cite{Kubica15}) one can view this gapped boundary as a gapped domain wall connecting two copies of the toric code. Upon crossing this domain wall, anyonic excitations are exchanged in the following manner:
\begin{align}
e_1 \leftrightarrow m_2, \qquad m_1 \leftrightarrow e_2.
\end{align}
Since this is a nontrivial automorphism of excitation labels, the cluster state cannot be prepared by a low depth quantum circuit as detailed in~\cite{BYCCSPT}. Gapped domain walls in higher-dimensional topological phases of matter can be also constructed from $0$-form SPT phases, leading to explicit construction of gapped domain walls in the higher-dimensional generalizations of the quantum double model.

\begin{figure}[htb!]%
	\centering
	\includegraphics[width=0.85\linewidth]{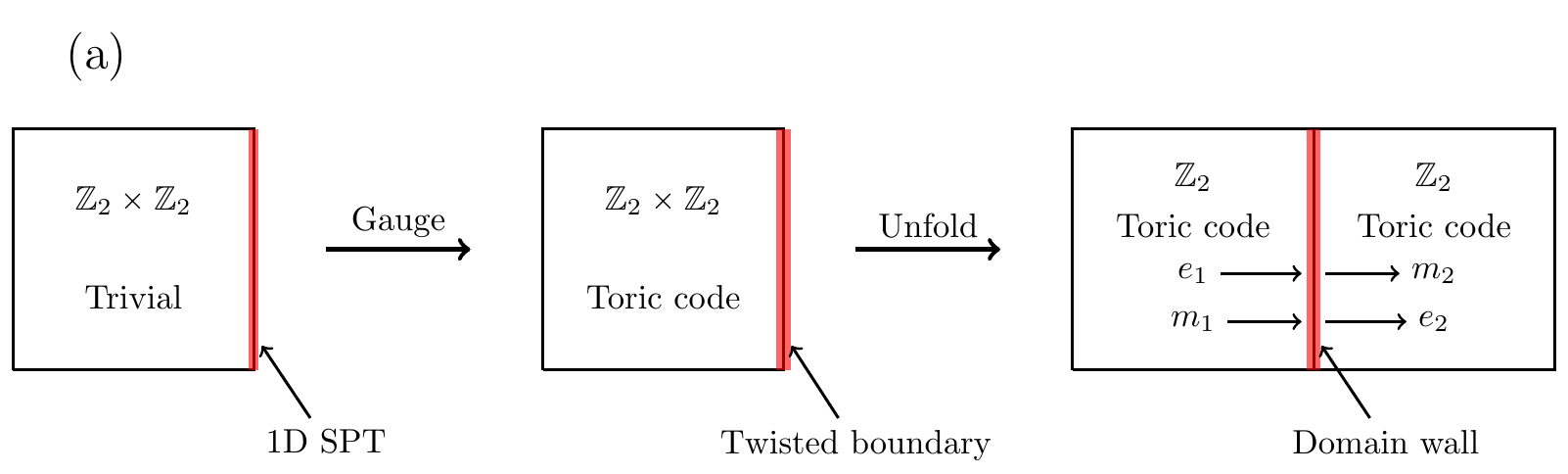}
	\qquad
	\includegraphics[width=0.8\linewidth]{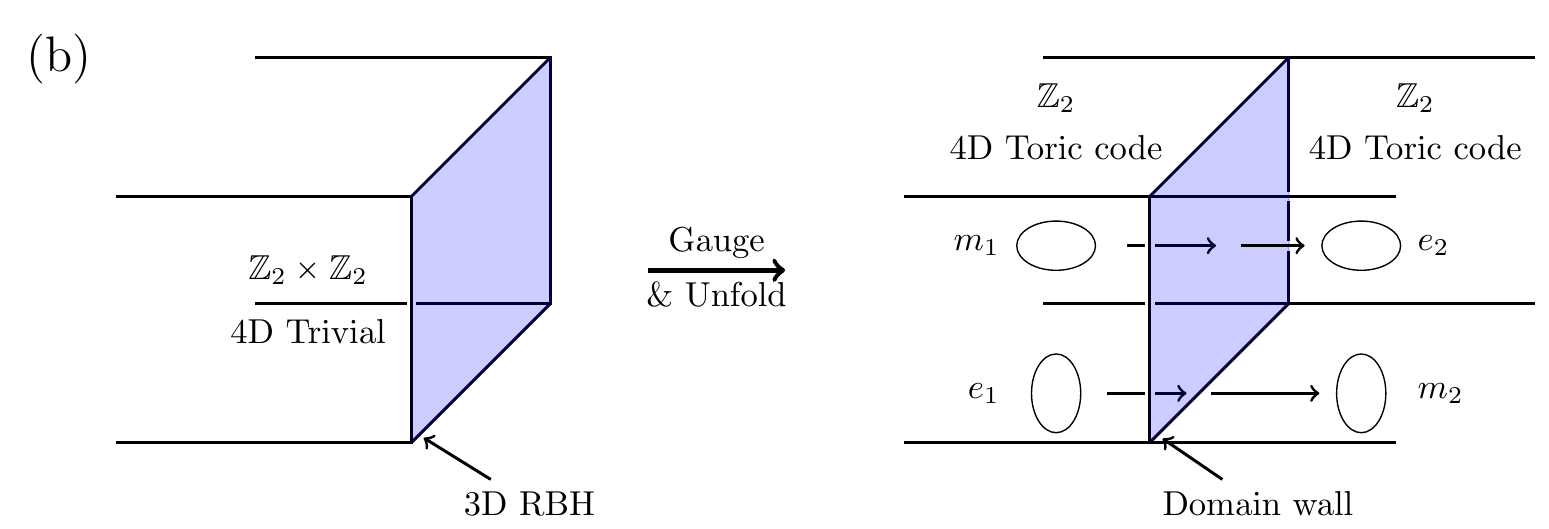} 
	\caption{(Color online) (a) Gauging $\zz_2 \times \zz_2$ symmetry of the two-dimensional model (which has the 1D cluster model on its boundary) leads to a twisted gapped boundary where point-like $e_1m_2$ and $e_2m_1$ particles may condense. This can be viewed as a nontrivial domain wall in the two-dimensional toric code. (b) Gauging the $\zz_2 \times \zz_2$ 1-form symmetry of the four-dimensional model (which has the three-dimensional RBH model on its boundary) leads to a nontrivial domain wall in the four-dimensional toric code, which exchanges electric and magnetic loop-like excitations. 
	} 
	\label{figGaugeWall}%
\end{figure}

Now let us turn to a construction of gapped domain walls from $1$-form SPT phases. Consider a four-dimensional system with $\zz_2 \times \zz_2$ $1$-form symmetry, defined on a lattice $\Lambda'$ with the cubic lattice~$\mathcal{C}$ (described in the previous section) as its boundary. We will consider the following Hamiltonian:
\begin{align}
H_{1} = - \sum_{v\in \text{bulk}(\Lambda')} X_{v} - H_{\text{RBH}}^{\calC},
\end{align}
where $H_{\text{RBH}}^{\calC}$ is the RBH Hamiltonian supported on qubits living on the three-dimensional boundary $\partial\Lambda' = \mathcal{C}$ of the lattice $\Lambda'$. We can gauge the above Hamiltonian to obtain two copies of the four-dimensional toric code with twisted gapped boundaries. On the boundary, loop-like excitations $e_{1}m_{2}$ and $e_{2}m_{1}$ may condense. Here, $e_i$ and $m_i$ ($i =1,2$) correspond to loop-like electric and magnetic excitations (i.e. violated $X$-type and $Z$-type stabilisers of the four-dimensional toric code, respectively). The $e_i$ and $m_i$ excitations occur on the one-dimensional boundary of a two-dimensional membrane of $Z$-type and $X$-type operators, respectively. One can consider this gapped boundary as a gapped domain wall connecting the two copies of the four-dimensional toric code. Upon crossing the domain wall (see Fig.~\ref{figGaugeWall}), the following exchange between electric and magnetic loop-like excitations is implemented
\begin{equation}
e_1 \leftrightarrow m_2, \qquad m_1 \leftrightarrow e_2.
\end{equation} 

This observation already provides an argument that the RBH model is an example of a nontrivial $1$-form SPT phase. To address the thermal stability of the SPT-order of the RBH model, one may appeal to the thermal stability of the four-dimensional toric code where the nontrivial braiding statistics between electric and magnetic loop-like excitations survive even at nonzero temperature. The fact that the gapped domain wall implements an exchange of loop-like excitations with nontrivial braiding properties at nonzero temperature is an indication that the underlying RBH Hamiltonian with 1-form symmetry is thermally stable. 

\subsection{Order parameters for detecting SPT-order of the thermal RBH model}\label{secMembrane}

We now give a direct proof of the nontrivial SPT-order of the thermal RBH model when the 1-form symmetry is enforced. The proof is based on a set of membrane operators that serve as order parameters. In addition to serving as witnesses of SPT-order, these membrane operators can be used to demonstrate the ability to perform gate teleportation in the MBQC scheme, as explored in section \ref{sec4}. These operators can be viewed as generalisations of the string order parameters used to detect SPT-order in one dimension \cite{KT1,KT2,HiddenSymBreak} and similar constructions can be made for other higher form SPT-ordered models. Such operators can be specified a two-dimensional surfaces as follows. For any dual-2-chain $\Gamma_1\in C_2^*$ and any 2-chain $\Gamma_2\in C_2$ (which will be thought of as surfaces in the primal and dual lattices respectively), we define a membrane operator
\begin{align}
M(\Gamma_1) := \prod_{\sigma_1^* \in\Gamma_1}K(\sigma_1), \qquad M(\Gamma_2) := \prod_{\sigma_2 \in\Gamma_2}K(\sigma_2), 
\end{align}
where the first product is over all 1-cells $\sigma_1$ whose dual belongs to $\Gamma_1$. By definition of the cluster terms in Eq. (\ref{clusterTerms}), the membrane operators can be written as follows
\begin{align}\label{memcluster2}
M(\Gamma_1)= X(\Gamma_1) \cdot Z(\partial^{*} \Gamma_1), \qquad M(\Gamma_2)= X(\Gamma_2) \cdot Z(\partial \Gamma_2).
\end{align}
Since the cluster terms are commuting, the membrane operators for any 2-chain and dual-2-chain will also commute with each other and the cluster Hamiltonian. Additionally, at zero temperature the cluster state will be a $+1$-eigenstate of these operators for any choice of $\Gamma_1$ and $\Gamma_2$ (as the cluster state is a $+1$-eigenstate of the cluster terms).  

We now specify a class of membrane operators that we will be interested in. First, let $(\hat{x}, \hat{y} , \hat{z})$ be a coordinate system of the cubic lattice $\calC$ (with opposite boundaries identified). We choose two two-dimensional slices $L \subseteq \calC$ and $R\subseteq \calC$ that are separated in the $\hat{z}$ direction by a distance of at least $d/4$ (where $d$ is the linear size of the lattice $\calC$). These two regions are required to be extensive in both the $\hat{x}$ and $\hat{y}$ directions (i.e. each region has the topology of a torus) as depicted in Fig.~\ref{figSheetOrder}.

We choose $\Gamma_1$ to be a nontrivial dual-2-cycle in the $\hat{x} - \hat{z}$ plane, which can be regarded as a noncontractible surface (see Fig. \ref{figHorMembrane}). Let $\Gamma_2$ be a 2-chain in the $\hat{y} - \hat{z}$ plane, with boundary 
\begin{equation}
\partial \Gamma_2 = S_2^L +  S_2^R,
\end{equation}
such that $S_2^L \subseteq L$ and $S_2^R \subseteq R$ are both nontrivial 1-cycles winding in the $\hat{y}$ direction. The membrane operators corresponding to $\Gamma_12$ and $\Gamma_2$ are illustrated in Fig.~\ref{figSheetOrder}. We note that the precise form of $\Gamma_1$ and $\Gamma_2$ is not important, any membranes that differ by a 2-boundary or a dual-2-boundary may be considered equivalent. We note that the distance between the left and right boundaries $S_2^L$ and $S_2^R$ is lower bounded by $d/4$.

\begin{figure}[htb!]%
	\centering
	\subfloat[]{{\includegraphics[width=0.49\linewidth]{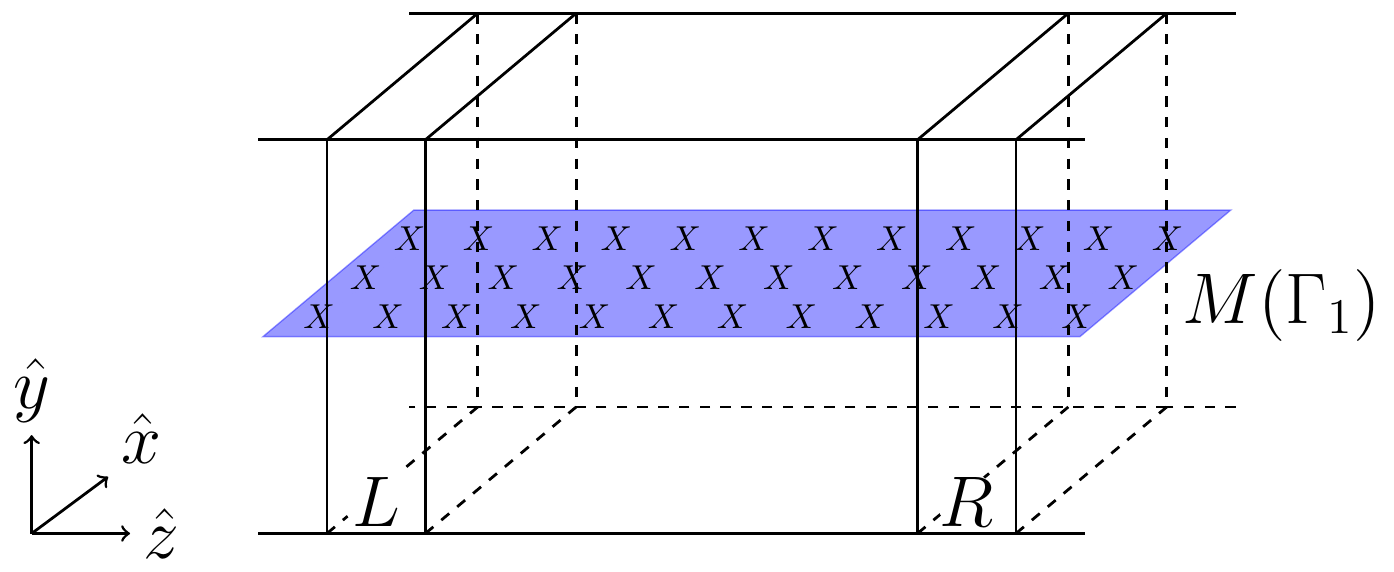} }\label{figHorMembrane}}%
	\qquad
	\subfloat[]{{\includegraphics[width=0.43\linewidth]{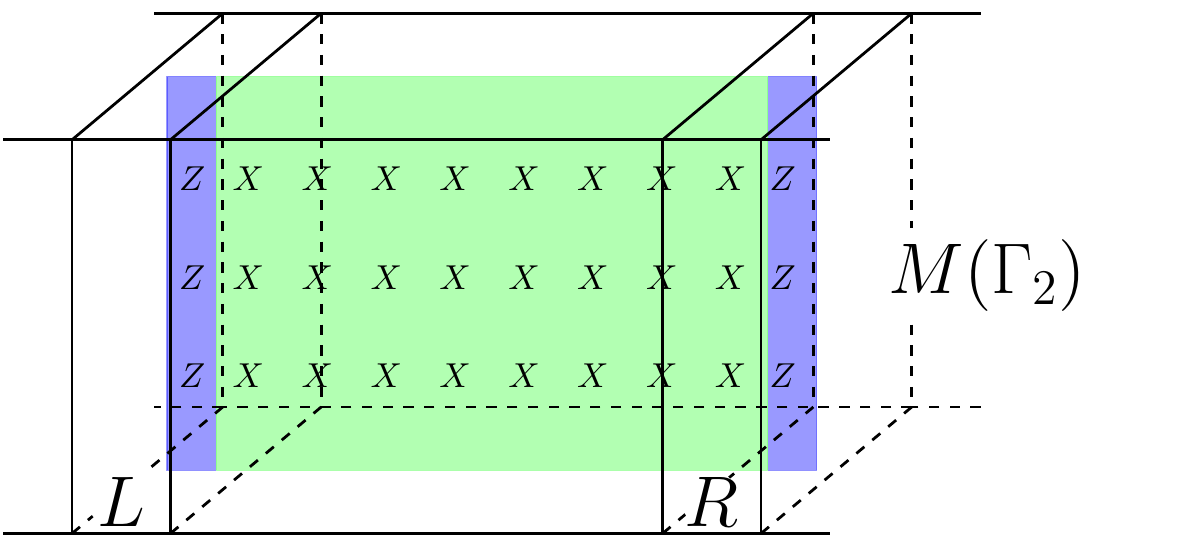} }\label{figVertMembrane}}%
	\caption{(Color online) (a) The membrane operator $M(\Gamma_1)$ and (b) the membrane operator $M(\Gamma_2)$. The top and bottom boundaries are identified, as are the front and back boundaries. The primal qubits lie on the darker (blue) sheets, and the dual qubits on the lighter (green) sheet.  The restrictions of these membrane operators to either $L$ or $R$ anti-commute. Note that length in the $\hat{z}$ direction has been exaggerated.}%
	\label{figSheetOrder}%
\end{figure}

These membrane operators are constructed to have nontrivial algebraic relations on the regions $L$ and $R$. Namely, let $M^{{L}}(\Gamma_1)$ and $M^{{L}}(\Gamma_2)$ be the restriction of $M(\Gamma_1)$ and $M(\Gamma_2)$ to the region ${L}$, respectively. Then this restriction gives an anti-commuting pair of operators 
\begin{equation}\label{EqAntiCom}
\{ M^{{L}}(\Gamma_1) , M^{{L}}(\Gamma_2) \} = 0.
\end{equation}
This is because the boundary of $M(\Gamma_2)$ consists of a string of Pauli $Z$ operators, which intersects the sheet of Pauli $X$ operators of $M(\Gamma_1)$ at a single site, as depicted in Fig.~\ref{figSheetOrder}. Similarly, the restriction of the membrane operators to ${R}$ gives a pair of anti-commuting operators.  By analogy to one-dimensional SPT phases, the membrane operators $M(\Gamma_1)$ and $M(\Gamma_2)$ generate a $\zz_2\times \zz_2$ group, while their restriction to the boundaries gives a nontrivial projective representation of the $\zz_2\times \zz_2$ group \cite{Else2012,Delse1,HCComplexity}.

For these choices, let $\Gamma=(\Gamma_{1}, \Gamma_{2})$ denote the pair of membranes and let $M_1 = M(\Gamma_1)$ and $M_2 = M(\Gamma_2)$. To define the order parameter, we must also allow for the ability to perform local error correction within a neighbourhood of each region $L$ and $R$. As we will see, this error correction will be a necessary ingredient to detect SPT-order in the RBH thermal state. In particular, let $\overline{L}$ and $\overline{R}$ be non-intersecting neighbourhoods of $L$ and $R$ respectively, and let $\calE_{\overline{L}} \otimes \calE_{\overline{R}}$ be any operation local to $\overline{L}$ and $\overline{R}$. Namely, $\calE_{\overline{L}} \otimes \calE_{\overline{R}}$ consists of measurements, followed by an outcome dependent local unitary, which will be thought of as an error correction map. For a state $\rho$, the order parameter is defined as the expectation value of the membrane operators, maximized over all locally error corrected states $\overline{\rho} = \calE_{\overline{L}} \otimes \calE_{\overline{R}}(\rho)$, 
\begin{equation}\label{eqorderparameter}
O_{\Gamma}(\rho): =  \max_{\overline{\rho} = \calE_{\overline{L}} \otimes \calE_{\overline{R}}(\rho)} \frac{1}{2}\Tr\left(\overline{\rho} (M_1 + M_2)\right).
\end{equation}
For our purposes it will be sufficient to consider error correction within neighbourhoods of $L$ and $R$ that have radius $\calO(\log(d))$. One can impose the additional restriction that the measurements and unitaries of $\calE_{\overline{L}} \otimes \calE_{\overline{R}}$ be symmetric, although this is not required to distinguish phases. 

In Lemma \ref{lemMembrane} we will derive an upper bound on the value of $O_{\Gamma}(\rho)$ for thermal states with trivial SPT-order. Then in Lemma \ref{lemClusterNontriv} we show that there exists a nonzero critical temperature $T_c$, such that $O_{\Gamma}(\rho_{\calC}(\beta)) \approx 1$ for the symmetric thermal state of the RBH model at $0 \leq T \leq T_c$. 

\begin{lem}\label{lemMembrane}
	For any $\zz_2 \times \zz_2$ symmetric ensemble $\rho_0$ that is $(r,\epsilon)$ SPT-trivial with $r$ sub-linear in the lattice size $d$, there exist sufficiently large $d$ such that $O_{\Gamma}(\rho_0) \leq1/2 + \epsilon$.
\end{lem}

\begin{proof}
	Since $\rho_0$ is $(r,\epsilon)$ SPT-trivial, we can approximate it by $\rho' = \sum_{a}p(a) \ket{\psi_{a}} \bra{\psi_{a}}$ up to error-$\epsilon$ in trace norm, where each state $\ket{\psi_{a}}$ is an $(r,0)$ SPT-trivial state, and $p(a)$ is a probability distribution. For each $\ket{\psi_a}$, we have $\bra{\psi_a} M_i \ket{\psi_a} = \bra{\phi} U_a^{\dagger} M_i U_a \ket{\phi}$ for some symmetric circuit $U_a$ of depth $r$, where $\ket{\phi}$ is a product state. 
	
	Let $w$ be the largest value out of $\calO(\log(d))$ and $r$. Take enlarged regions $\overline{L}$ of $L$ and $\overline{R}$ of $R$ obtained by taking $w$-neighbourhoods around $L$ and $R$ respectively. Since $r$ is sub-linear in $d$, then we can take $d$ sufficiently large such that $\overline{L} \cap \overline{R} = \emptyset$. For a transversal operator $A$ (meaning it is a tensor product of single-qubit operators), and a subregion $\chi$ of the lattice $\calC$, let $A^{\chi}$ denote the restriction of $A$ to $\chi$. Since the membranes $M_i$ are transversal, we can decompose them across the regions, $M_i = M_i^{\overline{L}} \otimes M_i^{bulk} \otimes M_i^{\overline{R}}$, where the bulk region is the complement of $\overline{L}\cup \overline{R}$. We now claim that the regions $\overline{L}$ and $\overline{R}$ are large enough so that $[M_i^{bulk}, U_a] = 0$ for $i=1,2$ and all $a$. 
	
	Firstly, we must have $[M_1 , U_a] = 0$ for all $a$. This is because for any local region $\chi$, there exists a dual-2-boundary $b_2'$ such that $M_1^{\chi} = S(b_2')^{\chi}$. Since each gate in $U_a$ has to be symmetric, it must commute with $S(b_2')$ for any dual-2-boundary $b_2'$. It follows that each gate must also commute with $M_1$. Now consider $M_2$, for any region $\chi$ away from the boundary of $\Gamma_2$, similarly we can always find a 2-boundary $b_2$ such that $M_2^{\chi} = S(b_2)^{\chi}$. Similarly, each gate in $U_a$ must commute with $S(b_2)$ for any 2-boundary $b_2$ and therefore also with $M_2^{bulk}$. This is not satisfied in general near the boundaries of $\Gamma_2$. But provided $M_2^{bulk}$ is supported a distance greater than the circuit depth $r$ away from the boundaries $S_2^L$ and $S_2^R$, then we have $[M_2^{bulk}, U_a] = 0$, $\forall a$.
	
	We can therefore write $U_a^{\dagger} M_i U_a =  M_{i,a}^{\overline{L}} \otimes M_{i}^{bulk} \otimes M_{i,a}^{\overline{R}}$, where $M_{i,a}^{\overline{L}}= U_a M_{i}^{\overline{L}} U_a^{\dagger}$ and similarly for $M_{i,a}^{\overline{R}}$. From Eq. (\ref{EqAntiCom}), the restriction of membrane operators $M_1$ and $M_2$ to the either of the disjoint regions $\overline{L}$ or $\overline{R}$, give rise to the following anti-commutation relations 
	\begin{equation}\label{EqAntiCom2}
	\{ M_{1}^{\overline{L}} , M_{2}^{\overline{L}} \} = \{ M_{1}^{\overline{R}} , M_{2}^{\overline{R}} \} = 0. \quad \textcolor{white}{\forall a}
	\end{equation}
	Because of unitary equivalence between the operators, we also have 	
	\begin{equation}\label{anticommute}
	\{M_{1,a}^{\overline{L}},M_{2,a}^{\overline{L}}\} = \{M_{1,a}^{\overline{R}},M_{2,a}^{\overline{R}}\} = 0  \quad \forall a,
	\end{equation}
	where these operators also have eigenvalues $\pm 1$. Since $\ket{\phi}$ is a symmetric product state, it is a $+1$-eigenstate of $M_{i}^{bulk} $. However, $\ket{\phi}$ cannot be a simultaneous eigenstate of both $M_{1,a}^{\overline{L}}$ and $M_{2,a}^{\overline{L}}$, nor of $M_{1,a}^{\overline{R}}$ and $M_{2,a}^{\overline{R}}$ due to the anti-commutation relations of Eq.~(\ref{anticommute}) and since $\overline{L}$ and $\overline{R}$ are disjoint. In particular, since $\ket{\phi}$ is a tensor product $\ket{\phi} = \ket{\phi}_{\overline{L}}\otimes \ket{\phi}_{bulk} \otimes \ket{\phi}_{\overline{R}}$ then 
	\begin{align}
	\bra{\psi_a} M_i \ket{\psi_a} &=	\bra{\phi} U_a^{\dagger} M_i U_a \ket{\phi} \\
	&= \bra{\phi}_{\overline{L}} M_{i,a}^{\overline{L}}\ket{\phi}_{\overline{L}} \cdot \bra{\phi}_{bulk} M_{i}^{bulk}\ket{\phi}_{bulk} \cdot \bra{\phi}_{\overline{R}} M_{i,a}^{\overline{R}}\ket{\phi}_{\overline{R}}\\
	&= \langle M_{i,a}^{\overline{L}}\rangle \cdot \langle M_{i,a}^{\overline{R}}\rangle \label{EqantiProd}
	\end{align}
	for $i = 1,2$, where $\langle M_{i,a}^{\overline{L}}\rangle = \bra{\phi}_{\overline{L}} M_{i,a}^{\overline{L}}\ket{\phi}_{\overline{L}} $ and $\langle M_{i,a}^{\overline{R}}\rangle = \bra{\phi}_{\overline{R}} M_{i,a}^{\overline{R}}\ket{\phi}_{\overline{R}}$. It is shown in \cite{AntiComBound} that for $k$ mutually anti-commuting operators $\{A_i\}$, each with eigenvalues $\pm1$, any state $\ket{\psi}$ satisfies the following inequality
	\begin{equation}\label{AntiBoundEq}
	\sum_{i = 1}^k \langle A_i \rangle_{\ket{\psi}} \leq \sum_{i = 1}^k \langle A_i \rangle_{\ket{\psi}}^2 \leq 1,
	\end{equation}
	where the expectation value is taken with respect to the state $\ket{\psi}$. Using Eq. (\ref{EqantiProd}), we have for the approximate state
	\begin{align}
	\Tr\left(\rho'(M_1 +M_2)\right) &= \sum_a p(a)\left(\langle M_{1,a}^{\overline{L}}\rangle  \cdot \langle M_{1,a}^{\overline{R}}\rangle  + \langle M_{2,a}^{\overline{L}}\rangle  \cdot \langle M_{2,a}^{\overline{R}}\rangle \right) \\
	&\leq \left( \langle M_{1,a}^{\overline{L}}\rangle^2 + \langle M_{2,a}^{\overline{L}}\rangle^2 \right)^{\frac{1}{2}} \left( \langle M_{1,a}^{\overline{R}}\rangle^2 + \langle M_{2,a}^{\overline{R}}\rangle^2 \right)^{\frac{1}{2}} \\
	&\leq 1,
	\end{align}
	where the first inequality is the Cauchy-Schwarz inequality and the second inequality is using Eq.~(\ref{AntiBoundEq}). Now for any error correction map $\calE_{\overline{L}} \otimes \calE_{\overline{R}}$ that is localized to the non-intersecting neighbourhoods $\overline{L}$ and $\overline{R}$ of $L$ and $R$ respectively, we have by the same argument
	\begin{equation}
	\Tr\left(\calE_{\overline{L}} \otimes \calE_{\overline{R}}(\rho')(M_1 +M_2)\right) \leq 1.
	\end{equation}
	Then since $\rho'$ and $\rho_0$ are close in trace norm, they have similar expectation values of bounded observables, in the following way. Assume $\calE =\calE_{\overline{L}} \otimes \calE_{\overline{R}}$ is the map which maximizes $O_{\Gamma}(\rho_0)$, then
	\begin{align}
	|O_{\Gamma}(\rho_0) - \frac{1}{2}\Tr\left(\calE(\rho')(M_1 +M_2)\right)| &= \frac{1}{2} |\Tr\left((M_1 + M_2)(\calE(\rho_0) - \calE(\rho'))\right)| \\
	&\leq \frac{1}{2}\norm{(M_1 + M_2)(\calE(\rho_0) -\calE(\rho'))}_1 \\
	&\leq \frac{1}{2}\norm{M_1 + M_2}_{\infty}\cdot \norm{\calE(\rho_0 )- \calE(\rho')}_1 \\
	&\leq \norm{\rho_0 - \rho'}_1 \\
	&\leq \epsilon,
	\end{align}
	where the second inequality follows from H{\"o}lder's inequality. The claim then follows. 
\end{proof}

One could define more complicated order parameters so that the bound on $O_{\Gamma}(\rho_0)$ in Lemma \ref{lemMembrane} can be made arbitrarily small. However, our choice and the above bound will be sufficient to show the RBH model has nontrivial SPT-order. Next we show that the thermal RBH model with 1-form symmetry has a high expectation value of the membrane operators provided the temperature is below some critical temperature $T_c$. We do this by showing that large loop excitations are confined in the low temperature phase. In subsection \ref{secIG} we will show that $T_c$ is the critical temperature of the three-dimensional $\zz_2$ Ising gauge model.
\begin{lem}\label{lemClusterNontriv}
	For the symmetric thermal Gibbs ensemble $\rho_{\calC}(\beta)$ of the RBH model with $\zz_2 \times \zz_2$ 1-form symmetry with $0 \leq T \leq {2}/{\log (5)}$, there exists a constant $\delta \textgreater 0$ (independent of systems size) such that for sufficiently large $d$ we have
	\begin{equation}
	O_{\Gamma}(\rho_{\calC}(\beta)) \geq 1 - \calO(d^{-\delta}).
	\end{equation} 
\end{lem}
\begin{proof} 
	Consider first the expectation value of $M_2$. Since $M_2$ can be constructed from a product of cluster terms (as in Eq. (\ref{memcluster2})), we have at zero temperature $\Tr(M_2\rho) = 1$. Using the symmetric Gibbs ensemble $\rho_{\calC}(\beta)$ in Eq. (\ref{symgibbscluster}), the expectation value of a membrane operator is given by 
	\begin{equation}\label{loopExp}
	\text{Tr}(\rho_{\calC}(\beta) M_2) = \sum_{(\gamma, \gamma')\in Z_1 \times {Z}_1^{*}}\text{Pr}_{\beta}((\gamma, \gamma')) \langle M_2 \rangle_{\ket{\psi(\gamma, \gamma')}}, 
	\end{equation}	
	where the expectation value is with respect to the excited state $\ket{\psi(\gamma, \gamma')} = Z(\gamma)Z(\gamma')\ket{\psi_{\calC}}$.	Let $|\Gamma_2 \cap \gamma'|$ denote the number of times $\gamma'$ intersects $\Gamma_2$. Since $\ket{\psi(\gamma, \gamma')}$ is a $\pm1$ eigenstate of $M(\Gamma_2)$, we have 
	\begin{equation}\label{errorCycle}
	\langle M(\Gamma_2) \rangle_{\ket{\psi(\gamma, \gamma')}} = \begin{cases}
	+1 \text{ if } |\Gamma_2 \cap \gamma'| = 0 \text{ mod } 2  \\
	-1 \text{ if } |\Gamma_2 \cap \gamma'| = 1 \text{ mod } 2. \
	\end{cases}
	\end{equation}
	The right-hand side of Eq. (\ref{errorCycle}) is independent of the 1-cycle $\gamma$ since it is supported on the dual lattice and therefore $Z(\gamma)$ commutes with $M_2$. Notice that a similar expression holds for $M_1$. We call $\gamma'$ an error cycle if $|\Gamma_2 \cap \gamma'| = 1$ mod 2 (and similarly for $\Gamma_1$). We will show that there exists a critical temperature $T_c$, below which, large error cycles are suppressed and that error correction on the boundaries can account for the remaining errors. First we define an approximate state, where large loop-like excitations have been removed. 
	
	We say $\gamma \in Z_1$ is a \textit{loop} if any proper subset $\gamma' \subsetneq \gamma$, is not a cycle. We can partition the set of 1-cycles according to the size of the largest loop they contain. Specifically, let $Z_1^{\alpha} \subseteq Z_1$ consist of the set of 1-cycles whose largest loops are of length smaller than $\alpha$ (a similar definition holds for $Z_1^{*{\alpha}} \subseteq Z_1^{*}$). Then define the approximate state
	\begin{equation}
	\rho_{\text{ap}}^{\alpha}(\beta) = \sum_{(\gamma, \gamma')\in Z_1^{\alpha}  \times Z_1^{*{\alpha}}}{\Pr}_{\beta}(\gamma, \gamma')\ket{\psi(\gamma, \gamma')}  \bra{\psi(\gamma, \gamma')}. 
	\end{equation}
	We claim that for a fixed $0 \leq T \textless T_c = 2/\log(5)$, there exists a constant $c$ such that for $\alpha = c\log(d)$, we have
	\begin{equation}\label{approxCluster}
	\norm{\rho_{\text{ap}}^{\alpha}(\beta) - \rho(\beta)}_1 \leq \calO(d^{-\delta}),
	\end{equation}
	for some constant $\delta \textgreater 0$. To see this, fix $\alpha = c\log(d)$ and let $V =\left( Z_1^{} \times Z_1^{*} \right) \setminus \left(Z_1^{ \alpha} \times Z_1^{* \alpha}\right)$, be set of (dual-)cycles containing a loop of size at least $\alpha$ (note that a loop may refer to a subset of a 1-cycle or a dual-1-cycle). Then we have 
	\begin{equation}
	\norm{\rho_{\text{ap}}^{\alpha}(\beta) - \rho(\beta)}_1 = \sum_{(\gamma, \gamma') \in V} {\Pr}_{\beta}(\gamma, \gamma').
	\end{equation}	
	We can bound the above equation using the following relation
	\begin{align}\label{loopSigma}
	\sum_{(\gamma, \gamma') \in V} {\Pr}_{\beta}(\gamma, \gamma') &\leq \sum_{\substack{\text{loops } l \in Z_1 \cup Z_1^* \\ |l| \geq \alpha}} ~ \sum_{\substack{(c_1, c_1') \in Z_1 \times Z_1^* \\  l \subseteq c_1
			\text{ or } l \subseteq c_1' }} \text{Pr}_{\beta}(c_1, c_1'), \\
	&\leq \sum_{\substack{\text{loops } l \in Z_1 \cup Z_1^* \\ |l| \geq \alpha}} ~e^{-2\beta |l|}\cdot \sum_{\substack{(c_1, c_1') \in Z_1 \times Z_1^* \\  l \nsubseteq c_1 \text{ and } l \nsubseteq c_1' }} \text{Pr}_{\beta}(c_1, c_1'), \\
	&\leq \sum_{\substack{\text{loops } l \in Z_1 \cup Z_1^* \\ |l| \geq \alpha}} ~e^{-2\beta |l|} \\
	&\leq \sum_{k \geq \alpha} N(k) e^{-2\beta k},
	\end{align}
	where $N(k)$ is the number of loops in $Z_1 \cup Z_1^{*}$ of size $k$. For the cubic lattice $\calC$, the number of loops $N(k)$ of size $k$ can be bounded by $N(k)\leq 2\frac{6}{5}|\Delta_0|5^k$ (we can upper bound the number of possible loops by counting the number of non-backtracking walks: a non-backtracking walk can begin at any vertex and can move in at most 5 independent directions). Therefore, provided $\beta \textgreater \log(5)/2$, we have 
	\begin{align}
	\sum_{(\gamma, \gamma') \in V} {\Pr}_{\beta}(\gamma, \gamma') &\leq \frac{12}{5} |\Delta_0|\sum_{k=\alpha}^{\infty} e^{-k(2\beta - \log(5))} \\
	&= c'|\Delta_0| e^{-\alpha(2\beta- \log(5))},\label{sigmabound}
	\end{align}
	where $c' = 12/5(1 - e^{(\log(5) - 2\beta)})$ is independent of $d$. Since $|\Delta_0| = (d+1)^3$, the error in Eq. (\ref{approxCluster}) is exponentially small in ${\alpha}$, provided the temperature is below a critical temperature $T_c$. Here, we have given a lower bound on $T_c$ by $2/\ln(5)$. In terms of the lattice size $d$ we have
	\begin{equation}\label{sigmaBound}
	\sum_{(\gamma, \gamma') \in V} {\Pr}_{\beta}(\gamma, \gamma') \leq \calO( d^{-c(2\beta-\log(5))+3}).
	\end{equation}
	Choosing $c\geq 3/(2\beta-\log(5))$, we have $\delta = c(2\beta-\log(5))-3 \textgreater 0$ and the claim follows. Notice that this argument shows that large loop excitations in the RBH thermal state are suppressed, and is similar to Peierls' argument for spontaneous magnetization in the two-dimensional Ising model~\cite{peierlsIsing}.
	
	Now we show that for these values of $T$ and $\alpha$, there exists an error correction map $\calE$ such that 
	\begin{equation}\label{EqapproxCorrected}
	\Tr \left(\calE(\rho_{\text{ap}}^{\alpha}(\beta))(M_1 + M_2)\right) \geq 2-\calO(d^{-\delta}).
	\end{equation}
	Indeed, notice that if $d$ is large enough, the approximate state contains no homologically nontrivial excitations, as they must have length at least $d$. These are the only types of errors that reduce the expectation value of $M_1$, and so the approximate state satisfies 
	\begin{align}
	\Tr(\rho_{\text{ap}}^{\alpha}(\beta)M_1) &= \Tr(\rho_{\text{ap}}^{\alpha}) \\
	&= 1 - \sum_{(\gamma, \gamma') \in V} {\Pr}_{\beta}(\gamma, \gamma') \\
	&\geq 1 - \calO(d^{-\delta}).
	\end{align}
	using Eq. (\ref{sigmaBound}). The only types of errors in the approximate state that reduce $M_2$ are dual-1-cycles containing a loop that wraps around a boundary component of $\partial\Gamma_2=S_2^L\sqcup S_2^R$. Therefore any excitation in $\rho_{\text{ap}}^{\alpha}(\beta)$ that gives rise to an error is contained within an $\alpha/2$ neighbourhood of $S_1^L$ and $S_1^R$. By measuring all cluster terms $K(\sigma_2)$ in an $\alpha/2$ neighbourhood of $\partial \Gamma_2$ one can determine the location of any possible error cycles (for sufficiently large $d$, these $\alpha/2$ neighbourhoods are non-intersecting). Then depending on the parity of the number of error loops, one can apply a correction operator $Z(\gamma')$ for some dual-1-cycle $\gamma'$ wrapping around $S_2^L$ or $S_2^R$, that returns $\rho_{\text{ap}}^{\alpha}(\beta)$ to the $+1$ eigenspace of $M_2$. Letting $\calE$ denote the measurement and recovery steps (which in particular does not change the expectation value of the other membrane operator $M_1$ since the recovery is a local unitary), the approximate state similarly satisfies $\Tr(\calE(\rho_{\text{ap}}^{\alpha}(\beta))M_2) \geq 1- \calO(d^{-\delta})$, and therefore Eq. (\ref{EqapproxCorrected}) holds.

	Finally, let $\calE$ be the aforementioned error correction map, using an argument similar to that in  Lemma~\ref{lemMembrane}, we have 
	\begin{align}
	\left| \Tr\left(\calE(\rho_{\text{ap}}^{\alpha}(\beta))(M_1 + M_2)\right) - \Tr\left(\calE(\rho_{\calC}(\beta))(M_1 + M_2)\right) \right| &\leq2 \norm{\calE (\rho_{\text{ap}}^{\alpha}(\beta)) - \calE(\rho_{\calC}(\beta))} _1 \\
	&\leq 2\norm{\rho_{\text{ap}}^{\alpha}(\beta) - \rho_{\calC}(\beta)}_1 \\
	&\leq \calO(d^{-\delta}).
	\end{align}
	Then using Eq. (\ref{EqapproxCorrected}) we have that 
	\begin{equation}
	O_{\Gamma}(\rho_{\calC}(\beta)) \geq 1 - \calO(d^{-\delta}),
	\end{equation}	
	completing the proof.
\end{proof}

Lemma \ref{lemClusterNontriv} tells us that $O_{\Gamma}(\rho_{\calC}(\beta)) \rightarrow 1$ in the limit of infinite system size. This, along with Lemma \ref{lemMembrane}, shows that the RBH cluster model, protected by 1-form symmetry has nontrivial SPT-order for temperatures $0 \leq T \leq T_c$. The key ingredient in the proof is that large loop configurations are energetically suppressed in the low temperature phase, and this results in a type of \textit{string tension}. This is the characteristic behaviour of the $\zz_2$ lattice gauge theory in three dimensions, and we make this connection precise in the next subsection. Above the critical temperature, the string tension disappears as large error cycles become entropically favourable \cite{KogutGauge,FSGauge} and thus $O_{\Gamma}(\rho_{\calC}(\beta))$ will approach 0. We correspondingly expect the SPT-order to disappear above $T_c$.

\subsection{Comparison with a three-dimensional Ising gauge model}\label{secIG}

Having proved that the nontrivial SPT-order of the RBH model under the 1-form symmetry survives at nonzero temperature, we now compare it to a three-dimensional Ising gauge model \cite{KogutGauge,3DIsingGauge}. This comparison is natural because the 1-form symmetry of the RBH model and the gauge symmetry of the three-dimensional Ising gauge model are closely related. The model can be defined on the same lattice $\calC$ as the RBH model, and the Hamiltonian is given by a sum of plaquette terms
\begin{equation}\label{eqIsingG}
H_{IG} = - \sum_{\sigma_2 \in \Delta_2} Z(\partial \sigma_2) - \sum_{{\sigma}_1 \in {\Delta}_1} Z(\partial^{*} {\sigma}_1).
\end{equation}
We notice that the first and second terms are supported on disjoint sublattices so that $H_{IG}$ describes two decoupled copies of a three-dimensional Ising gauge model on the cubic lattice. This model has local gauge symmetries, which are the 1-form operators of Eq. (\ref{1formsym}).

Excitations of this model take the form of loop-like objects, and can be created by products of Pauli $X$ operators. These loop-like excitations have an energy cost proportional to their length in the same way as the RBH model with 1-form symmetry. Indeed, the spectrum of $H_{IG}$ is identical to that of the RBH model $H_{\calC}$ with 1-form symmetry enforced, and one can construct a duality mapping between the 1-form symmetric model $H_{\calC}$ and two copies of the three-dimensional Ising gauge model $H_{IG}$.

This Ising gauge model $H_{IG}$ has a low-temperature ordered phase where the excitations have string tension, such that large loops excitations are suppressed. The suppression of large excitations was the necessary ingredient in the proof of Lemma \ref{lemClusterNontriv} which we use to show the nontriviality of the 1-form symmetric RBH model at nonzero temperature. Therefore the lower bound of $T_c$ in Lemma \ref{lemClusterNontriv} of $2/\log(5) \approx 1.24$ can be increased to the critical temperature of the three-dimensional Ising gauge model, which has been estimated via numerical simulations~\cite{3DIsingGauge}, to be $T_{IG} \approx 1.31$. 

It is worth noting that the model described by the Hamiltonian $H_{IG}$ and the RBH model $H_{\calC}$ belong to distinct phases at zero temperature under 1-form symmetries, since the three-dimensional Ising gauge model has long-range entangled (topologically ordered) ground states. This distinction persists to nonzero temperature $T$ with $0 \leq T \leq T_c$, as the $H_{IG}$ retains the same order as the three-dimensional toric code \cite{CCTopoMem}. Indeed, the three models: the trivial paramagnet $H_X$, the RBH model $H_{\calC}$ and the three-dimensional Ising gauge theory $H_{IG}$, all have the same spectrum under 1-form symmetries and belong to distinct symmetric phases for temperatures $0 \leq T \leq T_c$. From the viewpoint of quantum information processing tasks, each of these phases has distinct uses: $H_{IG}$ can be used as a memory at nonzero temperature for the storage of classical bits \cite{CCTopoMem}, while the RBH model $H_{\calC}$ is a universal resource for MBQC at nonzero temperature. 

\section{Localizable entanglement}\label{sec4}

In the previous section, we have shown that the RBH model possesses nontrivial SPT-order at nonzero temperature when protected by a 1-form symmetry, and we developed order parameters that detect this nontrivial SPT phase.  In subsection \ref{sec4A} we provide an operational interpretation for these order parameters in terms of quantifying the entanglement that can be localized between distant regions in the thermal state through measurements in the bulk.  This provides a connection with the zero-temperature results in 1D SPT models \cite{Else2012}, where all nontrivial SPT-ordered ground states possess long-range localizable entanglement. These order parameters are also relevant in the context of quantum computation, as localizable entanglement is the underlying mechanism through which---via gate teleportation--- the RBH thermal state functions as a resource for measurement-based quantum computation. 

In subsection \ref{sec4B} we then turn our attention back to the standard RBH model without symmetry, and reflect on the robustness of this model for measurement-based quantum computation even in the case where no symmetry is enforced.   We find a novel perspective: that error correction can be used to restore an effective 1-form symmetry, and when the correction is successful, the model can be used to localize entanglement between distant regions.  This provides a direct link between thermal SPT phase and fault-tolerant measurement-based quantum computation, or more generically, high error-threshold quantum computing architectures.

\subsection{Localizable entanglement in the 1-form SPT model}\label{sec4A}

A primitive form of computation is the ability to generate entanglement between distant regions. Localisable entanglement $\tilde L_{LR}$ is the average entanglement (according to some entanglement measure $E$) of the post measured state between two regions $L$, $R$, maximized over all choices of single-site measurements $M$ on the complement of $L\cup R$. Following~\cite{verstraete2004entanglement}, the localisable entanglement is defined as 
\begin{equation}
\tilde L_{LR}(\rho) = \max_{M} \sum_s p_s E(\rho_s),
\end{equation}
where $\rho_s = \Pi_s \rho \Pi_s/{\rm Tr}(\Pi_s \rho)$ is the post-measurement state associated with a local measurement projector $\Pi_s = |s_1\rangle\langle s_1| \otimes \cdots \otimes |s_n\rangle \langle s_n|$ on $(L \cup R)^c$ and measurement outcome $s=(s_1,\ldots,s_n)$, and $p_s={\rm Tr}(\Pi_s \rho)$ is the probability of outcome $s$.

In general, maximizing over all possible local measurements is difficult, but if the state $\rho$ has a high degree of symmetry then the optimal measurement bases $\Pi_s$ may be determined from symmetry arguments \cite{LEVR}. For the 3D cluster state with the 1-form symmetry, it is straightforward to show (following \cite{DBcompPhases, SPTent}) that the optimal local measurement bases for localizing entanglement are always the $X$-basis; i.e., one should perform local $X$ measurements on all spins in the bulk.  The localizable entanglement of the state $\rho$ can then be expressed as the average entanglement of the post-measurement state $\rho_s$ across the $L/R$ partition:
\begin{equation}
\tilde L_{LR}(\rho) = \sum_s p_s E(\rho_s).
\end{equation}
This entanglement is also known as the $\emph{SPT-entanglement}$~\cite{SPTent}, and shown to be an order parameter for SPT phases protected by onsite symmetries at zero temperature.  We note that, in the presence of the 1-form symmetry, localizable entanglement and SPT-entanglement are identical.

We now show that the order parameters $O_\Gamma(\rho)$ developed in the previous section serve as a witness for localizable entanglement of the thermal SPT state. We note that the membrane operators $M(\Gamma_i)$ take the form
\begin{equation}\label{eq:CommutationIdentity}
M(\Gamma_i) = M_i^{L} \otimes M_i^{bulk} \otimes M_i^{R},
\end{equation}
where the bulk region is the complement of $L\cup R$. Since $L$ and $R$ are 2-dimensional slices, the restrictions $M_1^{L}$ and $M_1^{R}$ are 1-dimensional strings of Pauli $X$ operators, and $M_2^L$ and $M_2^R$ are 1-dimensional strings of Pauli $Z$ operators. Consider measurement of Pauli $X$ on all qubits that either do not belong to the two-dimensional slices $L$ and $R$, or belong to 2-cells of $L$ and $R$. Then the post measured state is an eigenstate of a pair of two-dimensional toric codes, each defined on the slices $L$ and $R$ (see Ref. \cite{RBH} for details). The membrane operators restricted to these slices are equivalent to logical operators of the two-dimensional toric codes, and in particular may be written in terms of these logical operators as $M_1^{L\cup R} = \overline{X}_L\otimes\overline{X}_R$ and $M_2^{L\cup R} = \overline{Z}_L \otimes \overline{Z}_R$.

After performing the local $X$ measurements on the bulk qubits, the measurement projector $\Pi_s$ projects into eigenstates of $M_i^{bulk}$. Then the single qubit measurement outcomes can be multiplied to infer the outcome of each bulk operator $M_i^{bulk}$. This classical information is transmitted to $L$ and $R$ and we can infer the $\pm 1$ outcomes of the logical operators $\overline{X}_L\otimes \overline{X}_R$ and $\overline{Z}_L\otimes\overline{Z}_R$ for the post measured state. Note that due to the anti-commutation relations of Eq. (\ref{EqAntiCom}) these correlations are that of a maximally entangled state encoded within two two-dimensional toric codes. The order parameter of Eq. (\ref{eqorderparameter}) after measurement, $\langle \overline{X}_L\otimes\overline{X}_R + \overline{Z}_L\otimes\overline{Z}_R\rangle/2$, is therefore an entanglement witness for the entanglement between topological degrees of freedom. Note that measurement outcomes of $\overline{X}_L\otimes\overline{X}_R$ and $\overline{Z}_L\otimes\overline{Z}_R$ for the post measured state might potentially depend on the choice of membrane $M_i^{bulk}$, but as discussed in section \ref{secMembrane}, we can freely deform the membrane operators due to the 1-form symmetries, thus removing any ambiguity. This entanglement enables gate teleportation in the topological cluster state quantum computing scheme~\cite{Rau06}, using the thermal state as the resource state.

Having provided an operational interpretation of thermal SPT-order as localizable entanglement in measurement-based quantum computation, we now briefly consider the physical consequence of this localizable entanglement. 
Non-triviality of SPT-order manifests itself most dramatically through physical properties on the boundaries.
For instance, 1D nontrivial SPT phases typically exhibit robust gapless boundary modes similar to those in topological insulators. 
The aforementioned localizable entanglement, or SPT-entanglement, for 1D SPT phases directly measures the boundary degeneracy that appear when the system has open edges \cite{SPTent}.
For three-dimensional systems with symmetries, their two-dimensional boundaries may exhibit robust gapless modes, symmetry-breaking phases and/or 2D topological order \cite{BulkBoundarySPT, TranslationalSET}. 
With $1$-form symmetries in the bulk, the boundary of the 3D cluster state at zero temperature supports the two-dimensional toric code on effective qubits localized near the boundary, and localizable entanglement, as quantified by membrane operators, measures the boundary degeneracy of the toric code on the boundaries. It is tempting to speculate that the presence of localizable entanglement at nonzero temperature in the 3D cluster state suggests that this boundary topological order persists even at nonzero temperature due to $1$-form symmetries in the bulk.

\subsection{Recovering effective 1-form symmetry with error correction}\label{sec4B}

We have shown that the RBH model can retain its long-range localizable entanglement at nonzero temperature when a 1-form symmetry is enforced.  The original results of Ref.~\cite{RBH} demonstrate, however, that this localizable entanglement persists in the thermal state even without any symmetry protection!  This result is surprising because, as we have shown, the protection of a 1-form symmetry is necessary to define an SPT-ordered phase at nonzero temperature.  To add to the confusion, the transition in localizable entanglement in the unprotected model, from long-range at low temperature to short-range at high temperature, does not correspond to any thermodynamic transition. Indeed, the Gibbs state of the RBH model without symmetry protection has no thermodynamic phase transition, and is equivalent to the Gibbs state of a non-interacting paramagnet. What is the underlying quantum order that persists up until this transition in localizable entanglement?

We offer a resolution to this confusing situation, by demonstrating that the persistence of localizable entanglement in the RBH model to nonzero temperature can be understood through imposing an effective 1-form symmetry in the unprotected model via error correction.  The 1-form operators are not enforced \textit{a priori}, but their eigenvalues are reconstructed via the outcomes of the local measurements, and the resulting state can be `restored' to the SPT-ordered thermal state.  We can therefore relax the symmetry requirement on the model, provided it can be effectively restored through error correction. 

Consider the thermal state $\rho_{0}(\beta)$ of the RBH model $H_{\calC}$ where no symmetry is enforced. In the absence of a symmetry, $\rho_{0}(\beta)$ is equivalent to the exact cluster state with local $Z$ errors applied to each qubit with probability $p = (1+\exp(2\beta))^{-1}$, as shown in Ref.~\cite{RBH}.  In order to restore the 1-form symmetry, we follow the error correction scheme detailed in Ref.~\cite{RBH}, which is based on the techniques of Ref.~\cite{dennisTopo}.  We now outline the steps involved with this procedure and we note that error correction proceeds on each sublattice independently. 

Firstly, consider the measurement of all 1-form operators in the local generating set $\tilde{G} = \{S(\partial \sigma_3), S(\partial^{*}\sigma_0) ~|~ \sigma_3 \in \Delta_3, \sigma_0 \in \Delta_0 \}$ given by Eq.~(\ref{genset}), and let $\{ s_{b}=\pm 1 \}$ be the set of corresponding measurement outcomes. A syndrome is the set of all operators in $\tilde{G}$ which return measurement outcome $-1$ and can be found as the dual boundary $\partial^*(c_1')$ of an error chain $Z(c_1')$, and the boundary $\partial c_1$ of an error chain $Z(c_1)$, where $c_1\in C_1$ and $c_1' \in C_1^*$. To recover the 1-form symmetry, one can identify a recovery 1-chain $\gamma_1 \in C_1$ and dual-1-chain $\gamma_1' \in C_1^*$ such that 
\begin{equation}\label{eqRecoveryChain}
\partial(\gamma_1+c_1)=0, \quad \text{and} \quad \partial^*(\gamma_1'+c_1')=0.
\end{equation}
The recovery operator $U_{ \{ s_{b}\} } = Z(\gamma_1)Z(\gamma_1')$ is a product of Pauli $Z$ operators, which is dependent on the measurement outcomes. The post-correction state is 
\begin{align}\label{eqRecovery}
\rho_{\text{sym}}= \sum_{\{s_{b}\}} U_{\{ s_{b}\}} \big(\Pi_{\{ s_{b}\}}\rho_{\text{in}} \Pi_{\{ s_{b}\}} \big) U_{ \{ s_{b}\} }^{\dagger}
\end{align}
where $\Pi_{\{ s_{b}\}}$ is the projection operator onto subspace with syndrome values $\{s_{b}\}$. Since this error-corrected state $\rho_{\text{sym}}$ is $1$-form symmetric by construction, its nontriviality under $1$-form symmetries, in a sense of the circuit complexity, can be defined.

We have recovered the 1-form symmetry, but we have to determine when the error correction is successful, as the choice of recovery chains satisfying Eq. (\ref{eqRecoveryChain}) is not arbitrary. The measure of success is determined by the usefulness of the post-correction state for localizing entanglement, as we will discuss. We say error correction is successful if the recovery chains $\gamma_1$ and $\gamma_1'$ satisfy
\begin{equation}\label{eqSuccessfulChain}
\gamma_1+c_1\in B_1, \quad \text{and} \quad \gamma_1'+c_1' \in B_1^{*} 
\end{equation} 
meaning they are homologically trivial. This means we only need to find recovery chains that are equivalent to $\gamma_1$ and $\gamma_1'$ up to a 1-boundary and a dual-1-boundary, respectively. Optimal error correction finds the most probable equivalence class of chains satisfying Eqs. (\ref{eqRecoveryChain}) and (\ref{eqSuccessfulChain}) for the given syndrome and is known as maximum-likelihood decoding \cite{dennisTopo}. 

The error correction succeeding is equivalent to the post-correction state $\rho_{sym}$ having the same $+1$ expectation values of the operators $X(z_2)$ and $X(z_2')$ as the cluster state, where $z_2 \in Z_2$ is a nontrivial 2-cycle, and $z_2' \in Z_2^*$ is a nontrivial dual-2-cycle. In this case, the corrected state can be reliably used to localize entanglement between distant regions $L$ and $R$, since the measurement outcomes of bulk of the membrane operators $M_i^{bulk}$ in Eq. (\ref{eq:CommutationIdentity}) can be accurately determined.  In the case that $\gamma_1+c_1$ or $\gamma_1'+c_1'$ are homologically nontrivial, then we say a logical error has occurred, and there is no entanglement in the post measured state. 

Throughout the above discussion, an important consequence of the localizable entanglement protocol is that one can defer the error correction procedure until after the single qubit measurements have been performed. In particular, rather than measure 1-form operators explicitly, one can perform all of the single qubit $X$ measurements first and take products of measurement outcomes to infer the eigenvalues of the 1-form operators $\tilde{G}$. One can classically process the measurement outcomes to identify the post measured state, as pointed out in Ref. \cite{RBH}. This gives perspective on why the localisable entanglement persists in the thermal RBH model without symmetry protection, as the measurement outcomes used to localize entanglement also provide the potential for error correction. 

A subtlety in this argument is the fact that the definition of $\rho_{\text{sym}}$ depends on the error correction protocol, which determines the choice of recovery map $U_{\{ s_{b}\}}$ in Eq. (\ref{eqRecovery}). In order to discuss the hidden SPT-order in an initial state $\rho_{0}(\beta)$, it is sensible to use the optimal quantum error correction protocol to construct the $1$-form symmetric ensemble $\rho_{\text{sym}}$. The question of finding a threshold for this optimal error correction can be rephrased as a problem of finding a phase transition in a certain statistical mechanical model, the random-plaquette $\zz_2$ gauge theory in three-dimensions \cite{dennisTopo}. The random-plaquette $\zz_2$ gauge theory undergoes a phase transition between an low-temperature ordered and a high-temperature disordered phase \cite{wang2003confinement}.   The ordered phase corresponds to the ability to successfully perform error correction with a high success probability in the RBH model at low temperature.  The threshold can be found at the critical point in the three-dimensional random-plaquette $\zz_2$ gauge theory along the Nishimori line (see Fig.~\ref{figGlass.jpg}). The critical point corresponds to a temperature of $T_{0} \approx 0.6$, which lower bounds the transition in localizable entanglement \cite{RBH}. It is thus natural to speculate that the thermal SPT-order in $\rho_0(\beta)$ persists up to $T_{0}$.

So far we have considered the thermal state of the RBH model both with and without 1-form symmetries enforced. A natural family of models which interpolates between these two cases is the symmetric Hamiltonian of Eq. (\ref{symHam}), with finite strength symmetry terms
\begin{align}\label{eqInterpolateHamiltonian}
H(\lambda) = H_{\calC} - \lambda \sum_{S_b\in \tilde{G}} S_{b}, \qquad \lambda\geq 0.
\end{align}
In the limit of $\lambda \rightarrow \infty$, the thermal state is the 1-form symmetric state, for which measurement of the operators $\tilde{G}$ always returns $+1$. As we have discussed, the related statistical model is the three-dimensional Ising gauge theory (or equivalently, the random-plaquette $\zz_2$ gauge theory with no randomness), which has a critical temperature at $T_{\infty} \approx 1.3$. Below this temperature we can always localize entanglement between distant boundaries. 

For intermediate values of $\lambda \in (0, \infty)$, excitations also have an additional energy cost at their boundaries (proportional to $\lambda$), as there is a finite energy penalty to violating the symmetry. Increasing $\lambda$ penalizes excitations which cannot be successfully corrected, leading to increased success rate if the same protocol is used. Finding the success rate appears difficult, as the corresponding statistical model is three-dimensional random-plaquette $\zz_2$ gauge theory, but with correlation between the plaquette random variables. If one neglects these correlations (which will be valid for small $\lambda$), then the transition temperature for finite $\lambda$ would be approximated by the line separating the order and disorder phases in the phase diagram of the three-dimensional random plaquette $\zz_{2}$ gauge theory (see Fig.~\ref{figGlass.jpg}). 

We remark that the protocol dependence of characterising topological order of thermal states is a generic challenge, both in the presence or absence of symmetries. It has been shown by Hastings, that the 4D toric code is topologically ordered at sufficiently small but finite temperature using the fact that quantum error correction protocol reliably works at low temperature \cite{HastingsThermal}. 

\begin{figure}[htb!]
	\centering
	\includegraphics[width=0.33\linewidth]{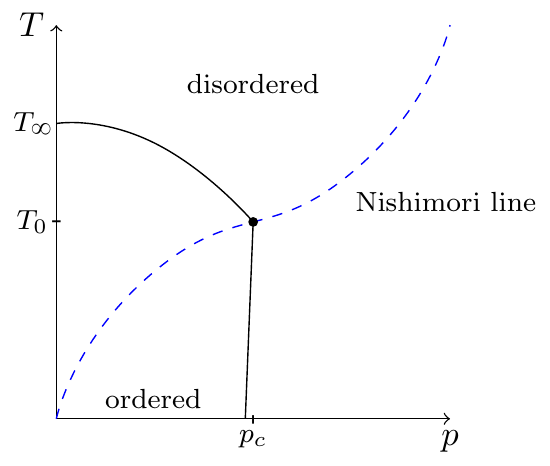}
	\caption{(Color online) Sketch of the phase diagram of the three-dimensional random-plaquette $\zz_2$ gauge theory~\cite{dennisTopo}. The random-plaquette $\zz_2$ gauge theory has $\pm1$ couplings, and the fraction of negative couplings is labelled $p$, the disorder strength. The disorder strength $p$ is on the horizontal axis, and temperature $T$ is on the vertical axis. The solid (black) line is the boundary between the ordered and disordered phase.  The dashed (blue) curve is the Nishimori line $e^{-2\beta}=p/(1-p)$.  The Nishimori point at $(p_c,T_0)$ lies at the intersection of the phase boundary and the Nishimori line. The critical temperatures of $H(\lambda)$ in the limiting case of $\lambda=0,\infty$ are depicted on the vertical axis. For intermediate values $\lambda \in (0, \infty)$, if correlation between plaquette random variables is ignored, the critical temperature is expected to interpolate between $T_{0}$ and $T_{\infty}$.
	} 
	\label{figGlass.jpg}
\end{figure}

\section{Outlook}\label{sec5}
Stability of thermal SPT-order provides a physical account for the surprisingly high error threshold attained in quantum computing architectures involving the 3D cluster state as well as a guiding principle to look for useful resource states for fault-tolerant quantum computation. Our work also opens new avenues for studies of higher-form SPT phases and their thermal properties with possible applications to quantum information processing as well as realizations of higher-form symmetries. Despite the theoretical beauty of higher-form SPT phases, the practical challenge was that physically realistic condensed matter systems do not naturally seem to possess higher-form symmetries. Our perspective on error correction in the 3D cluster state suggests that $1$-form symmetries can emerge from error correction even if we do not impose them as physical symmetries. This raises an intriguing possibility of realizing higher-form symmetries in an emergent manner through quantum error correction. With this perspective, one can ask whether the three-dimensional models of Refs.~\cite{fujii2012topologically,li2011thermal}, which have thermal states that are universal for MBQC, have underlying symmetries that give rise to SPT-order at nonzero temperature. In addition, our generalized definition of topological order at nonzero temperature in the presence of symmetries may be of independent interest as it provides insight into generalizing the Davies map formalism to simulate thermalization for quantum many-body systems with symmetries. This may be interesting in the context of symmetry-enriched topological phases, where the stability of a quantum memory may be enhanced by symmetry.

Thermal SPT phases are likely to find other applications in a broader context of fault-tolerant quantum computation. One particularly promising avenue is single-shot error correction~\cite{Bombin15}, which can significantly reduce the computational overhead in quantum computation.  Conventional error correction needs to take into account a possibility of faulty measurements, and thus repeated measurements are required to get reliable syndrome values. Single-shot error correction, where each syndrome is measured only once, is possible for topological stabilizer quantum codes which retain topological order at nonzero temperature~\cite{Bombin15}. While this observation relates thermal topological order to single-shot error correction, what remains as a puzzle is the fact that the 3D gauge color code~\cite{bombin2015gauge,kubica2015universal}, an example of a subsystem quantum code, also admits single-shot error correction. This fact strongly suggests that the gauge color code retains some sort of order at nonzero temperature, but such thermal order would appear to be in conflict with the thermal instability of topological order at nonzero temperature in all the known three-dimensional models~\cite{brownreview}. Our findings on thermal SPT-order hints that the 3D gauge color code may possess SPT-order protected by some set of symmetry operators that enable single-shot error correction. 

Our perspective of the nontrivial 1-form SPT model as a gapped domain wall described in section~\ref{sec3} raises an interesting question concerning topological defects associated with such a 3D domain wall. In a two-dimensional toric code, defects associated with the endpoints of a gapped domain wall can be viewed as Majorana fermions~\cite{BombinIsing}. This observation led to a huge body of work on characterizations of topological defects in two-dimensional topologically ordered systems~\cite{barkeshli2014symmetry,teo2015theory}. In our construction of a three-dimensional gapped domain wall associated with a nontrivial 1-form SPT model, its two-dimensional boundary may be viewed as some kind of topological defect in the 4D toric code. Characterization of such higher-dimensional defects and their thermal stability may be an interesting future question.  We note also that the thermal stability of Majorana fermions in nanowires is also of interest~\cite{Ped15,campbell2015decoherence} and our characterisation of thermal SPT stability may contribute to this work.

\begin{acknowledgments}
We acknowledge discussions with Andrew Doherty. BY and AK would like to thank the University of Sydney for their hospitality. SR and BY thank Robert Raussendorf for his hospitality during their visit to the University of British Columbia. SDB acknowledges support from the U.S. Army Research Office under Contract No. W911NF-14-1-0133, and the ARC via the Centre of Excellence in Engineered Quantum Systems (EQuS), project number CE110001013, and via project DP170103073. Research at Perimeter Institute is supported by the Government of Canada through Industry Canada and by the Province of Ontario through the Ministry of Research and Innovation. AK acknowledges funding provided by the Institute for Quantum Information and Matter, an NSF Physics Frontiers Center (NFS Grant PHY-1125565) with support of the Gordon and Betty Moore Foundation (GBMF-12500028).
\end{acknowledgments}

\end{document}